\documentclass[aos]{imsart}
\RequirePackage{amsthm,amsfonts,amssymb}
\RequirePackage[numbers,sort]{natbib}
\RequirePackage[colorlinks,citecolor=blue,urlcolor=blue]{hyperref}
\RequirePackage{graphicx}

\usepackage[reqno]{amsmath}

\usepackage{enumerate}% allows enumerate with latin characters
\usepackage{dsfont}% allows doublestroke symbols with \mathds{}
\usepackage[usenames,dvipsnames]{color}
\usepackage{mathrsfs}
\usepackage{placeins}
\usepackage{graphicx}
\usepackage{tikz}
\usetikzlibrary{arrows}
\usepackage{url}
\usepackage{comment}
\usepackage{multicol}
\usepackage{cleveref}
\usepackage{stmaryrd}
\usepackage{float}
\restylefloat{table}

\usepackage{caption}  
\usepackage{subcaption}
\usepackage{natbib}
\numberwithin{equation}{section}

\usepackage{aligned-overset}
\usepackage{bbm}
\usepackage{longtable} % For long tables
\usepackage{booktabs}
\crefname{assumptionletter}{Assumption}{Assumption}
\Crefname{assumptionletter}{Assumption}{Assumption}
\Crefname{assumptionHRV}{Assumption HRV}{Assumption HRV}
\usepackage{bibunits}

\defaultbibliographystyle{imsart-nameyear} 
\usepackage{etoolbox} %% Underfull HBox in .bbl
\apptocmd{\sloppy}{\hbadness 10000\relax}{}{}

\crefname{equation}{}{} %No equation when referencing equations
\Crefname{equation}{}{} %No equation when referencing equations
\theoremstyle{plain} %default (text in italic)
\newtheorem{theorem}{Theorem}[section]
\newtheorem{lemma}[theorem]{Lemma}
\newtheorem{proposition}[theorem]{Proposition}
\newtheorem{corollary}[theorem]{Corollary}
\newtheorem{assumptionletter}{{{Assumption}}}
\newtheorem*{assumptionHRV}{Assumption HRV}

\newcommand{\thistheoremname}{}
\newtheorem*{genericthm*}{\thistheoremname}
\newenvironment{namedthm*}[1]
  {\renewcommand{\thistheoremname}{#1}%
   \begin{genericthm*}}
  {\end{genericthm*}}

\theoremstyle{remark}
\newtheorem{definition}[theorem]{Definition}

\newtheorem{example}[theorem]{Example}

\newtheorem{remark}[theorem]{Remark}

\def\tablename{\footnotesize \textsc{Table}}

\newcommand{\bthe}{\begin{theorem}}
\newcommand{\ethe}{\end{theorem}}

\newcommand{\ben}{\begin{enumerate}}
\newcommand{\een}{\end{enumerate}}

\newcommand{\bit}{\begin{itemize}}
\newcommand{\eit}{\end{itemize}}

\newcommand{\beq}{\begin{equation}}
\newcommand{\eeq}{\end{equation}}

\newcommand{\ble}{\begin{lemma}}
\newcommand{\ele}{\end{lemma}}

\newcommand{\bde}{\begin{definition}\rm}
\newcommand{\ede}{\halmos\end{definition}}

\newcommand{\bco}{\begin{corollary}}
\newcommand{\eco}{\end{corollary}}

\newcommand{\bpr}{\begin{proposition}}
\newcommand{\epr}{\end{proposition}}

\newcommand{\brem}{\begin{remark}\rm}
\newcommand{\erem}{\end{remark}}

\newcommand{\bproof}{\begin{proof}}
\newcommand{\eproof}{\end{proof}}

\newcommand{\bexam}{\begin{example}\rm}
\newcommand{\eexam}{\end{example}}

\newcommand{\bfi}{\begin{fig}}
\newcommand{\efi}{\end{fig}}

\newcommand{\btab}{\begin{tab}}
\newcommand{\etab}{\end{tab}}

\newcommand{\beao}{\begin{eqnarray*}}
\newcommand{\eeao}{\end{eqnarray*}\noindent}

\newcommand{\balo}{\begin{align*}}
\newcommand{\ealo}{\end{align*}}

\newcommand{\balm}{\begin{align}}
\newcommand{\ealm}{\end{align}\noindent}

\newcommand{\beam}{\begin{eqnarray}}
\newcommand{\eeam}{\end{eqnarray}\noindent}

\newcommand{\barr}{\begin{array}}
\newcommand{\earr}{\end{array}}

\newcommand{\C}{\mathbb{C}}

\newcommand{\E}{\mathbb{E}}

\newcommand{\N}{\mathbb{N}}
\renewcommand\P{\mathbb{P}}
\newcommand{\Q}{\mathbb{Q}}
\newcommand{\R}{\mathbb{R}}

\newcommand{\LB}[1]{{\color{blue} #1}}

\def\bF{\mathbb{F}}
\def\bV{\mathbf{V}}
\def\bI{\mathbb{I}}

\def\bA{\boldsymbol A}
\def\bB{\underline{\boldsymbol  A}}

\def\bV{\boldsymbol V}
\def\bX{\boldsymbol X}
\def\bY{\boldsymbol Y}
\def\bZ{\boldsymbol Z}
\def\bF{\boldsymbol F}

\def\bp{\boldsymbol p}

\def\bT{\boldsymbol T}
\def\bcT{\boldsymbol{\mathcal{T}}}
\def\cT{\mathcal{T}}

\newcommand{\ubar}[1]{\mkern 1.5mu\underline{\mkern-1.5mu#1\mkern-1.5mu}\mkern 1.5mu}

\def\bpb{{\ubar{\widehat{\bp}}}}
\def\pb{\ubar{\widehat{p}}}
\def\rhob{\ubar{\widehat{\rho}}}

\def\bpbt{{\ubar{\widetilde{\bp}}}}
\def\pbt{\ubar{\widetilde{p}}}
\def\rhobt{\ubar{\widetilde{\rho}}}

\def\bx{\boldsymbol x}
\def\by{\boldsymbol y}
\def\bz{\boldsymbol z}

\def\bw{\boldsymbol w}
\def\bv{\boldsymbol v}
\def\bp{\boldsymbol p}
\def\bq{\boldsymbol q}
\def\bI{\boldsymbol I}
\def\bH{\boldsymbol H}
\def\bN{\boldsymbol N}

\def\bTheta{\boldsymbol \Theta}

\def\bSigma{\boldsymbol \Sigma}

\DeclareMathOperator*{\argmin}{arg\,min}

\newcommand{\vague}{\stackrel{\lower0.2ex\hbox{$\scriptscriptstyle
                    \it{v} $}}{\rightarrow}}
\newcommand{\weak}{\stackrel{\lower0.2ex\hbox{$\scriptscriptstyle
                    \it{w} $}}{\rightarrow}}
\newcommand{\what}{\stackrel{\lower0.2ex\hbox{$\scriptscriptstyle
                    \it{\hat{w}} $}}{\rightarrow}}
\newcommand{\eqdis}{\stackrel{\lower0.2ex\hbox{$\scriptscriptstyle
                    \mathrm{d}$}}{=}}
\newcommand{\distr}{\stackrel{\lower0.2ex\hbox{$\scriptscriptstyle
                    \it{d} $}}{\rightarrow}}

\newcommand{\Rd}{\mathbb{R}^d_+} 
\newcommand{\Sd}{\mathbb{S}^{d-1}_+}
 
\newcommand{\vb}{\, \vert \, \beta \, \vert \,} 
\newcommand{\vX}{\lVert \bX \rVert} 
\newcommand{\Pd}{\mathcal{P}_d^*} 
\newcommand{\ninf}{n \rightarrow \infty} 
\newcommand{\tinf}{t \rightarrow \infty}

\newcommand{\limn}{\lim_{n \rightarrow \infty}}

\newcommand{\Pconv}{\overset{\mathbb{P}}{\longrightarrow}}
\newcommand{\Dconv}{\overset{\mathcal{D}}{\longrightarrow}}
\newcommand{\di}{\, \mathrm{d}}

\newcommand{\td}{T_{n,2^d}'}

\newcommand{\diag}{\mathrm{diag}}

\newlength{\dhatheight}

\renewcommand{\hat}{\widehat}
\DeclareMathOperator*{\argmax}{arg\,max}

\DeclareMathOperator{\AIC}{AIC}
\DeclareMathOperator{\BIC}{BIC}
\DeclareMathOperator{\QAIC}{QAIC}

\DeclareMathOperator{\MSEIC}{MSEIC}
\DeclareMathOperator{\BICU}{BICU}
\DeclareMathOperator{\BICL}{BICL}
\DeclareMathOperator{\MSE}{MSE}
\DeclareMathOperator{\KL}{KL}
\DeclareMathOperator{\Bin}{Bin}
\DeclareMathOperator{\Mult}{Mult}
\allowdisplaybreaks %% display breaks for align
\definecolor{darkgreen}{RGB}{0,139,0}

\begin{document}
\begin{bibunit}

\begin{frontmatter}
\title{Information criteria for the number of \vspace*{0.2cm} \\ directions of extremes in high-dimensional data}
\runtitle{Information criteria for extreme directions}

\begin{aug}
  \author{\fnms{Lucas} \snm{Butsch}\ead[label=e1]{lucas.butsch@kit.edu} }%\orcid{0000-0002-xxxx-xxxx}}
    \and
   \author{\fnms{Vicky} \snm{Fasen-Hartmann}\ead[label=e2]{vicky.fasen@kit.edu}\orcid{0000-0002-5758-1999}}
%\thanksref{t1}\thankstext{t1}
 \address{Institute of Stochastics, Karlsruhe Institute of Technology \\[2mm] \printead[presep={ }]{e1,e2}}

%  \thanksref{T1}
%  \thankstext{T1}{}

  \runauthor{L. Butsch and  V. Fasen-Hartmann}
\end{aug}

\begin{abstract}

In multivariate extreme value analysis, the estimation of the dependence structure in extremes is  demanding, especially in the context of high-dimensional data. Therefore, a common approach is to reduce the model dimension by considering only the directions in which extreme values occur. In this paper, we use the concept of sparse regular variation  recently introduced by \citet{meyer_sparse} to derive information criteria for the number of directions in which extreme events occur, such as
a Bayesian information criterion (BIC), a mean-squared error-based information criterion (MSEIC), and a quasi-Akaike information criterion (QAIC) based on the Gaussian likelihood function. As is typical in extreme value analysis, a challenging task is the choice of the number $k_n$ of observations used for the estimation. Therefore, for all information criteria, we present a two-step procedure to estimate both the number of directions of extremes and an optimal choice of $k_n$. We prove that the
AIC of \citet{meyer_muscle23} and the MSEIC are inconsistent information criteria for the number of extreme directions whereas the BIC and the QAIC are consistent information criteria. 
Finally, the performance of the different information criteria is compared in a simulation study and applied on wind speed data. 

\end{abstract}

%\begin{keyword}[class=AMS]


\end{frontmatter}
%\maketitlehttps://arxiv.org/abs/math.PR/0000000

%==================================================================================================
\section{Introduction}

Multivariate extreme value statistics analyses the probabilities of joint extreme events in multivariate data with a wide range of applications, such as finance, insurance, meteorology, hydrology and, more generally, environmental risks due to the influence of climate change. This is a challenging task, especially for high-dimensional data, where modern research combines knowledge from extreme value theory with multivariate statistics and machine learning.

Multivariate regular variation is a classical concept for modeling multivariate extremes (\citet{resnick1987,resnick2007,Falk:Buch}). Suppose $\bX \in \Rd$ is a $d$-dimensional random vector and there exists an index $\alpha > 0$ (tail index) and a measure $S$ on the unit sphere $\Sd \coloneqq \{ \bx \in \Rd: \ \Vert \bx \Vert = 1 \}$ (spectral measure) such that
\begin{equation} \label{RVMeyer}
\P \left( \frac{\vX}{t} > r, \frac{\bX}{\vX} \in A \Big| \, \vX > t \right) \longrightarrow r^{-\alpha} S(A),\quad t \rightarrow \infty,
\end{equation}
for all $r > 0$ and all Borel sets $A \subset \Sd$ with $S( \partial A) = 0$, then $\bX$ is called \textit{multivariate regularly varying }of index $\alpha$. The spectral measure $S$ contains the information about the dependence structure in the extremes of $\bX$ and therefore a particular goal is the determination of $S$. However, in high-dimensional data sets where $d$ is large, this can be challenging and computationally intensive because the dependence structure in the extremes is usually complex. In the case of high dimensions, the spectral measure is often sparse and has support in a lower-dimensional subspace.
Therefore, a standard approach from multivariate statistics is to first apply a dimension reduction method to find the support of $S$ and then to estimate $S$, which drastically reduces the computational time and the quality of the estimation.

The literature on dimension reduction methods for multivariate extremes using statistical learning methods has grown rapidly in recent years. Starting with \citet{chautru}  who first applies a principal component analysis (PCA) and then a cluster analysis with spherical $k$-means to the spectral measure of a multivariate regularly varying random vector to find a group of variables that are jointly extreme. The reconstruction error of PCA is then analyzed in \citet{DS:21} and recently, \citet{Sabourin_et_al_2024} extend the PCA approach to Hibert-valued regularly varying random objects, whereas \citet{AMDS:22} use with kernel PCA a nonlinear generalization of PCA. In addition, \citet{CT:19,MR4582715} apply a PCA to the tail pairwise dependence matrix. The unsupervised learning approach of using spherical $k$-means, a variant of $k$-means, for cluster analysis in extreme observations was taken up in \citet{AMDS:24,Bernard_2013,JW:20,Fomichov:Ivanovs}. The topic of this paper is support identification of the spectral  measure, and the related literature  is 
\citet{damex,tawn,meyer_muscle23,pmlr-v139-jalalzai21a}.
A completely different line of research to represent the sparsity structure in multivariate models are graphical models as, e.g., \citet{Engelke:Hitz,Engelke:Volgushev,engelke2024,Gissibl_et_al,Gissibl_et_al:2018}, to name only a few.
A very nice overview of recent advances in probabilistic and statistical aspects of sparse structures in extremes is given in \citet{Engelke:Ivanovs}.

The support of $S$ can be identified by the disjoint partition  of the unit sphere
$\Sd$ into sets of the form
\begin{equation} \label{def:C_beta}
C_\beta \coloneqq \{\bx \in \Sd :  x_i > 0 \, \text{ for } \, i \in \beta, x_i = 0 \, \text{ for } \, i \notin \beta \} \subseteq \Sd, \quad \beta\subset \{1, \ldots, d\}.
\end{equation}
Knowing $S(C_\beta)$ for all $\beta\subseteq \{1, \ldots, d\}$ allows us to draw conclusions about the support of $S$ and the directions of the extremes. 
Of course, $S(C_\beta)>0$ implies that the components in the set $\beta$ are jointly extreme, we have an extreme event in the direction $\beta$.
However, the disjoint partition of $\Sd$ consists of $2^d-1$ sets so it is huge for large values of $d$, and estimating  $S(C_\beta)$ is non-trivial.
On the one hand, $C_\beta = \partial C_\beta$  and therefore the interior of $C_\beta$ is the empty set. 
As a consequence, if $S( C_\beta) > 0$ then the convergence in \eqref{RVMeyer} for $A=C_\beta$ does not necessarily hold. 
On the other hand, if  $\bX$ has a continuous distribution there are empirically no observations in the set $C_\beta$. Therefore, the empirical estimator for $S(C_\beta)$ based on \cref{RVMeyer} is not consistent and useful anymore. 
 To avoid this problem, the support detection algorithm DAMEX (Detecting Anomalies among Multivariate EXtremes) of  \citet{damex} works with truncated $\varepsilon$-cones to generate continuity sets that approximate the sets in \eqref{def:C_beta}, and \citet{tawn} use the concept of hidden regular variation   on a collection of nonstandard subcones
 of $[0,\infty]^d\backslash{\{0\}}$.

A completely different approach to mitigate this problem is proposed in \citet{meyer_sparse,meyer_muscle23} by introducing the concept of sparse regular variation, which is equivalent to regular variation under some mild assumptions (see \Cref{sec:preliminaries} for a definition). The main difference between regular variation and sparse regular variation is that the self-normalization \ $\bX / \vX$ in \eqref{RVMeyer} is replaced by the Euclidean projection $\pi (\bX/t)$ of $\bX/t$ for large $t > 0$, where the Euclidean projection $\pi : \, \Rd \rightarrow \mathbb{S}^{d-1}_+$ is defined as in  \citet{duchi} as $
 \pi(\bv)= \argmin_{\bw \in \R_+^d: \|w\|_1=1} \lVert \bw - \bv \rVert_2^2$. The advantage of this approach is that $\pi (\bX/t)$
 usually has more zero entries than \ $\bX / \vX$ and therefore, is more sparsely populated and advantageous when only a few components are extreme together, as in a high-dimensional setting. Since their empirical estimator for the number of extreme directions in the sparse regularly varying model is biased, indeed overestimates the true number of directions, they develop an Akaike Information Criterion (AIC) consisting of two steps. In the first step, they estimate the number of extreme directions by the AIC for \textit{bias selection}, but as usual, in extreme value theory, the estimation depends on the chosen threshold that goes into the estimation; the observations above this threshold determine the extreme observations. Therefore, they extend the AIC for \textit{bias selection} to an AIC for \textit{threshold selection}, where the threshold is also estimated. What is really special is that they were able to develop a method to estimate the number of extreme directions and the threshold at the same time, both of which are very challenging tasks on their own. But as we prove in \Cref{th:AIC_Cons}, the AIC for bias selection is not a weakly consistent information criterion, as is often the case for Akaike's information criteria, and so we develop alternatives. Consistency is examined only for \textit{bias selection} and not for \textit{threshold selection}, because there is no "true" threshold. Here, we have the well-known bias-variance tradeoff: If the threshold is chosen too high, there are not enough extreme observations leading to a high variance, and if it is too low, non-extreme observations lead to a bias in the estimation.

%%%%%%%%%%%%%%%%%%%%%%%%%%%%%%    §4    %%%%%%%%%%%%%%%%%%%%%%%%%%%%%%
In this paper we use the approach of \citet{meyer_muscle23} of sparse regular variation and propose three different information criteria to estimate the number of extreme directions and the choice of the threshold, the BIC, QAIC and MSEIC for \textit{bias selection} and \textit{threshold selection}, which are particularly suitable for high dimensional data with a sparsity structure in the extreme behavior. Thus, we develop procedures to estimate the number of extreme directions and the optimal choice of the threshold at the same time. The application of these information criteria is very simple in practice and not computationally intensive.  
Besides the AIC, the Bayesian Information Criterion (BIC), which goes back to \citet{schwarz}, is the most popular in practice and tries to select the model with the highest posterior probability. The statistical model behind our BIC is the same as that of the AIC in \cite{meyer_muscle23}, where we fit a multinomial model to the number of extreme observations in the subspaces $C_{\beta}$ and derive an asymptotic upper bound on the posterior likelihood, which then defines the BIC. 
In contrast, the QAIC for Quasi-Akaike Information Criterion approximates the Kullback-Leibler divergence of the true model and a Gaussian model, rather than a multinomial model as used in the AIC and BIC, respectively. The advantage of BIC and QAIC over AIC is that they are consistent information criteria for \textit{bias selection}. Finally, the third method, MSEIC, stands for mean-squared error information criteria, because we approximate the mean-squared error (MSE) of the relative number of extreme observations and the true probabilities of extremes in the different subspaces $C_\beta$. Although MSEIC is not consistent for bias selection, it performs extremely well in all simulations.

\subsection*{Structure of the paper} The paper is organized as follows. In \Cref{sec:preliminaries} we properly define extreme directions based on the concept of sparse regular variation and introduce consistent and asymptotically normally distributed estimators for the probabilities of the extreme directions as in \citet{meyer_muscle23}. We also present statistical models for some of our information criteria. 
 The main results of the paper are derived in \Cref{sec:BIC,sec:QAIC,sec:MSEIC}. In \Cref{sec:QAIC}, we first introduce the $\QAIC$ for bias selection and threshold selection following the framework of Akaike information criteria, which aims to minimize the expected Kullback-Leibler (KL) divergence, here applied to a Gaussian likelihood function.  We prove that, unlike the $\AIC$ proposed by \citet{meyer_muscle23},  the $\QAIC$ for bias selection is a consistent information criterion.  
In \Cref{sec:MSEIC}, we develop the MSEIC and finally, in \Cref{sec:BIC}, the BIC for both bias selection and threshold selection. 
In addition, we demonstrate in these sections that the BIC is a consistent information criterion for bias selection, whereas the MSEIC is not consistent. Moreover,  we compare all information criteria in a simulation study in \Cref{sec:NumExp} and apply them to extreme wind data from the Republic of Ireland in \Cref{sec:Application}. Finally, we draw some conclusions in \Cref{sec:conclusion}.
The main proofs of the paper are moved to the appendix, while the proofs of some auxiliary results can be found in the supplementary material.  

\subsection*{Notation} In this paper, we use the following notation.  
 For a vector $\bx=(x_1,\ldots,x_d)^\top \in \R^d$ and a set $I \subset \{1, \ldots, d \}$ we write $\bx_I\in\R^{|I|}$ for $(x_i)_{i\in I}$ and $\diag(\bx)\in\R^{d\times d}$ for a diagonal matrix with the components of $\bx$ on the diagonal. Furthermore, $\bI_d\in\R^{d\times d}$ is  the  identity matrix, $\mathbf{0}_d \coloneqq (0, \ldots, 0)^\top \in \R^d$ is the zero vector  and $\mathbf{1}_d: = (1, \ldots, 1)^\top \in \R^d$ is the vector  containing only $1$. Moreover, $\Vert \bx \Vert \coloneqq \Vert \bx \Vert_1$ is the $L_1$-norm and $\Vert \bx \Vert_2$ is the Euclidean norm for $\bx\in\R^d$. The unit sphere $\Sd=\{\bx\in\left[0,\infty\right)^d:\,x_1+\cdots+x_d=1\}$ is defined with respect to the $L_1$-norm. 
  For $a \in \R$, $\bx,\by \in \R^d$ operations as $\bx^a, \sqrt{\bx}$ and $\bx\cdot \by$ are meant component-wise. The gradient of a function $f: \R^d \mapsto \R^k$ is written as $\nabla f(\bx) \in \R^{k \times d}$ for $\bx\in\R^d$ and the partial derivative with respect to the $i$-the component $x_i$ of  $\bx=(x_1,\ldots,x_d)^\top$ is $\frac{\partial}{\partial x_{i}} f(\bx)$. By $| a |$ we denote the absolute value of a real number $a$ and by  $|A|$ the cardinality of a set $A$, but the meaning should be clear from the context. 
In addition, $\mathcal{P}_d$ is the power set of the set $\{1,\ldots,d\}$ and $\Pd \coloneqq \mathcal{P}_d \setminus \emptyset$.
  Finally,  $\Dconv$ is the notation for convergence in distribution and  $\Pconv$ is the notation for convergence in probability.

\section{Preliminaries} \label{sec:preliminaries}
This section addresses the main concepts of the paper which are based on \citet{meyer_sparse,meyer_muscle23}. 
We start with an introduction into sparse regular variation and then derive a proper definition of \textit{extreme direction} in \Cref{sec:Extreme:direction}. The challenging task in the statistical inference of extreme directions is the detection of the \textit{bias directions} which are rigorously defined and motivated in \Cref{sec:bias direction}.
Then, in \Cref{sec:2.3}, we give an overview on the statistical inference of the empirical estimator of the probabilities of extreme directions and the assumptions of the present paper. Finally, in \Cref{sec:Statistical:Models},  we present statistical models on which the information criteria are based.

\subsection{Sparse regular variation and extreme directions} \label{sec:Extreme:direction}

First, we introduce the concept of sparse regular variation with the Euclidean projection $\pi :\,  \Rd \rightarrow \mathbb{S}^{d-1}_+$ defined  as $
 \pi(\bv)= \argmin_{\bw \in \R_+^d: \|w\|_1=1} \lVert \bw - \bv \rVert_2^2$.

\begin{definition}\label{def:srv}
An $\mathbb{R}_+^d$-valued random vector $\bX$ is called \textit{sparse regular varying}, if a $\mathbb{S}_+^{d-1}$-valued random vector $\bZ$   and a non degenerate random variable $R$ exist such that
\begin{equation*}
\P \left(  \frac{ \vX}{t} > r , \pi \left(\frac{\bX}{ t } \right) \in A \; \middle  \vert  \;   \vX > t \right) \rightarrow \P ( R > r, \bZ \in A), \quad \tinf,
\end{equation*}
for all $r > 0$ and all Borel sets $A \subset \Sd$ with  $\P(  \bZ \in \partial A) = 0$.
\end{definition}

\begin{remark}
\begin{itemize}
\item[(a)] Note that $R$ is Pareto$(\alpha)$-distributed for an $\alpha > 0$ and models the radial part, whereas the  $\Sd$-valued random vector $\bZ$ corresponds to the angular part. Therefore, we write briefly $\bX \in \text{SRV}(\alpha,\bZ)$.
\item[(b)] The concept of sparse regular variation introduced by \citet{meyer_sparse} is currently limited to random vectors in the positive orthant. A corresponding theory for $\mathbb{R}^d$-valued random vectors has not yet been developed. Consequently, in this paper, we also restrict our analysis to random vectors in the positive orthant, which aligns with ours and many other applications. 
\end{itemize}
\end{remark}

A proper definition of extreme direction is now the following, where we use the notation that $\mathcal{P}_d$ is the power set of the set $\{1,\ldots,d\}$ and $\Pd \coloneqq \mathcal{P}_d \setminus \emptyset$.

\begin{definition}
 A direction $\beta \in \Pd$ is  an \textit{extreme direction}, if $\P(\bZ \in  C_\beta) > 0.$
The set of all extreme directions is denoted as
\begin{equation*} 
\mathcal{S} ( \bZ) \coloneqq \{ \beta \in \Pd : \P( \bZ \in C_\beta ) > 0 \} \quad \mathrm{with} \quad s^* \coloneqq | \mathcal{S} ( \bZ) |.
\end{equation*}
\end{definition}

\newpage

\begin{remark}~
\begin{enumerate}[(a)]
 \item   The use of the $L_1$-projection leads to a sparse representation, in the sense that under $\pi$ more components are projected to zero compared to the normalization $\bv \mapsto \bv / \Vert \bv \Vert$. 
 Therefore, it is not surprising that according to  \citet[Theorem 2]{meyer_sparse},  $S(C_{\beta}) > 0$ implies $\P( \bZ \in C_{\beta}) > 0$ for $\beta  \in \mathcal{P}^*_d$. Thus, an extreme direction under regular variation is as well an extreme direction under sparse regular variation but the opposite does not necessarily hold. However, the maximal directions 
 under regular variation and sparse regular variation are equivalent, such that we do not lose much information on the support of $S$ under sparse regular variation. 
    Note that a direction $\beta \in \mathcal{P}^*_d$ is called a maximal direction of the regularly varying random vector $\bX$ if $\P(\bTheta \in C_\beta) > 0$ and $\P(\bTheta \in C_{\beta'}) = 0$ for all $\beta \subset \beta' \in \mathcal{P}^*_d$.
In the case of sparse regular variation, the definition of a maximal direction is analogous, except that the random vector $\bTheta$ is replaced by $\bZ$.
\item  Since the preimages $\pi^{-1}(C_\beta)$ are sets with positive Lebesgue measure, 
the sets $C_\beta$ are continuity sets of $\P(\bZ \in \cdot )$. Finally, from
 \citet[Proposition 2]{meyer_sparse} we know that
 $$\P ( \pi(\bX/t) \in C_{\beta} \, \vert \, \lVert \bX \rVert > t ) \longrightarrow \P ( \bZ \in C_{\beta} ), \quad  \text{ as } \tinf,$$
so that $\P ( \bZ \in C_{\beta} )$ can be estimated empirically in contrast to $S(C_{\beta} )$.

\end{enumerate}
\end{remark}

The aim of the paper is to estimate $s^*$, the number of extreme directions under sparse regular variation, through the use of information criteria.

\subsection{Bias directions} \label{sec:bias direction} 
A major challenge for the estimation of the extreme directions is that the empirical estimators of the probabilities $\P(\bZ\in C_\beta)$, $\beta  \in \mathcal{P}^*_d$, detect more extremal directions than there are true extremal directions, which we call \textit{bias directions}. To understand the idea of bias directions better we require some further notation.
Suppose $\Vert \bX_{(1,n)}\Vert\geq \cdots\geq \Vert \bX_{(n,n)}\Vert $ is the order statistic of $\Vert \bX_1\Vert, \ldots, \Vert \bX_n\Vert $ and
 the number of extreme observations used for the estimations is denoted by $k_n \in \N$, whereas we assume that $k_n\to\infty$ as $n\to\infty$. Suppose that there exists a sequence of high thresholds  $u_n > 0 $ for $ n \in \N$ such that $k_n / n \sim \P ( \lVert \bX \rVert > u_n) $  and $u_n \rightarrow \infty$ as $ n \rightarrow \infty$.  
Due to  \citet[Proposition 1]{meyer_muscle23} the empirical estimator
\begin{equation*} \label{eq:def_Tn}
\frac{T_n(C_\beta, k_n)}{k_n}  \coloneqq \frac{1}{k_n} \sum_{j=1}^n \mathbbm{1}\left\{ \pi(\bX_j/ \Vert \bX_{(k_n+1,n)} \Vert) \in C_\beta , \lVert \bX_j \rVert > \Vert \bX_{(k_n+1,n)} \Vert \right\},
\end{equation*}
 of the probability 
\begin{equation} \label{prob p}
p(C_\beta) \coloneqq \P ( \bZ \in C_\beta) = \limn 
\P( \pi(\bX/u_n) \in C_\beta \, \vert \, \, \lVert \bX \rVert > u_n)
\end{equation}
is a consistent estimator, so that the empirical observed set  of extreme directions is 
\begin{equation*}
\widehat{\mathcal{S}}_n(\bZ) \coloneqq \{ \beta \in \Pd:  T_n(C_{\beta}, k_n)  > 0 \}. 
\end{equation*}
To be able to relate the true set of extreme directions
$\mathcal{S} ( \bZ)$
with the empirically estimated set of extreme directions, we define the set
\begin{eqnarray*}
\mathcal{R} &\coloneqq & \{ \beta \in \Pd: \lim_{n\to\infty}k_n p_n(C_\beta) = \infty  \}
\qquad \text{ and } \qquad r \coloneqq | \mathcal{R}|,
\end{eqnarray*}  where $\mathcal{R} $ depends on the chosen sequence $(k_n)_{n\in\N}$, which we neglect for the ease of notation, and 
\begin{eqnarray*}
p_n(C_\beta)&\coloneqq&\P( \pi(\bX/u_n) \in C_\beta \, \vert \, \, \lVert \bX \rVert > u_n).
\end{eqnarray*}
Of course, $\beta\in \mathcal{S}(\bZ)$ implies $k_np_n(C_\beta)\to\infty$  such that trivially,  $\mathcal{S}(\bZ)\subseteq \mathcal{R}$ and \linebreak $s^*\leq r$. 
Under the  Assumption HRV, a shorthand for hidden regular variation,  we can  say more  about the relations of these sets. 
\begin{assumptionHRV}\label{asu:HRV} 
    For every $\beta \in \Pd$ we define the cone
    \begin{eqnarray*}
    \C_\beta\coloneqq\left\{\bx=(x_1,\ldots,x_d)^\top\in\R_+^d:\sum_{j\in\beta}(x_j-\max_{i\in\beta^c} x_i )\geq 0\right\}\subseteq \R_+^d
\end{eqnarray*}
    and suppose that the random vector $\bX$ is multivariate regular varying on $\Rd \setminus \C_\beta$  with tail index $\alpha(\beta)$ and exponent measure $\mu_\beta$ satisfying  
    \begin{equation*}
        \mu_\beta\left(\left\{ \bx=(x_1,\ldots,x_d)^\top \in \Rd : \max_{i \in \beta} x_{i} < 1, \min_{i \in \beta^c} x_{i} \ge 1\right\} \right) > 0. 
    \end{equation*}
\end{assumptionHRV}
A conclusion from \citet[Proposition 2]{meyer_muscle23} is then that  under Assumption HRV even 
\begin{equation} \label{sungleichung}
\lim_{n \rightarrow \infty} \P(\mathcal{S}(\bZ)\subseteq \mathcal{R}\subseteq 
\widehat{\mathcal{S}}_n(\bZ))=1
\end{equation}
holds. Thus, the empirical estimator tends to overestimate the set of extreme directions (but does not underestimate it asymptotically). On the one hand, 
for $n$ large and $\beta \in \Pd$ with $T_n(C_\beta,k_n)=0$
this means that $\beta$ is not an extreme direction. But on the other hand,  for $n$ large there might be a $\beta \in \Pd$ with $T_n(C_\beta,k_n)>0$ which is not an extreme direction; a mathematical more rigorous interpretation is given in \citet{meyer_muscle23}. Such a direction is referred to as a \textit{ bias direction}. The main challenge is to identify these bias directions. 

\begin{remark} \label{Remark:r}
There exists as well a stronger statement than \eqref{sungleichung}. Suppose additionally that $\lim_{n\to\infty}k_n p_n(\beta)= 0$ for all $\beta \in \Pd \setminus \mathcal{R}$. A conclusion of  
\citet[Lemma 1]{meyer_muscle23} is then that $\lim_{n\to\infty}\mathbb{P}( T_n(C_\beta,k_n) = 0)=1$ 
for all $\beta \in \Pd \setminus \mathcal{R}$ and  hence, 
\begin{equation*} 
\lim_{n \rightarrow \infty} \P(\mathcal{S}(\bZ)\subseteq \mathcal{R}=
\widehat{\mathcal{S}}_n(\bZ))=1.
\end{equation*}
In particular, this means that $\widehat r_n:=\vert\mathcal{S}_n(\bZ) \vert\Pconv r$ as $n\to\infty$.
\end{remark}

\subsection{Statistical inference for the probabilities of extreme directions} \label{sec:2.3}

The general assumptions of the present paper are motivated by the statistical inference of the probabilities of extreme directions as derived in \citet{meyer_muscle23}. To understand the statistical inference and hence, the assumptions,
we have to enumerate the $\beta \in \Pd$ in the following way with $p(C_{\beta})$ as defined in \eqref{prob p}:
\begin{align*}
\beta_1 \coloneqq{}& \argmax_{\beta \in \mathcal{P}^*_d} p(C_{\beta}), \\
\beta_2 \coloneqq{}& \argmax_{\beta \in \mathcal{P}^*_d \setminus \{\beta_1\}} p(C_{\beta}),\\
\vdots\,\,& \quad \; \;  \\
\beta_{s^*} \coloneqq{}& \argmax_{\beta \in \mathcal{P}^*_d \setminus \{\beta_1,\ldots, \beta_{s^*-1}\}} p(C_{\beta}),
\end{align*}
where the remaining $\beta_{s^*+1},\ldots, \beta_{2^d-1}$ with $p(C_{\beta_j}) = 0,\, j = s^*+1,\ldots,2^d-1$, are ordered in an arbitrary but fixed order such that  $\beta_{j} \in \mathcal{R}$  for $j = s^*+1 , \ldots, r$. 
 We write briefly for $j = 1,\ldots,2^d-1$,
 \begin{align*}
 p_{j} \coloneqq{} &p(C_{\beta_j}), &p_{n,j} \coloneqq {}& p_n(C_{\beta_j})  \coloneqq \P( \pi(\bX/u_n) \in C_{\beta_j} \, \vert \, \, \lVert \bX \rVert > u_n), \\ \cT_{n,j} \coloneqq{}& \cT_n(C_{\beta_j}), &T_{n,j}(k_n) \coloneqq{}&  T_n(C_{\beta_j}, k_n), 
 \end{align*}
where 
 \begin{equation*} 
\frac{\mathcal{T}_n(C_\beta)}{k_n}  \coloneqq \frac{1}{k_n} \sum_{j=1}^n \mathbbm{1}\{ \pi(\bX_j/ u_n) \in C_\beta , \lVert \bX_j \rVert > u_n \}.
\end{equation*}
Finally, we define the associated vectors
\begin{align*}
\bp \coloneqq{}& (p_1 , \ldots, p_{r})^\top , &\bp_n &\coloneqq  (p_{n,1} , \ldots, p_{n,r})^\top, \\
\bcT_n \coloneqq{}&  (\mathcal{T}_{n,1} , \ldots, \mathcal{T}_{n,r})^\top, 
&\bT_n(k_n) &\coloneqq (T_{n,1}(k_n) , \ldots, T_{n,r}(k_n))^\top. 
\end{align*}

In the next theorem, we summarize the asymptotic behavior of these estimators as derived in \citet[Theorem 1 and Proposition 3]{meyer_muscle23}. 

\begin{proposition}\label{th:Theorem1_MW}
    Suppose Assumption HRV holds 
    and the sequence $(k_n)_{n\in\N}$ in $\N$ with $k_n\to\infty$ and $k_n/n\to0$ satisfies  $ \mathcal{R} = \widehat{\mathcal{S}}_n(\bZ)$  almost surely for all $n$ large enough.
    Furthermore, assume that for some $\tau > 0$  and any $j = 1, \ldots, r$ as $\ninf$,
        \begin{equation*}
            \sup_{r \in [ \frac{1}{1 + \tau}, 1 + \tau]} \sqrt{\frac{k_n}{p_{n,j}}} \Bigl| \frac{n}{k_n} \P( \bX / u_n \in  \{\bx \in \Rd :  r \Vert \bx \Vert > 1, \pi( r\bx ) \in C_{\beta_j}  \} ) - r^{\alpha(\beta_j)} p_{n,j} \Bigr| \rightarrow 0.
        \end{equation*}
    \begin{enumerate}[(a)]
        \item Then, as $\ninf$,
        \begin{equation*}
            \sqrt{k_n} \diag( \bp_{n})^{-1/2} \Big( \frac{{\bcT}_{n}}{k_n} - \bp_{n} \Big)  \Dconv \mathcal{N}_{{r}}(\mathbf{0}_{r}, \bI_{r}).
        \end{equation*}
        \item If additionally $\sqrt{k_n} ( p_{n,j} - p_j) \rightarrow 0$ as $\ninf$ and $j = 1, \ldots, {r}$, then  as $\ninf$,
        \begin{align*}
            \sqrt{k_n} \diag( & \bp_{n})^{-1/2}\left( \frac{\bT_{n}(k_n)}{k_n} - \bp_{n} \right)\Dconv \Big(\bI_{r} - \sqrt{\bp} \cdot \sqrt{\bp}^\top \Big) \mathcal{N}_{{r}}(\mathbf{0}_{r}, \bI_{r}).
        \end{align*}
    \end{enumerate}
\end{proposition}

Motivated by this result we define for any $n\in\N$   
\begin{eqnarray*}
    \bp_n^*:=(p_{n,1},\ldots,p_{n,s^*},\rho_n,\ldots,\rho_n)^\top\in\R^r \quad \text{ with } \quad  \rho_n \coloneqq \frac{1}{r-s^*}\sum_{j=s^*+1}^{r}p_{n,j}
\end{eqnarray*}
and suppose the following assumption throughout the paper.

\begin{assumptionletter}~ \label{asu:directions} \label{Assumption:main} 
\begin{enumerate}[({A}1)] 
\item  \label{(A1)} Suppose $(k_n)_{n\in\N}$ is a sequence in $\N$ with $k_n\to\infty$ and $k_n/n\to0$. Furthermore $ \mathcal{R} = \widehat{\mathcal{S}}_n(\bZ)$ almost surely for all $n$ large enough, which implies \linebreak  $r=|\mathcal{R}|=|\widehat{\mathcal{S}}_n(\bZ)|\geq s^*$ almost surely for all $n$ large enough.
\item \label{(A2)} $T_{n,1}(k_n)\geq T_{n,2}(k_n)\geq \cdots\geq T_{n,r}(k_n)$ almost surely for all $n$ large enough.
\item \label{(A4)} Suppose that as $n\to\infty$,
\begin{align*}
            \sqrt{k_n} \diag( & \bp_{n}^*)^{-1/2}\left( \frac{\bT_{n}(k_n)}{k_n} - \bp_{n}^* \right)\Dconv \Big(\bI_{r} - \sqrt{\bp} \cdot \sqrt{\bp}^\top \Big) \mathcal{N}_{{r}}(\mathbf{0}_{r}, \bI_{r}).
        \end{align*}
\item \label{(A5)} Suppose that as $n\to\infty$,
\begin{align*}
            \sqrt{k_n} \diag( & \bp_{n}^*)^{-1/2}\left( \frac{\bcT_{n} }{k_n} - \bp_{n}^* \right)\Dconv   \mathcal{N}_{{r}}(\mathbf{0}_{r}, \bI_{r}).
        \end{align*}
\end{enumerate}
\end{assumptionletter}

\begin{remark}~
\begin{enumerate}[(a)]
    \item A justification of Assumption (\ref{asu:directions}\ref{(A1)})  
    is given in \Cref{Remark:r}, where a sufficient
    criterion for 
    $\lim_{n \rightarrow \infty} \P(\mathcal{R}= 
\widehat{\mathcal{S}}_n(\bZ))=1$ is stated.
Assumption (\ref{asu:directions}\ref{(A1)}) is particularly useful for modelling purposes, as can be seen in the derivation of the $\AIC$ in \citet{meyer_muscle23}, and from other statements in that paper such as \Cref{th:Theorem1_MW} above. If Assumption (\ref{asu:directions}\ref{(A1)}) is not made, then  the consistency results in this paper can be obtained by replacing $r$ with $\widehat r_n:=|\mathcal{S}_n(\bZ)|$ and assuming $\sqrt{k_n\rho_n}(\widehat r_n-r)\Pconv 0$ (cf. \Cref{remark:AIC:r}
and \Cref{remark:QAIC:r}).
    \item Assumption (\ref{asu:directions}\ref{(A2)}) is motivated by the fact that we have $\bT_n(k_n) / k_n \Pconv \bp$ and thus, for $n$ sufficiently large $\bT_n(k_n)$ is ordered by size with probability close to $1$ because $\bp$ is ordered by size.
    \item The assumptions (\ref{asu:directions}\ref{(A4)}) and (\ref{asu:directions}\ref{(A5)}) are not strong, in the case $p_{n,s^*+1}=\ldots=p_{n,r}=\rho_n$, \Cref{th:Theorem1_MW} gives a sufficient criteria for (\ref{asu:directions}\ref{(A4)}) or (\ref{asu:directions}\ref{(A5)})  to hold. 
\end{enumerate}
\end{remark}

The following lemma is a direct consequence of \Cref{Assumption:main}.

\begin{lemma} \label{cor:Theorem1_MW}
Suppose \Cref{asu:directions,Assumption:main} holds. Then the following statements are valid.
\begin{itemize}
    \item[(a)] $\rho_n\to 0$ and $\rho_nk_n\to\infty$ as $n\to\infty$.
    \item[(b)] For $j =  1, \ldots, s^*$ and $n\to\infty$, 
    $$ \frac{T_{n,j}(k_n)}{k_n p_{n,j}} \Pconv 1 \quad \text{ and } \quad 
     \frac{\cT_{n,j}}{k_n p_{n,j}} \Pconv 1.$$ 
    \item[(c)] For $j=s^*+1,\ldots,r$ and $n\to\infty$, 
    $$\frac{T_{n,j}(k_n)}{k_n \rho_{n}} \Pconv 1 
    \quad \text{
    and } \quad \frac{T_{n,j}(k_n)}{k_n } \Pconv 0,$$
    and similarly,
    $$\frac{\cT_{n,j}}{k_n \rho_{n}} \Pconv 1 
    \quad \text{
    and } \quad \frac{\cT_{n,j}}{k_n } \Pconv 0.$$   
\end{itemize}
\end{lemma}

\subsection{Statistical models} \label{sec:Statistical:Models}

 A challenging task in extreme value theory is the optimal choice of $k_n$, the number of extreme observations used for the estimation procedure. Therefore, we follow a two-step procedure as motivated in \citet{meyer_muscle23}. In the first step, we fix $k_n$ and estimate the relevant extreme directions $\beta \in \mathcal{S}(\bZ)$  and separate them from the so-called bias directions $\beta \in \mathcal{\widehat{S}}_n(\bZ) \setminus \mathcal{S}(\bZ)$ using some information criteria. Therefore this step is called \textit{bias selection}. In the second step, we estimate the threshold $k_n$, this step is therefore named \textit{threshold selection}. In the following subsections, we present some statistical models for the \textit{bias selection} in \Cref{sec:localModel}
and  the statistical models for the \textit{threshold selection} in \Cref{sec:globalModel}.

\subsubsection{The local model for the bias selection} \label{sec:localModel}

Due to Assumption (\ref{asu:directions}\ref{(A1)})  
with $r=|\widehat{\mathcal{S}}_n(\bZ)|$ 
the random vector $\bT_n(k_n)$ is multinomial distributed with $k_n$ repetitions and unknown $r$-dimensional probability vector $\bp_{n,k_n}$ 
which converges as $n\to\infty$ to $\bp$.
To detect the bias directions and hence, to estimate $s^*$, the idea is now to fit  for any $s\in\{1,\ldots,r\}$ a  multinomial distribution from the class 
$\{\Mult(k_n,\bA_s(\widetilde{\bp}^s)):\,\widetilde{\bp}^s\in  \Theta_s \}$ where $\bA_s:\R^{s}\to\R^r$ is defined as $$\bA_s(\widetilde{\bp}^s)= \left(\widetilde p_1^s,\ldots,\widetilde p_s^s, \frac{1 - \sum_{j=1}^s \widetilde p_j^s}{r-s},\ldots,\frac{1 - \sum_{j=1}^s \widetilde p_j^s}{r-s}\right)^\top$$ and the parameter space $\Theta_s$ is defined as
\begin{equation*}
    \Theta_s \coloneqq \left\{\widetilde{\bp}^s = (\widetilde{p}_1^s, \ldots, \widetilde{p}_s^s) \in (0,1)^{s} :\,\widetilde p_1^s\geq\cdots\geq \widetilde p_s^s,\, \sum_{j=1}^s \widetilde{p}_j^s  < 1\right\},
\end{equation*}
which reflects that there are $r-s$ bias directions. Finally, we define $$ \widetilde{\rho}^s \coloneqq \frac{1 - \sum_{j=1}^s \widetilde{p}_j^s}{r-s}\in(0,1) \quad \text{ for }\,\widetilde{\bp}^s \in \Theta_s.$$  We summarize this in the following model.

\medskip
\textsc{MODEL} $M^s_{k_n}$: \hspace{0.2cm}
\textit{The family of multinomial distributions $\{\Mult(k_n,\bA_s(\widetilde{\bp}^s ) ) :\,\widetilde{\bp}^s\in  \Theta_s \}$ 
with likelihood function 
\begin{align*} 
L_{M^s_{k_n}}(\widetilde{\bp}^s\, \vert \, \bT_n(k_n)) &= \frac{k_n!}{\prod_{j=1}^{r} T_{n,j}(k_n)!} \prod_{j=1}^{s} (\widetilde{p}^s_{j})^{T_{n,j}(k_n)} \prod_{j=s+1}^{r} (\widetilde{\rho}^s)^{T_{n,j}(k_n)}
\end{align*}
and 
log-likelihood function 
\begin{align} \label{eq:logLikelihood}
\log L_{M^s_{k_n}}  (\widetilde{\bp}^s\, \vert \, \bT_n(k_n)) &= \log (k_n!) -   \sum_{j=1}^{r} \log(T_{n,j}(k_n)!) + \sum_{j=1}^{s} T_{n,j}(k_n) \log(\widetilde{p}^s_{j}) \nonumber \\
&\quad  + \log(\widetilde{\rho}^s)  \sum_{j=s+1}^{r} T_{n,j}(k_n)
\end{align}
is called Model $M^s_{k_n}$.} \\

Now, an information criterion aims to find the Model $M_{k_n}^s$ from $s\in\{1,\ldots,r\}$ which best fits the distribution of $\bT_n(k_n)$ and results in an estimator $\widehat s_n$ for $s^*$. 
Then, for a given estimator $\widehat{s}_n$ of $s^*$ we estimate the probability vector $\bp$ by 
\begin{equation} \label{eq:estimator_hellinger}
    \widehat{\bp}_{n,*}^{\, \widehat{s}_n} \coloneqq \bigg( \frac{\widehat{p}_{n,1}^{\, \widehat{s}_n}}{ \sum_{j=1}^{\widehat{s}_n} \widehat{p}_{n,j}^{\, \widehat{s}_n}} , \ldots, \frac{\widehat{p}_{n,{\widehat{s}_n}}^{\, \widehat{s}_n}}{ \sum_{j=1}^{\widehat{s}_n} \widehat{p}_{n,j}^{\, \widehat{s}_n}}, 0 , \ldots, 0 \bigg)^\top,
\end{equation}
where
\begin{eqnarray} \label{MLE:mult}
     \widehat{\bp}_n^s \coloneqq ( \widehat{p}_{n,1}^s, \ldots, \widehat{p}_{n,s}^s)^\top \coloneqq \Big(\frac{T_{n,1}(k_n)}{k_n}, \ldots, \frac{T_{n,s}(k_n)}{k_n}  \Big)^\top  
\end{eqnarray}
is the maximum likelihood estimator (MLE) of the multinomial model $M_{k_n}^s$ (see \citet{meyer_muscle23}, Section 4.1).
 Finally, we define 
\begin{eqnarray*}
    \widehat{\rho}_n^s \coloneqq\frac{1}{r-s}\Big(1-\sum_{j=1}^s\widehat{p}_{n,j}^s\Big) =\frac{\sum_{j=s+1}^{r} T_{n,j}(k_n)}{(r-s)k_n}
\end{eqnarray*}
as estimator for $\widetilde \rho^s$.

\subsubsection{The global model for the threshold \texorpdfstring{$k$\textsubscript{$n$}}{kn}} \label{sec:globalModel}
Next, we extend the previous model and assume that $k_n \in \N$ is not fixed anymore, it has additionally to be estimated. For this task, we use all observations $\bX_1, \ldots, \bX_n$ and not only the $k_n$ largest observations.  We consider an artificial random vector $\bT_n'=(T_{n,1}',\ldots,T_{n,2^d}')^\top$ in $\R^{2^d}$ which includes extreme and non-extreme observations, where the  $2^d-1$ components $T_{n,1}',\ldots,T_{n,2^{d-1}}'$ count the number of extreme observations in the subsets $C_{\beta_1}, \ldots, C_{\beta_{2^d-1}} $. The  $2^d$-th component $T_{n,2^d}'$ counts the number of non-extreme values and is $\Bin(n, 1 - q_n)$-distributed for some $q_n\in(0,1)$.
To be more precise we assume that $\bT_n' \sim \Mult(n,\bp_n')$ with
\begin{eqnarray*}
    \bp_n'=(q_np_{n,1}',\ldots,q_np_{n,2^{d-1}}',1-q_n)
\end{eqnarray*}
and the conditional distribution given $\td = n-k_n$  satisfies
\begin{eqnarray} \label{2.4}
   \P_{(T_{n,1}',\ldots,T_{n,2^d-1}')|T_{n,2^d}'=n-k_n}= \P_{(T_{n,1}(k_n),\ldots,T_{n,2^d-1}(k_n))}.
\end{eqnarray}
The idea of this assumption is that if we have $k_n$ extreme observations (and hence, $n-k_n$ non-extreme observations), then the distribution of the extreme directions $(T_{n,1}',\ldots,T_{n,2^d-1}')$ in the global model is the same as that of the local model $(T_{n,1}(k_n),\ldots,T_{n,2^d-1}(k_n))$ with threshold $k_n$.

Now, the approach to detect the bias directions and the threshold $k_n$ is similar to the previous section.  We fit a  multinomial distribution from the class 
$\{\Mult(n,\bA_s'(\widetilde{\bp}^{'s})):\,\widetilde{\bp}^{'s}\in  \Theta_s' \}$ to the artificial random vector $\bT_n'$ where $\bA_s':\R^{s+1}\to\R^{2^d}$ is defined as
\begin{align*} 
      \bA_s'(\widetilde{\bp}^{'s})=(q^{\prime s}\widetilde{p}^{'s}_1,\ldots,q^{\prime s}\widetilde{p}^{'s}_s , \underbrace{q^{\prime s} \frac{1 - \sum_{j=1}^s \widetilde{p}^{'s}_j}{r-s},\ldots,q^{\prime s}\frac{1  -   \sum_{j=1}^s \widetilde{p}^{'s}_j}{r-s}}_{r-s}, \underbrace{0,\ldots,0}_{2^d - r - 1},1-q^{\prime s})^{\! \top}
\end{align*}
and the parameter space $\Theta_s'$ is
\begin{align*}
    \Theta_s'& \coloneqq \left\{\widetilde{\bp}^{ \prime s} = (\widetilde{p}_1^{\prime s}, \ldots,  \widetilde{p}_s^{\prime s},q^{\prime s}) \in (0,1)^{s+1} : \,\widetilde p_1^{\prime s}\geq\cdots\geq \widetilde p_s^{'s} ,\, \sum_{j=1}^s \widetilde{p}_j^{\prime s}  < 1\right\}
    =\Theta_s\times (0,1).
\end{align*}
Finally, we define 
$$\widetilde{\rho}^{\prime s}   \coloneqq  \frac{1 - \sum_{j=1}^s \widetilde{p}_j^{\prime s} }{r-s} \quad  \text{ for }\widetilde{\bp}^{ \prime s} \in  \Theta_s'.$$ This ends in the following model.

\medskip
\textsc{MODEL} $M^{\prime s}_n$: \hspace{0.2cm} \label{asu:Model}
\textit{The family of multinomial distributions $\{\Mult(n,\bA_s'(\widetilde{\bp}^{\prime s})):\,\widetilde{\bp}^{\prime s}\in  \Theta_s' \}$ 
with log-likelihood function 
\begin{align}
\log L_{M_n^{\prime s}}  (\widetilde{\bp}^{\prime s}\, \vert \, \bT_n') =& \log(n!) - \sum_{j=1}^{2^d} \log(T_{n,j}'!) +  \sum_{j=1}^{s} T_{n,j} \log( {\widetilde{q}} {\widetilde{p}^s_{j}}) \nonumber \\
& + \left(  \sum_{j= s +1}^{2^d-1} T_{n,j}' \right) \log( {\widetilde{q}} {\widetilde{\rho}^s})   + T_{n,2^d}' \log(1-{\widetilde{q}}) \label{Model global}
\end{align}
is called Model $M^{\prime s}_n$.
}
\medskip

To link the global model with the local model we require further assumptions.

\begin{assumptionletter}~ \label{asu:Model_global}
\begin{enumerate}[(B1)]
\item \label{asu:BIC_global_expectation_local}
Suppose $\td$  and $\bT_{n}$ are independent, and for $j = 1, \ldots, r$ we have as $n\to\infty$,
\begin{equation*}
          \E \left[ \frac{1}{n - \td}T_{n,j}' \vert \td \right] = \E \left[ \frac{1}{k_n} T_{n,j}(k_n)  \right]  + o_\P( 1 ).
    \end{equation*}
\item \label{asu:BIC_global_expectation_local2} Suppose  for $j = 1, \ldots, r$ we have as $n\to\infty$,
\begin{align*}
      \E \left[ \frac{1}{(n - \td)^2}(T_{n,j}')^2 \vert \td \right] &= \E \left[ \frac{1}{k_n^2} (T_{n,j}(k_n) )^2  \right]  + o_\P( 1 ).
\end{align*}
\item  There exist constants $K_1, K_2 \in (0,\infty)$ such that $$K_1 < \liminf\limits_{\ninf} \frac{n q_n}{k_n}\leq\limsup\limits_{\ninf} \frac{n q_n}{k_n} < K_2.$$ \label{asu:BIC_glob_qn_kn}
\end{enumerate}
\end{assumptionletter}

Due to the Assumptions (\ref{asu:Model_global}\ref{asu:BIC_global_expectation_local}) and (\ref{asu:Model_global}\ref{asu:BIC_global_expectation_local2})  the first and second moment of the relative number of extreme observations in the global model and the local model behave similarly.
The last  Assumption (\ref{asu:Model_global}\ref{asu:BIC_glob_qn_kn}) gives a connection between the asymptotic behavior of $q_n$ and $k_n$. In particular, it implies $k_n=O(nq_n)$ as $n\to\infty$.

\section{Quasi-Akaike information criterion} \label{sec:QAIC}

In the following, we propose an information criterion inspired by the Akaike information criterion and therefore, we refer to as \textit{quasi-Akaike information criterion} ($\QAIC$). Unlike the approach of \citet{meyer_muscle23}, which is based on the likelihood function of a multinomial distribution, our method employs the Gaussian distribution. More specifically,  
the Akaike information criterion  ($\AIC$)
introduced by  \citet{meyer_muscle23} for selecting the number of extreme directions
is motivated by minimizing the expected Kullback-Leibler (KL) divergence between the true distribution of $\bT_n(k_n)$  and the multinomial distribution $\Mult(k_n,\widehat{\bp}_{n}^s)$ where $\widehat{\bp}_{n}^s$ is the MLE given in \eqref{MLE:mult}. The AIC
 is defined as
\begin{equation} \label{def:AIC_MW}
    \AIC_{k_n}(s) \coloneqq -  \log L_{M^s_{k_n}}(\widehat{\bp}_{n}^s\, \vert \, \bT_n(k_n))  + s, \quad s = 1, \ldots, r,
\end{equation}
for fixed $k_n$. The number  $s^*$ of extreme directions is then estimated via $$\widehat{s}_n = \argmin_{s = 1, \ldots, r} \AIC_{k_n}(s).$$
However, a limitation of the $\AIC$ is that it is not a weakly consistent information criterion which is typically expected in a fixed-dimensional setting as $\ninf$ and $d \in \N$  (see \citet{ModelSelection, C:16}).

\begin{theorem} \label{th:AIC_Cons}
Suppose \Cref{asu:directions,Assumption:main} holds. Then
\begin{equation*}
    \limn \P( \AIC_{k_n}(s) > \AIC_{k_n}(s^*) ) \begin{cases}
< 1 \quad  \text{ for } s > {s^*}, \\
= 1 \quad \text{ for } s < {s^*}.
\end{cases}
\end{equation*}
\end{theorem}
A key conclusion of \Cref{th:AIC_Cons} is that the $\AIC$ has asymptotically a non-vanishing probability of overestimating $s^*$ and hence, it is not a weakly consistent information criterion. 
The proof of \Cref{th:AIC_Cons}, along with all proofs of this section, is relegated to \Cref{sec:proof_QAIC}.

\begin{remark} \label{remark:AIC:r}
 Suppose Assumption (\ref{asu:directions}\ref{(A1)}) is replaced by the condition $\sqrt{k_n\rho_n}(\widehat r_n-r)\Pconv 0$ and that the $\AIC$ is defined using $\widehat r_n$ instead of $r$. Then 
$$\sqrt{k_n \rho_n} \sum_{j=r+1}^{\widehat r_n} \left( \frac{T_{n,j}(k_n)}{\rho_n k_n}  - 1 \right)=o_\P(1)$$
and hence, if we follow the proof of \Cref{th:AIC_Cons}, we see that
the consistency result remains true for this modified $\AIC$, which is finally used in practice. 
\end{remark}
In contrast, the main advantage of the $\QAIC$, which we introduce next, is that it is a weakly consistent information criterion.  

\subsection{Quasi Akaike information criterion for the number of directions \texorpdfstring{$s$}{s}}

The reason  behind employing the likelihood function of a Gaussian distribution for the $\QAIC$ is that due to \Cref{asu:directions} the asymptotic behavior as $ \ninf$,
\begin{align*} 
\sqrt{k_n} \diag( \bp_n^*)^{-1/2} \left( \frac{\bcT_{n}}{k_n} - \bp_n^* \right) \Dconv \mathcal{N}_r ( \mathbf{0}_r, \bI_{r}) 
\end{align*}
holds, 
i.e. the asymptotic distribution of $\bcT_{n}$ is  similar to the distribution of a $r$-variate normal distribution with mean $k_n \bp_n^*$ and   covariance matrix $k_n \diag(\bp_n^*) $.  
Therefore, the idea is to calculate the expected  Kullback-Leibler divergence of the true distribution $\P_{\bcT_n}$ of $\bcT_n$ with density $f_*$ and the normal distribution  $\mathcal{N}_r( k_n \bB_s(\bpbt^s),  k_n \diag(\bB_s(\bpbt^s))   ), \bpbt^s = ( \pbt^s_1, \ldots, \pbt^s_s, \rhobt^s) \in \R_+^{s+1}$,  where $\bB_s:\R_+^{s+1}\to\R_+^r$ is defined as $$\bB_s(\bz)= \Big(z_1,\ldots,z_s,z_{s+1},\ldots,z_{s+1}\Big)^\top.$$
The likelihood function of $\mathcal{N}_r( k_n \bB_s(\bpbt^s),  k_n \diag(\bB_s(\bpbt^s)) )$ is denoted by  $L_{\mathcal{N}_r}(  \bpbt^s  |\bcT_n) $.
For $ \bpbt^s $ we use the estimator  
\begin{eqnarray}  \label{6.1}
    \begin{array}{rl}
    \bpb_n^s ( \widetilde{\bcT}_n )  & \coloneqq ( \pb_{n,1}^s ( \widetilde{\bcT}_n ), \ldots, \pb_{n,s}^s( \widetilde{\bcT}_n ), \rhob_n^s ( \widetilde{\bcT}_n ) )^\top \in \R_+^{s+1} \quad \quad \text{ with }\\
   \pb_{n,j}^s ( \widetilde{\bcT}_n ) & \coloneqq {\displaystyle \frac{\widetilde{\cT}_{n,j}}{k_n}, }\quad j=1,\ldots,s, \qquad
    \rhob_n^s  ( \widetilde{\bcT}_n ) {\displaystyle \coloneqq \frac{1}{r-s} \sum\limits_{j={s+1}}^r \frac{\widetilde{\cT}_{n,j }}{k_n} }
\end{array}
\end{eqnarray}
where  $\widetilde{\bcT}_n$ is an i.i.d. copy of $\bcT_n$.

\begin{remark}
It might happen that $\sum_{j=1}^s \pb_{n,j}^s ( \widetilde{\bcT}_n ) + (r-s) \rhob^s_n( \widetilde{\bcT}_n) \ne 1$.  In this case, $ \bB_s( \bpb_n^s ( \widetilde{\bcT}_n ))$ is in general not a probability vector
and
$( \pb_{n,1}^s ( \widetilde{\bcT}_n ), \ldots, \pb_{n,s}^s( \widetilde{\bcT}_n ))\notin\Theta_s$. But due to Assumption (\ref{asu:directions}\ref{(A5)}) we have
as $n\to\infty$,
\begin{eqnarray*}
    \frac{\pb_{n,j}^s ( \widetilde{\bcT}_n )}{p_{n,j}}\stackrel{\P}{\longrightarrow}1
    \quad \text{ and } \quad 
     \frac{\rhob_n^s ( \widetilde{\bcT}_n )}{\frac{1}{r-s}\sum_{j=s+1}^rp_{n,j}}\stackrel{\P}{\longrightarrow}1,
\end{eqnarray*}
such that $\lim_{n\to\infty}\P(( \pb_{n,1}^s ( \widetilde{\bcT}_n ), \ldots, \pb_{n,s}^s( \widetilde{\bcT}_n )) \in\Theta_s)=1.$
\end{remark}
 In summary, we calculate
\begin{align}
&\E\left[\KL(\P_{\bcT_n},\mathcal{N}_r( k_n \bB_s(\bpbt^s), k_n \diag( \bpbt^s) ) )|_{ \bpbt^s = \bpb^s_n( \widetilde{\bcT}_n ) }\right] \nonumber \\ 
& \qquad  = \E \left[ \log  f_*(\bcT_n ) \right]-
\E \left[ \log \left( L_{\mathcal{N}_r} (  \bpb^s_n( \widetilde{\bcT}_n )| \bcT_n  )\right)  \right].\label{eq:KL}
\end{align}
\begin{remark} 
 The $\AIC$ is based on the multinomial distribution whereas the $\QAIC$ is based on the multivariate normal distribution. Although it seems at first view that both approaches are different they are related due to local limit theorems for the multinomial distribution as given in \citet{multinomial_local_limit}.
\end{remark}

Next, we derive an auxiliary result that helps to approximate the second term in \eqref{eq:KL} for $s\geq s^*$.

\begin{proposition} \label{th:QAIC_Likelihood_Approx} 
Suppose Assumption~\ref{Assumption:main} holds and $s \ge s^*$. Furthermore, let $\widetilde{\bcT}_n$ be an independent and  identically distributed  copy of $\bcT_n$, and let $\bpb^s_n( \widetilde{\bcT}_n)$ be the  estimator in \eqref{6.1} and similarly we define $\bpb^s_n({\bcT}_n)$. Then there exists a random variable $Y$ with $\E[Y]=0$ such that as $\ninf$,
\begin{align*}
\log L_{\mathcal{N}_r} &  (  \bpb^s_n( \widetilde{\bcT}_n )  \, \vert \, \bcT_n ) + \frac12 r \log(2 \pi) + \frac12 r \log(k_n) \\
&  +  \frac12 \sum_{j=1}^s \log (\pb^s_{n,j}( \bcT_n )) + \frac12 (r - s) \log( \rhob^s_{n}( \bcT_n) ) + \frac{r+ s +1}{2}
\Dconv  Y.
\end{align*}
\end{proposition}
 Therefore, for $s\geq s^*$ we approximate the second term in \cref{eq:KL} by
\begin{align*}
&\hspace*{-0.6cm} 
 - \E\left[\log L_{\mathcal{N}_r}   (  \bpb^s_n( \widetilde{\bcT}_n )  \, \vert \, \bcT_n )\right]\\ 
\approx \, & \, \frac12 \E \Bigg[  r \log(2 \pi) + r \log(k_n) +  \sum_{j=1}^s \log (\pb^s_{n,j}( \bcT_n ))    + (r - s) \log( \rhob^s_{n}( \bcT_n )) + r+ s+1  \Bigg]
\end{align*}
and neglect the expectation. The first term  $ \E \left[ \log  f_*(\bcT_n ) \right]$ in \cref{eq:KL} and the $+1$ do not influence the choice of the model, therefore we skip them. This leads to the following definition of the theoretic quasi-information criterion for $s\geq s^*$,
\begin{eqnarray*}
    \QAIC'_{k_n}(s) &\coloneqq & r \log(2 \pi) + r \log(k_n) +   \sum_{j=1}^s \log ( \pb^s_{n,j}( \bcT_n ))  + (r - s) \log(  \rhob^s_n( \bcT_n ))  + r +s.
\end{eqnarray*}
If $s < s^*$ this information criterion works as well since 
\begin{align*}
 \sum_{j=1}^s \log & ( \pb^s_{n,j}( \bcT_n ))  + (r - s) \log(  \rhob^s_n( \bcT_n )) \\
 &\stackrel{\P}{\longrightarrow} \sum_{j=1}^{s}\log(p_j)+(r-s)\log\left(\frac{\sum_{j=s+1}^{s^*}p_j}{r-s}\right)>-\infty
 \intertext{and for $s>s^*$ we have}
 \sum_{j=1}^s \log & ( \pb^s_{n,j}( \bcT_n ))  + (r - s)  \log(  \rhob^s_n( \bcT_n )) \stackrel{\P}{\longrightarrow} -\infty.
\end{align*} 
 Therefore, the information criterion does not select $s < s^*$.

Moreover, since
\begin{align*}
    \sum_{j=1}^s &\log ( \pb^s_{n,j}( \bcT_n ))  + (r - s)  \log(  \rhob^s_n( \bcT_n )) \\
     &- \sum_{j=1}^s \log ( \widehat p^s_{n,j}( \bT_n(k_n) ))  +\LB{-} (r - s)  \log(  \widehat\rho^s_n( \bT_n(k_n) )) \stackrel{\P}{\longrightarrow}
    0
    \end{align*}
the choice between  estimator  
$\bpb^s_n(\bcT_n)$
or $\widehat{\bp}^s_n=\widehat{\bp}^s_n(\bT_n(k_n))\in\Theta_s$ with $\widehat{\rho}^s_n = \widehat{\rho}^s_n(\bT_n(k_n))$ does not significantly change the outcome, so either can be used. Since in applications $u_n$ and hence, $\bpb^s_n(\bcT_n)$ is unknown, we finally define the information criterion based on the estimators $\widehat{\bp}^s_n$ and $\widehat{\rho}^s_n $.

\begin{definition}
For the number of extreme directions $s$ with fixed $k_n$ the \textit{quasi Akaike information criterion} ($\QAIC$) is defined as
\begin{eqnarray*}
    \QAIC_{k_n}(s) &\coloneqq & r \log(2 \pi) + r \log(k_n) +   \sum_{j=1}^s \log ( \widehat{p}^s_{n,j})  + (r - s) \log(  \widehat{\rho}^s_n )  + r +s
\end{eqnarray*}
for $s = 1, \ldots, r$ and an estimator for $s^*$ is $\widehat{s}_n \coloneqq \argmin_{1 \le s \le r} \QAIC_{k_n}(s)$.
\end{definition}

\begin{remark} $\mbox{}$
\begin{enumerate}[(a)]
    \item During the derivation of the $\QAIC$ we assumed that $r$ is constant and hence, it should not influence the optimal value of the $\QAIC$. However, the simulation study shows that in applications $r$ has a significant impact on the performance of the $\QAIC$ because in practice $r$ depends on $k_n$. 
    \item The derivation of a $\QAIC$ with an estimator based on the likelihood function of the normal distribution $L_{\mathcal{N}_r}$ is possible with similar results but leads to a more elaborate and longer calculation. In this case, the estimator is given by
    \begin{align*}
{\widehat{p}^{\, G}}_{n,j} &= \frac{-1}{2 k_n} + \sqrt{ \frac{1}{4 k_n^2} + \frac{T_{n,j}(k_n)^2}{k_n^2} }, \quad j = 1, \ldots, s,\\
{\widehat{\rho}^{\, G}}_{n} &= \frac{-1}{2 k_n} + \sqrt{ \frac{1}{4 k_n^2} + \frac{1}{r - s} \sum_{j=s +1}^r \frac{T_{n,j}(k_n)^2}{k_n^2} }.
\end{align*}
    The performance of both approaches is similar and therefore only $\QAIC$ is included in the simulation study.
\end{enumerate}

\end{remark}

\begin{theorem} \label{th:QAIC_Consistency}
Suppose \Cref{asu:directions,Assumption:main} holds. Then
    \begin{align*}
        \limn \P(\QAIC_{k_n}(s) - \QAIC_{k_n}(s^*) > 0) = 1 \, \quad \text{ for }  s \neq s^*.
    \end{align*}
\end{theorem}

 Compared to the $\AIC$, the $\QAIC$  has the advantage that it is weakly consistent for fixed $k_n$ in contrast to the $\AIC$.

\begin{remark} \label{remark:QAIC:r}
 Suppose Assumption (\ref{asu:directions}\ref{(A1)}) is replaced by the condition $\widehat r_n\Pconv r$ and that the $\QAIC$ is defined using $\widehat r_n$ instead of $r$. Then the consistency result remains true for this modified $\QAIC$. Note that here a weaker condition is used as for the $\AIC$ in \Cref{remark:AIC:r}, where we required $\sqrt{k_n\rho_n}(\widehat r_n-r)\Pconv 0$.
\end{remark}

\subsection{Quasi Akaike information criterion for the threshold \texorpdfstring{$k$\textsubscript{$n$}}{kn}}

For the $\QAIC$ for the threshold $k_n$ we follow the definition of the global model for the $\AIC$ in \citet{meyer_muscle23} which is defined as
\begin{align*}
  \AIC_{n,s}(k_n) &\coloneqq   \frac{\AIC_{k_n}(s) }{k_n} +  \frac{k_n}{n}
\end{align*}
with $\AIC_{k_n}(s)$ as in \Cref{def:AIC_MW}.
However, since we consider two times the negative likelihood instead of just the negative likelihood we include additionally the factor $1/2$ and obtain the following information criterion.

\begin{definition} 
For the number of exceedances $k_n$ the \textit{quasi-Akaike information criterion} ($\QAIC$) \textit{for the threshold $k_n$} for the Model $M^{\prime s}_n$ is defined as
\begin{align*}
    \QAIC_{n,s}(k_n) \coloneqq{}&  \frac{\QAIC_{k_n}(s) }{2 k_n} + \frac{k_n}{n}\\
      ={}&  \frac{r \log(2 \pi) + r \log(k_n) +   \sum_{j=1}^s \log ( \widehat{p}^s_{n,j} ) + (r - s) \log(  \widehat{\rho}^s_{n,j}  )  + r +s }{2 k_n} + \frac{k_n}{n} 
\end{align*}
for $k_n=1,\ldots,n$ with estimator $\widehat{k}_n \coloneqq \argmin_{k_n \in K} \bigl\{ \min_{1 \le s \le r}  \QAIC_{n,s}(k_n) \bigr\}$ for $K \subset \{1,\ldots, n\}$.
\end{definition}

\begin{remark} 
 An interpretation of this information criterion is as follows. The division by $k_n$ can be seen as a weight, which is assigned to a pair $(s,k_n)$. Therefore, when $k_n$ is large, the weight of the corresponding model gets smaller.  Also, $k_n / n$ corresponds to the relative proportion of extreme observations and acts as a penalty for increasing $k_n$. 
\end{remark}

\section{Mean squared error information criterion} \label{sec:MSEIC}
Next, we explore an information criterion based on the mean squared error (MSE) for both the number of directions $s$ in \Cref{sec:MSEIC_local}  as well as for the threshold $k_n$ in \Cref{sec:MSEIC_global}, which performs in particular well for a small number of observations. The proofs of this section are moved to \Cref{sec:proof_MSEIC}.

\subsection{Mean squared error information criterion for the number of extreme directions \texorpdfstring{$s$}{s}} \label{sec:MSEIC_local}

The basic idea of the AIC is to minimize the Kullback-Leibler distance of the true distribution and a parametric family of distributions. This minimum is approximated by the expected Kullback-Leibler distance of the true distribution and the estimated distribution as is done in \Cref{eq:KL}. In the following, we use the same ideas but instead of using the Kullback Leibler distance we use the normalized mean-squared error (MSE) of the parameter estimator and find an approximation of

\begin{eqnarray} \label{eq:MSE_Def}
     \MSE_{k_n}(s) \coloneqq \E\left[\ell^2 \big(\bpb^s_n( \widetilde{\bT}_n(k_n))\vert \bT_n \big)\right]
\end{eqnarray}
instead of $\E\left[\log L_{\mathcal{N}_r}   (  \bpb^s_n( \widetilde{\bcT}_n )  \, \vert \, \bT_n(k_n) )\right]$ as is done in \Cref{eq:KL}, where $\widetilde{\bT}_n(k_n)$ is an independent and  identically distributed copy of $\bT_n(k_n)$ and
\begin{align}
    \ell^2 \big(\bpbt^s \vert \bT_n(k_n) \big)  \coloneqq{} & \left \Vert \sqrt{k_n} \diag( \bB_s(\bpbt^s ))^{-1/2} \left( \frac{\bT_n(k_n)}{k_n}  - \bB_s( \bpbt^s ) \right) \right\Vert_2^2 \nonumber\\
    ={}&  \sum_{j=1}^s  \frac{k_n}{ \pbt^s_j }  \left(   \frac{T_{n,j(k_n)} }{k_n} - \pbt^s_j \right)^2 + \frac{k_n}{ \rhobt^s } \sum_{j=s+1}^r  \left(\frac{ T_{n,j}(k_n)}{k_n} - \rhobt^s \right)^{2} \nonumber
\end{align}
for $\bpbt^s = ( \pbt^s_1, \ldots, \pbt^s_s, \rhobt^s) \in \R_+^{s+1}$.  
Note, if in Assumption (\ref{asu:directions}\ref{(A4)})  not only the weak convergence but also the componentwise $L_1$ convergence holds, then \linebreak  $\lim_{n\to\infty} \E\left[\ell^2 \big((p_{n,1},\ldots,p_{n,s^*}, \rho_n)\vert \bT_n(k_n) \big)\right]=r-1$ which motivates this approach. 
First, we derive an auxiliary result that helps to approximate $\ell^2 \big(\bpb^s_n( \widetilde{\bT}_n(k_n))\vert \bT_n(k_n) \big)$.

\begin{theorem} \label{conv:MSE_Chisq} 
Suppose Assumption~\ref{Assumption:main} holds and $s \ge s^*$. Furthermore, let $\widetilde{\bT}_n(k_n)$ be an independent and  identically distributed copy of $\bT_n(k_n)$, and let $\bpb^s_n( \widetilde{\bT}_n(k_n))$ be the  estimator in \eqref{6.1}. Similarly, we define $\bpb^s_n({\bT}_n(k_n))$. Then there exists a random variable $Y$ with $\E[Y]=0$ such that as $\ninf$,
\begin{align*}
  \ell^2 \big(\bpb^s_n(\widetilde{\bT}_n(k_n))& \vert \bT_n(k_n) \big)   -  \frac{k_n}{ \rhob_{n}^s( \bT_n(k_n)) } \sum_{j=s+1}^r  \left(\frac{ T_{n,j}(k_n)}{k_n} - \rhob_{n}^s( \bT_n(k_n)) \right)^{2} - 2 s \Dconv Y.
\end{align*}
\end{theorem}

Therefore, for $s\geq s^*$ we approximate \eqref{eq:MSE_Def} by
\begin{align*}
\MSE_{k_n}(s) \approx \, \E \Bigg[   \frac{k_n}{ \rhob_{n}^s( \bT_n(k_n)) } \sum_{j=s+1}^r  \left(\frac{ T_{n,j}(k_n)}{k_n} - \rhob_{n}^s( \bT_n(k_n)) \right)^{2} + 2 s \Bigg].
\end{align*}

Analogously to \Cref{sec:QAIC}, we neglect the expectation, which leads to the following information criterion.

\begin{definition}
For the number of extreme directions $s$ with fixed $k_n$ the \textit{mean squared error information criterion} ($\MSEIC$) is defined as
\begin{eqnarray*}
\MSEIC_{k_n}(s) &\coloneqq & \frac{k_n}{\sum_{l=s+1}^r \frac{T_{n,l}(k_n)}{k_n(r-s)}} \sum_{j=s+1}^r  \left( \frac{T_{n,j}(k_n)}{k_n} -  \sum_{i=s+1}^r \frac{T_{n,i}(k_n)}{k_n(r-s)} \right)^{ 2} + 2 s,  
\end{eqnarray*}
for $s = 1, \ldots, r-1$ with $\MSEIC_{k_n}(r):=2 r$. An estimator for $s^*$ is defined by $\widehat{s}_n \coloneqq \argmin_{1 \le s \le r} \MSEIC_{k_n}(s)$.
\end{definition}

\begin{theorem} \label{th:MSE_Consistency}
Suppose \Cref{asu:directions,Assumption:main} holds. Then
 \begin{equation*}
    \limn \P( \MSEIC_{k_n}(s) > \MSEIC_{k_n}(s^*) ) \begin{cases}
< 1 \quad \text{ for } s > s^*, \\
= 1 \quad \text{ for } s < s^*.
\end{cases}
\end{equation*}
\end{theorem}

In particular, for $s<s^*$ this information criterion is consistent, but unfortunately not for $s>s^*$. However, this is not surprising because the basic ideas are related to the AIC which is also not a consistent information criterion. However, the simulation study in \Cref{sec:NumExp} shows that MSEIC performs extremely good in practice.

\subsection{Mean squared error information criterion for the threshold \texorpdfstring{$k$\textsubscript{$n$}}{kn}} \label{sec:MSEIC_global}

Now,  we extend the information criterion $\MSEIC$ to choose the optimal threshold $k_n$. Therefore, we use not only our knowledge about the extreme observations but also our knowledge of the non-extreme observations, similarly to the global model $M_n'^{s}$, only that there is no distributional assumption. As before we assume here that $\bT_{n,\{1, \ldots, r\}}'$ pertains the information about the observed extreme directions
and $\td$ the non-extreme observations, where $\td$ 
  is assumed to be binomially distributed. 
The MSE information criterion for the threshold $k_n$ is then defined as weighted MSE 
\begin{align} \label{eq:MSE_glob}
    \MSE_n^{\prime s} \coloneqq& \E \left[    q'  \E   \left[   \left\Vert \sqrt{n- \td}  \diag( \bp' )^{-1/2}  \left( \frac{\bT_{n, \{1, \ldots, r\}}'}{n- \td} -    \bp' \right)  \right\Vert^2_2  \right] \bigg\vert_{ \bp' =  \widehat{\bp}_n( \frac{\widetilde{\bT}_n(k_n)}{k_n}), q' = \frac{k_n}{n}     }  \right] \nonumber \\
   & +  \E \left[ (1- q') \E \left[   \left\Vert \sqrt{n} ( q' ( 1-q') )^{-1/2} \left(  \frac{\td}{n}  - (1-q') \right) \right\Vert^2_2  \right] \bigg\vert_{  q' = \frac{k_n}{n}  }  \right]
\end{align}
with weight  $q'$ for the estimation of the probabilities of extreme directions and weight $(1-q')$ for the estimation of the probability of non-extremes. Since we want to make statements about the optimal choice of $k_n$ which models the number of extreme directions, the weight in the estimation of the probabilities of the extreme directions is chosen higher.
A connection between the MSE information criterion for the threshold $k_n$ and the MSE information criterion for the number of extreme directions $s$ exists through the following theorem. 

\begin{theorem} \label{th:MSE_global_derivation}
    Suppose Assumptions (\ref{asu:Model_global}\ref{asu:BIC_global_expectation_local}),  (\ref{asu:Model_global}\ref{asu:BIC_global_expectation_local2}) and $k_n ( 1 - \frac{n q_n}{k_n})^2 \to 0$ as $n\to\infty$.
    Then
\begin{align*}
    \MSE_n^{\prime s}   & =  q_n     \left(   \MSE_{k_n}(s)+    \frac{n}{k_n}    + no \left(  \frac{1}{n q_n} \right) \right).
\end{align*}
\end{theorem}
Since $q_n$ is not influenced by $k_n$ and $o((n q_n)^{-1})$ is of a smaller order than $1/k_n$ by Assumption (\ref{asu:Model_global}\ref{asu:BIC_glob_qn_kn}), we neglect $q_n$ and the last term. Consequently, we define the following information criterion.

\begin{definition}
For the number of exceedances $k_n$ the \textit{mean squared error information criterion ($\MSEIC$) for the threshold} $k_n$ for the Model $M^{\prime s}_n$ is defined as
\begin{eqnarray*}
     \MSEIC_{n,s}(k_n)  &\coloneqq &\MSEIC_{k_n}(s)  +\frac{n}{k_n}, \quad k_n=1,\ldots,n,
\end{eqnarray*}
with estimator $\widehat{k}_n \coloneqq \argmin_{k_n \in K} \bigl\{ \min_{1 \le s \le r}  \MSEIC_{n,s}(k_n) \bigr\}$ for $K \subset \{1,\ldots, n\}$.
\end{definition}

\begin{remark}
    The general structure of this threshold information criterion differs from the other derived information criteria for the threshold selection as $$\AIC_{n,s}(k_n)=\frac{\AIC_{k_n}(s) }{k_n} +  \frac{k_n}{n}\quad \text{ and } \quad \QAIC_{n,s}(k_n)=\frac{\QAIC_{k_n}(s) }{2 k_n} + \frac{k_n}{n}.$$ 
    Therefore, we performed a simulation study with the criterion $\MSEIC_{k_n}(s)/k_n  + k_n / n $, defined analog to $\AIC_{n,s}(k_n)$. The simulation study confirms that this choice of information criteria is not the suitable choice. The result is not surprising,
    since $\MSEIC$ is not based on a likelihood-based approach. 
\end{remark}

\section{Bayesian information criterion} \label{sec:BIC}

In addition to the $\AIC$, the Bayesian information criterion (BIC) introduced in \citet{schwarz} is the most popular one. The basic idea of the BIC is to find the model with the highest posterior probability given the data. First, we derive a BIC for $s$ in \Cref{sec:localBIC} and then for $k_n$ in \Cref{sec:BIC_Global}.
The proofs of this section can be found in \Cref{sec:proof_BIC}.

\subsection{Bayesian information criterion for the number of extreme directions \texorpdfstring{$s$}{s}}
\label{sec:localBIC}

 In the following, we derive a BIC for $s$ by bounding the posterior probability as in \citet{bicgeneralisation}. Therefore, 
we assume throughout this section Model $M^s_{k_n}$ and use the following notation. Let 
$\Q$ be a discrete prior distribution over the set of models \linebreak $\{M^s_{k_n}:\,s=1,\ldots,r\}$, $g(\,\cdot\, \vert \, M^s_{k_n})$ be the prior density over the parameter space $\Theta_s$ given Model $M^s_{k_n}$,  
$L_{M^s_{k_n}}(\,\cdot\, \vert \, \bT_n(k_n))$ be the likelihood function of Model $M^s_{k_n}$ if we observe $\bT_n(k_n)$
and $f$ be the (unknown) marginal probability of $\bT_n(k_n)$. 
Given the data $\bT_n(k_n)$ the goal is to determine the Model $M^s_{k_n}$ with the highest posterior probability  $\P(  M^s_{k_n} | \bT_n(k_n) )$ for $ s = 1, \ldots, r$. 
Therefore, note that Bayes Theorem yields for the posterior  density for $M^s_{k_n}$ and $\widetilde \bp^s$ 
\begin{equation*}
    h( (M^s_{k_n}, \widetilde{\bp}^s) \, \vert \, \bT_n(k_n) ) = \frac{   L_{M^s_{k_n}}(\widetilde{\bp}^s\, \vert \, \bT_n(k_n))  g(\widetilde{\bp}^s\vert \, M^s_{k_n})\, \Q(M^s_{k_n})}{ f(\bT_n(k_n))}.
\end{equation*}
Hence,  the posterior probability for $ M^s_{k_n}$ is
\begin{align*}
\P(  M^s_{k_n} | \bT_n(k_n) )  =  \frac{\Q(M^s_{k_n}) \int_{ \Theta_s} L_{M^s_{k_n}}(\widetilde{\bp}^s\, \vert \, \bT_n(k_n))  g(\widetilde{\bp}^s\, \vert \, M^s_{k_n}) \di \widetilde{\bp}^s}{f(\bT_n(k_n))}. 
\end{align*}
Consequently maximizing the posterior probability is equivalent to minimizing  
\begin{align} \label{eq:BIC_Post_Prob}
-2 \log  \P(  M^s_{k_n} | \bT_n(k_n) ) = &2 \log f(\bT_n(k_n)) -2 \log \Q(M^s_{k_n}) \nonumber \\
& -2 \log \Bigl( \int_{ \Theta_s} L_{M^s_{k_n}}(\widetilde{\bp}^s\, \vert \, \bT_n(k_n))  g(\widetilde{\bp}^s\, \vert \, M^s_{k_n}) \di \widetilde{\bp}^s \Bigr).
\end{align}
For the derivation of the BIC, we require further assumptions.
\begin{assumptionletter} \label{asu:BIC_local}
    For any $s \in \{1, \ldots, r\}$ we assume the following:
    \begin{enumerate}[(C1)]
     \item There exist constants $0 < b \le B < \infty$ such that the prior density $g(\,\cdot\, \, \vert \, M^s_{k_n})$ on $\Theta_s$ satisfies \label{asu:BIC_local3}
     \begin{equation*}
         b \le g(\widetilde{\bp}^s \, \vert \, M^s_{k_n}) \le B \quad \text{for all } \widetilde{\bp}^s \in \Theta_s.
     \end{equation*}  
    \item The prior distribution $\Q$ is a uniform distribution on $\{M^s_{k_n}: \, s = 1, \ldots, r\}$, i.e.\  $\Q(M^s_{k_n}) = \frac1r$ for $s = 1, \ldots, r$. \label{asu:BIC_local4}
     \item   $ k_n \rho_n^{ 5/3} \rightarrow \infty$  and $k_n \rho_n^2 \rightarrow 0$. 
    \label{asu:BIC_local1}
        \end{enumerate}
\end{assumptionletter}

\begin{remark} $ $
\begin{enumerate}[(a)]
    \item   Both Assumptions (\ref{asu:BIC_local}\ref{asu:BIC_local3}) and (\ref{asu:BIC_local}\ref{asu:BIC_local4}) are assumptions on prior distributions, and they reflect that we have no prior information in advance. The lower bound of Assumption (\ref{asu:BIC_local}\ref{asu:BIC_local3}) can be relaxed since we require only a lower bound in the neighborhood of $\widehat p_n^s$. However, it has been omitted in this paper for the sake of brevity.

    \item The assumption on the uniform distribution on the set of all possible models in  (\ref{asu:BIC_local}\ref{asu:BIC_local4})  is an uninformative prior distribution where all models have the same probability. Thus, the term $-2 \log \Q(M^s_{k_n})=2\log r $  in \eqref{eq:BIC_Post_Prob} is independent of $s$ and has, from a theoretical point of view, no influence on the information criterion. Of course, it is possible to use a prior distribution depending on $s$ but then the BIC receives an additional penalty term.
    
    \item The assumption $ k_n \rho_n^{ 5/3} \rightarrow \infty$ in (\ref{asu:BIC_local}\ref{asu:BIC_local1}) ensures that $\rho_n$ does not converge to zero too quickly.
    
\end{enumerate}
\end{remark}

The next theorem gives an upper bound for 
\begin{align*}
  -2 \log  \E_{g_s} [ L_{M^s_{k_n}}(\widetilde{\bp}^s\, \vert \, \bT_n(k_n)) ] \coloneqq -2 \log \int_{ \Theta_s} L_{M^s_{k_n}}(\widetilde{\bp}^s\, \vert \, \bT_n(k_n))  g(\widetilde{\bp}^s\, \vert \, M^s_{k_n}) \di \widetilde{\bp}^s,
\end{align*} 
  whereby $\E_{g_s}$ denotes the conditional expectation regarding the prior density $ g(\cdot \, \vert \, M^s_{k_n})$ on $\Theta_s$.  This results then in an upper bound for the negative log posterior probability of the $s$-th Model $M^s_{k_n}$ given $\bT_n(k_n)$. 

\begin{theorem} \label{th:BIC_post_prob}
    Suppose Assumptions \ref{asu:directions}, (\ref{asu:BIC_local}\ref{asu:BIC_local1}) and (\ref{asu:BIC_local}\ref{asu:BIC_local3}) hold. Then  the inequality 
\begin{align*}
        -2 \log & \, \E_{g_s} [ L_{M^s_{k_n}}(\widetilde{\bp}^s\, \vert \, \bT_n(k_n)) ]\\
       \leq &   -2   \log L_{M^s_{k_n}}(\widehat{\bp}_n^s\, \vert \, \bT_n(k_n))- s \log(2 \pi)  + 2 s \log \left( k_n \sqrt{\frac{r}{r-s}}  \right) -2\log b  + o_\P(1)
    \end{align*}
  as $\ninf$  holds.
\end{theorem}
Plugging in Assumption (\ref{asu:BIC_local}\ref{asu:BIC_local4}) and the upper bound  in \Cref{th:BIC_post_prob} in \cref{eq:BIC_Post_Prob} results in
\begin{align*}
-2 &\log \P(  M^s_{k_n}(k_n) | \bT_n(k_n) ) \notag \\
&= 2 \log f(\bT_n(k_n)) + 2 \log r -2 \log  \E_{g_s} [ L_{M^s_{k_n}}(\widetilde{\bp}^s\, \vert \, \bT_n(k_n)) ]    \notag \\
&\le -2   \log L_{M^s_{k_n}}(\widehat{\bp}_n^s\, \vert \, \bT_n(k_n))- s \log(2 \pi)  + 2 s \log \left( k_n \sqrt{\frac{r}{r-s}}  \right) \notag \\
        & \qquad   +2 \log f(\bT_n(k_n)) - 2\log b + 2 \log r + o_\P(1) \label{BIC_local_log_probability}.
\end{align*}
This motivates the definition of the following information criterion, where the terms $ 2 \log f(\bT_n(k_n)) - 2\log b + 2 \log r$ are neglected as they are not influenced by $s$. 

\begin{definition}
For the number of extreme directions $s$ with fixed $k_n$ the \textit{Bayesian information criterion concerning the upper bound} ($\BICU$) is defined as
\begin{align*}
    \BICU_{k_n}(s) \coloneqq  -2 \log L_{M^s_{k_n}}(\widehat{\bp}_n^s\, \vert \, \bT_n(k_n))   + 2 s \log \left( k_n  \right) + s \log \left(\frac{ r}{2 \pi ( r-s) }  \right), \; 
\end{align*}
for $s = 1, \ldots, r-1$ and an estimator for $s^*$ is $\widehat{s}_n \coloneqq \argmin_{1 \le s \le r-1} \BICU_{k_n}(s)$.
\end{definition}

Motivated by the $\BICU$, which is based on the largest eigenvalue $\lambda_{n,1}$ from \Cref{th:BIC_lokal_Eigenwerte_Absch}, we define a BIC based on a lower bound for the posterior distribution by using the smallest eigenvalue $\lambda_{n,2} = k_n/T_{n,1}(k_n)$ from \Cref{th:BIC_lokal_Eigenwerte_Absch}. 

\begin{definition} \label{def:local_BIC_lower}
    For the number of extreme directions $s$ with fixed $k_n$ the \textit{Bayesian information criterion concerning the lower bound} ($\BICL$) for Model $M_{k_n}^s$ is defined as
\begin{align*}
    \BICL_{k_n}(s) &\coloneqq  -2 \log L_{M^s_{k_n}}(\widehat{\bp}_n^s\, \vert \, \bT_n(k_n))   + s \log \left( k_n  \right) + s \log \left(\frac{ k_n}{2 \pi T_{n,1}(k_n)}  \right), \; s = 1, \ldots, r,
\end{align*}
and an estimator for $s^*$ is $\widehat{s}_n \coloneqq \argmin_{1 \le s \le r} \BICL_{k_n}(s)$.
\end{definition}

\begin{theorem} \label{th:BICU_Consistency}
Suppose \Cref{asu:directions,Assumption:main} holds. Then 
 \begin{enumerate}[(a)]
     \item  ${\displaystyle 
    \limn \P( \BICU_{k_n}(s) > \BICU_{k_n}(s^*) ) = 1 \quad  \text{ for } s \ne {s^*},}$
\item  ${\displaystyle 
    \limn \P( \BICL_{k_n}(s) > \BICL_{k_n}(s^*) ) = 1 \quad  \text{ for } s \ne {s^*}.} $
 \end{enumerate}

\end{theorem}
Thus, in contrast to the $\AIC$ criterion,  both information criteria are weakly consistent and select asymptotically with probability $1$ the true Model $M_{k_n}^{s^*}$. 
This is also a typical property of Bayesian information criteria (see \citet{ModelSelection, C:16}).

\subsection{Bayesian information criterion for the threshold \texorpdfstring{$k$\textsubscript{$n$}}{kn}} \label{sec:BIC_Global}

In the following, we determine an upper bound for the posterior probability of the global Model $M_n^{\prime s}$ analog to the previous Section~\ref{sec:localBIC} using the following assumptions. 

\begin{assumptionletter}~ \label{asu:BIC_global}
Suppose the following statements hold. 
\begin{enumerate}[(D1)]
    \item There exist constants $0 < b' \le B' < \infty$ such that the prior density $g'(\,\cdot\, \, \vert \, M^{\prime s}_{n})$ on $\Theta_s'$ satisfies  \label{asu:BIC_global_prior_density}
     \begin{equation*}
         b' \le g'(\widetilde{\bp}^{\prime s} | M^{\prime s}_n) \le B' \quad \text{for all } \widetilde{\bp}^{\prime s} \in \Theta_s^{\prime}. 
     \end{equation*}
\item The prior distribution $\Q'$  is a uniform distribution on $\{M^{\prime s}_n:\,s=1,\ldots,r\}$,  i.e.  $\Q'(M^{\prime s}_n) = \frac1r$ for $s = 1, \ldots, r$.  \label{asu:BIC_global_prior_density2}
    \item $\lim\limits_{n\to\infty} n q_n^{5/3} = \infty$ and $\lim\limits_{n\to\infty} n q_n^2 = 0$. \label{asu:BIC_conv_nqn53}     
 \item \label{E1} For $\E_{\lambda}  [ L_{M^s_{n-T_{n,2^d}'}} ( \widetilde{\bp}^s \, \vert \, \bT_{n, \{1, \ldots, r \}}') ] \coloneqq  \int_{ \Theta_s}  L_{M^s_{n-T_{n,2^d}'}} ( \widetilde{\bp}^s \, \vert \, \bT_{n, \{1, \ldots, r \}}')   \di \widetilde{\bp}^s$ 
 the following upper bound 
    \begin{align*}
   \E   \Bigl[  -2 \log& \E_{\lambda}  [ L_{M^s_{n-T_{n,2^d}'}} ( \widetilde{\bp}^s \, \vert \, \bT_{n, \{1, \ldots, r \}}') ]  \Bigr] \notag\\
    \leq \, &  \, \E \Big[  \E \Big[ -2 \log L_{M^s_{n-T_{n,2^d}'}} ( \widehat{\bp}_n^s(\bT_{n, \{1, \ldots, r \}}') \, \vert \, \bT_{n, \{1, \ldots, r \}}') \Big| \td  \Big]  \Big]\nonumber\\
    & \, + 2 s \E \Big[  \log \left( (n - \td ) \sqrt{ \frac{r}{r-s}} \right)  \Big] - s \log(2 \pi)   + o(1) 
\end{align*}
holds.
\end{enumerate}
\end{assumptionletter}

\begin{remark} $\mbox{}$ 
\begin{itemize} 
    \item[(a)] Assumptions (\ref{asu:BIC_global}\ref{asu:BIC_global_prior_density}) and (\ref{asu:BIC_global}\ref{asu:BIC_global_prior_density2}) in the global model correspond to the Assumptions (\ref{asu:BIC_local}\ref{asu:BIC_local3}) and (\ref{asu:BIC_local}\ref{asu:BIC_local4}) in the local model.  Assumption  (\ref{asu:BIC_global}\ref{asu:BIC_conv_nqn53}) is the counterpart to Assumption (\ref{asu:BIC_local}\ref{asu:BIC_local1}) for the binomial part of the likelihood function in the global model.
     \item[(b)] Assumption (\ref{asu:BIC_global}\ref{asu:BIC_conv_nqn53}) %
     ensures a suitable convergence rate of $q_n$ and implies $n q_n \rightarrow \infty$. For example $q_n \coloneqq n^{-11/20}$ fulfills the conditions of Assumption (\ref{asu:BIC_global}\ref{asu:BIC_conv_nqn53}).
     \item[(c)]  Assumption  (\ref{asu:BIC_local}\ref{asu:BIC_local1}) for the local model is
     required for the proof of  \Cref{th:BIC_post_prob}. 
     Assumption (\ref{asu:BIC_global}\ref{E1}) for the global model is motivated from  \Cref{th:BIC_post_prob} and \eqref{2.4}. Because we then  obtain directly 
     \begin{align*}
   \E   \Bigl[  -2 \log& \E_{\lambda}  [ L_{M^s_{n-T_{n,2^d}'}} ( \widetilde{\bp}^s \, \vert \, \bT_{n, \{1, \ldots, r \}}') ] \vert T_{n,2^d}'=k_n \Bigr] \notag\\
    \leq \, &  \, \E \Big[   -2 \log L_{M^s_{n-T_{n,2^d}'}} ( \widehat{\bp}_n^s(\bT_{n, \{1, \ldots, r \}}') \, \vert \, \bT_{n, \{1, \ldots, r \}}') \Big|   T_{n,2^d}'=k_n \Big]\nonumber\\
    & \, + 2 s \E \Big[  \log \left( (n - \td ) \sqrt{ \frac{r}{r-s}} \right)  \vert T_{n,2^d}'=k_n\Big] - s \log(2 \pi)   + o(1) 
\end{align*}
    for $k_n$ satisfying the assumptions of the previous section and $T_{n,1}'\geq T_{n,2}'\geq\ldots\geq T_{n,r}'$.
    Assumption (\ref{asu:BIC_global}\ref{E1}) for the global model is only a slightly stronger assumption than  Assumption  (\ref{asu:BIC_local}\ref{asu:BIC_local1}) for the local model.     
\end{itemize}
\end{remark}

In analogy to \Cref{sec:localBIC}, the goal is to derive asymptotic bounds for $-2 \log  \P(  M^{\prime s}_{n} |\bT_n' )$ which we obtain through upper bounds for 
\begin{equation} \label{eq:BIC_global_post}
-2   \log \E_{g_s'} [L_{M^{\prime s}_n}( \widetilde{\bp}^{\prime s}|\bT_n') ] \coloneqq -2 \log \left\{ \int_{\Theta_s'}  L_{M^{\prime s}_n} (\widetilde{\bp}^{\prime s} \, \vert \, \bT_n') \cdot g'(  \widetilde{\bp}^{\prime s} | M^{\prime s}_n)  \di \widetilde{\bp}^{\prime s} \right\},
\end{equation}
where $\E_{g_s'}$ denotes the conditional expectation with respect to the prior density $g'(\,\cdot \,| M^{\prime s}_n)$ on $\Theta_s'$ given Model $M_{n}^{\prime s}$.

\begin{theorem} \label{th:BIC_global_bound}
Under Assumptions~(\ref{asu:Model_global}\ref{asu:BIC_global_expectation_local}), (\ref{asu:Model_global}\ref{asu:BIC_glob_qn_kn}) and \ref{asu:BIC_global}  the asymptotic upper bound as $n\to\infty$,
    \begin{align*}
-2   &\E\big[ \log\E_{g_{s}'}  [ L_{M^{\prime s}_n} (\widetilde{\bp}^{\prime s} \, \vert \, \bT_n') ] \big] \\
&\leq n q_n \Biggl( -2 \frac{\E[ \log L_{M^s_{k_n}} ( \widehat{\bp}_n^s \, \vert \, \bT_n(k_n) ) ]}{k_n}      + 2  \frac{s}{n q_n} \log \left( k_n \sqrt{ \frac{r}{ 2 \pi  (r-s)}} \right)   +   \frac{ 2 \log( n) }{n q_n}    + C   \Biggl)
\end{align*}
holds, where  $C > 0$ is a constant independent of $s$ and $n$.
\end{theorem}
    Compared to \Cref{th:BIC_post_prob} in the previous section,  we take additionally the expectation in \Cref{th:BIC_global_bound} to achieve a connection between the global model and the local model.

 \Cref{th:BIC_global_bound} motivates the definition of the following information criterion, where the expectation is omitted, the inequality is divided by $n q_n$ and the term $  2 \log( b') /(nq_n) $ as well as $C$ are neglected as they are either constant concerning $s$ or converge to zero uniformly.

\begin{definition}
For the number of exceedances $k_n$ the \textit{Bayesian information criterion concerning the upper bound ($\BICU$) for the threshold $k_n$} for Model $M^{\prime s}_n$ is defined as
\begin{align*}
    \BICU_{n,s}(k_n) &\coloneqq  \frac{-2 \log L_{M^s_{k_n}} ( \widehat{\bp}_n^s \, \vert \, \bT_n(k_n) ) + 2s \log \left( k_n \right) + s \log \left(  \frac{r}{2 \pi (r-s)} \right) }{k_n} +  \frac{  \log( n^2)  }{k_n}\\
    &\;= \frac{\BICU_{k_n}(s)}{k_n}  +  \frac{ \log( n^2) }{k_n}, \quad 
\end{align*}
for $k_n=1,\ldots,n$,
with estimator $\widehat{k}_n \coloneqq \argmin_{k_n \in K} \bigl\{ \min_{1 \le s \le r}  \BICU_{n,s}(k_n) \bigr\}$ for $K \subset  \{1,\ldots, n\}$ for $k_n$.
\end{definition}

Similarly to \Cref{def:local_BIC_lower} we also define the Bayesian information criterion based on the lower bound for the threshold $k_n$.

\begin{definition}
For the number of exceedances $k_n$ the \textit{Bayesian information criterion concerning the lower bound ($\BICL$) for the threshold $k_n$} for Model $M^{\prime s}_n$ is defined as
\begin{eqnarray*}
    \BICL_{n,s}(k_n) &\coloneqq&   \frac{\BICL_{k_n}(s)}{k_n}  +  \frac{ \log( n^2) }{k_n}, \quad k_n=1,\ldots,n,
\end{eqnarray*}
with estimator $\widehat{k}_n \coloneqq \argmin_{k_n \in K} \bigl\{ \min_{1 \le s \le r}  \BICL_{n,s}(k_n) \bigr\}$ for $K \subset  \{1,\ldots, n\}$.
\end{definition}

\section{Simulation study} \label{sec:NumExp}
In this section, we compare the performance of the different information criteria through a simulation study. 
Therefore, we simulate $n$ times a multivariate regularly varying random vector $\bX$ of dimension $d$.
For the distribution of $\bX$, we distinguish two cases: Either $\bX$ exhibits asymptotic independence (\Cref{sec:asymp_indep}) or asymptotic dependence (\Cref{sec:7.3}); these examples can be found in \citet{meyer_muscle23} as well. In both examples, we estimate  the parameter $s^*$ based on the $n$ observations with the different information criteria: $\AIC, \BICU, \BICL, \MSEIC$ and $\QAIC$, and then estimate the probability vector $\bp = (p_1, \ldots, p_{s^*}, 0 , \ldots, 0)^\top$ by $ \widehat{\bp}_{n,*}^{\, \widehat{s}_n}$ given in \cref{eq:estimator_hellinger}. For  comparison, we run simulations for the local model with $k_n=0.05\cdot  n$ and for the global model with an estimated $k_n$.
Since $r$ is not known we use the estimator $$\widehat r_n=\vert\widehat{\mathcal{S}}_n(\bZ)\vert = \vert\{ \beta \in \Pd:  T_n(C_{\beta}, k_n)  > 0 \}\vert$$
at this point. 
In total, we conducted 500 repetitions with sample sizes $n = 1000$, $5000$, $ 10000$, $20000$.
The code for the following simulation study is available at  \url{https://gitlab.kit.edu/projects/164856}.

\subsection{Error measures}
To quantify the discrepancy between the true distribution $p$ and the estimated distribution $ \widehat{\bp}_{n,*}^{\, \widehat{s}_n}$  in \cref{eq:estimator_hellinger}  we use different measures. We start with the Hellinger distance, which is for discrete probability measures $\mathbb{P}$ and $\mathbb{Q}$ with probabilities $p_1, \ldots, p_m$ and $q_1, \ldots, q_m$ for $m \in \N$ given by  $H(P,Q) \coloneqq  \frac{1}{\sqrt2} \lVert \bp - \bq \rVert_2$ where $\bp=(p_1,\ldots,p_m)^\top$ and $\bq=(q_1,\ldots,q_m)^\top$.
Since our primary goal is the identification of the relevant directions $s^*$, we employ alternative measures. These measures evaluate the validity of a detected direction, without considering the weight assigned to it.\smallskip

To be more precise, the confusion matrix visualizes the performance of an information criterion. Suppose an information criterion gives $\widehat{s}$ as an estimator for the number $s^*$ of true directions of $2^d-1$ possible directions. Then we define the confusion matrix for the different information criteria (IC) 
\begin{center}
\begin{tabular}{ c |  c | c | c }
 & Theoretic direction & No theoretic direction & \#Directions \\ \hline
 IC detects direction& True positive (TP) & False positive (FP) & $\widehat{s}$\\ \hline
IC detects no direction & False negative (FN) & True negative (TN) & $2^d-1 - \widehat{s}$\\ \hline
\#Directions & $s^*$ & $2^d-1 -s^*$  & $2^d-1 $
\end{tabular} 
\end{center}
and as error measures
\begin{align*}
    \text{Accuracy Error} &\coloneqq 1 - \frac{\text{TP} + \text{TN}}{\text{TP} + \text{TN} + \text{FP} + \text{FN}} = \frac{ \text{FP} + \text{FN}}{2^d-1}, \\
    \text{$F_1$ Error} &\coloneqq 1 - \frac{2 \text{TP}}{2 \text{TP} + \text{FP} + \text{FN}} = 1 - \frac{2 \text{TP}}{s^* + \widehat{s} },
\end{align*}
which reflects the errors. 
If we take $1 - \text{Accuracy Error}$ and $1 -  \text{$F_1$ Error}$, respectively, we obtain the original definition in \citet{confusion_matrix}  such that our error measures are negatively oriented and a lower value is better.   
    The Accuracy Error measures the relative number of false classified directions, whereas the $F_1$ Error is the harmonic mean based on the precision and the recall. Note, that the precision error is the relative amount of actual theoretical directions to the number of detected directions whereas the recall gives the proportion of theoretical directions.

\subsection{Asymptotic tail independent model} \label{sec:asymp_indep}

In the first example, we consider $d$-dimensional i.i.d. random vectors 
whose spectral measure only concentrates on the axis. To define their distribution,  we assume that $\bH = (h_{ij})_{1 \leq i,j \leq d} \in \R^{d \times d}$ with  $h_{ij} \overset{ \text{\tiny  i.i.d.}}{\sim} \mathcal{U}((0,1))$ and

\begin{equation*}
\bSigma \coloneqq \diag( h_{11}^{-1/2}, \ldots, h_{dd}^{-1/2}) \cdot \bH^\top \cdot \bH \cdot \diag( h_{11}^{-1/2}, \ldots, h_{dd}^{-1/2}).
\end{equation*}
Note that  $\bSigma_{ii} = 1,\, i = 1, \ldots, d$ and $\bSigma_{ij} < 1, \, i \neq j$. 
Suppose now $\bY=(Y_1,\ldots,Y_d) \sim\mathcal{N}_d ( \mathbf{0}_d, \bSigma)$  under the condition of $\bSigma$ whose components have, by construction, as marginal distribution the standard normal distribution $\Phi$. It is well known that the multivariate normal distribution with correlations smaller than $1$ exhibits pairwise asymptotic independence (\citet{resnick1987}, Corollary 5.28). Now, let 
$\bY^1,\ldots,\bY^{n}$ be an i.i.d. sequence of random vectors with distribution $\bY$ and  define the i.i.d. random vectors $\bX^i = (X^i_1, \ldots, X^i_d)^\top \in \Rd$, $i = 1, \ldots,n$, as 
\begin{equation*}
X^i_j \coloneqq \frac{1}{1 - \Phi ( Y^i_j)}, \quad 1 \leq  j \leq d,
\end{equation*}
which are regularly varying with tail index $\alpha = 1$ and exhibit pairwise asymptotic independence so that the extreme directions are the $s^* = d$ axes. 
For our simulation study, we assume now that $d=s^*=40$;  
the results are presented in  \Cref{fig:asymp_indep_fixed_k}, on the left hand side for the local model with $k_n=0.05\cdot n$
and on the right hand side for the global model.
\begin{figure}
     \centering
     \begin{subfigure}[b]{0.49\textwidth}
         \centering
         \includegraphics[width=\textwidth]{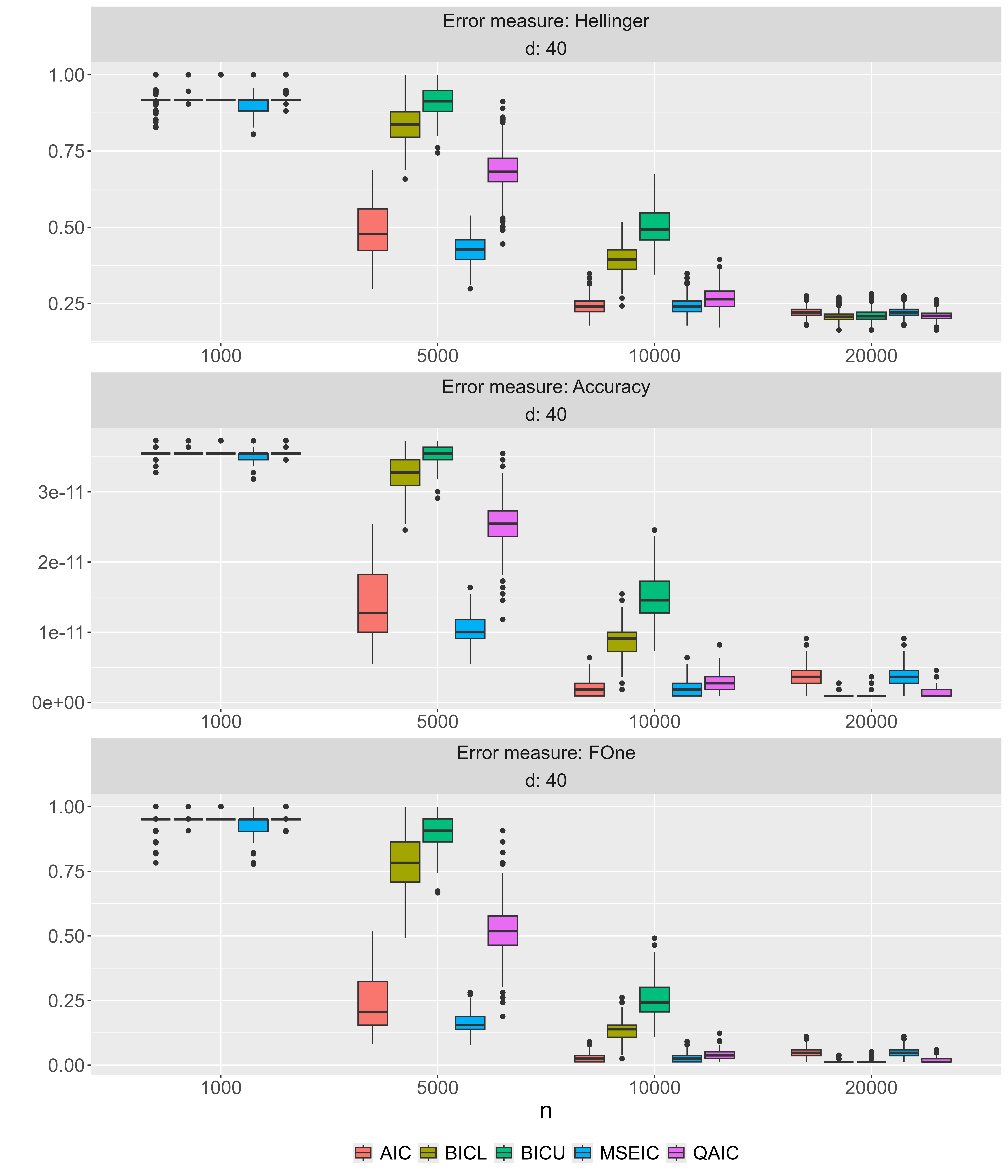}
         \caption{Local model with $k_n / n = 0.05$}
         %\label{fig:y equals x}
     \end{subfigure}
     \hfill
     \begin{subfigure}[b]{0.49\textwidth}
         \centering
         \includegraphics[width=\textwidth]{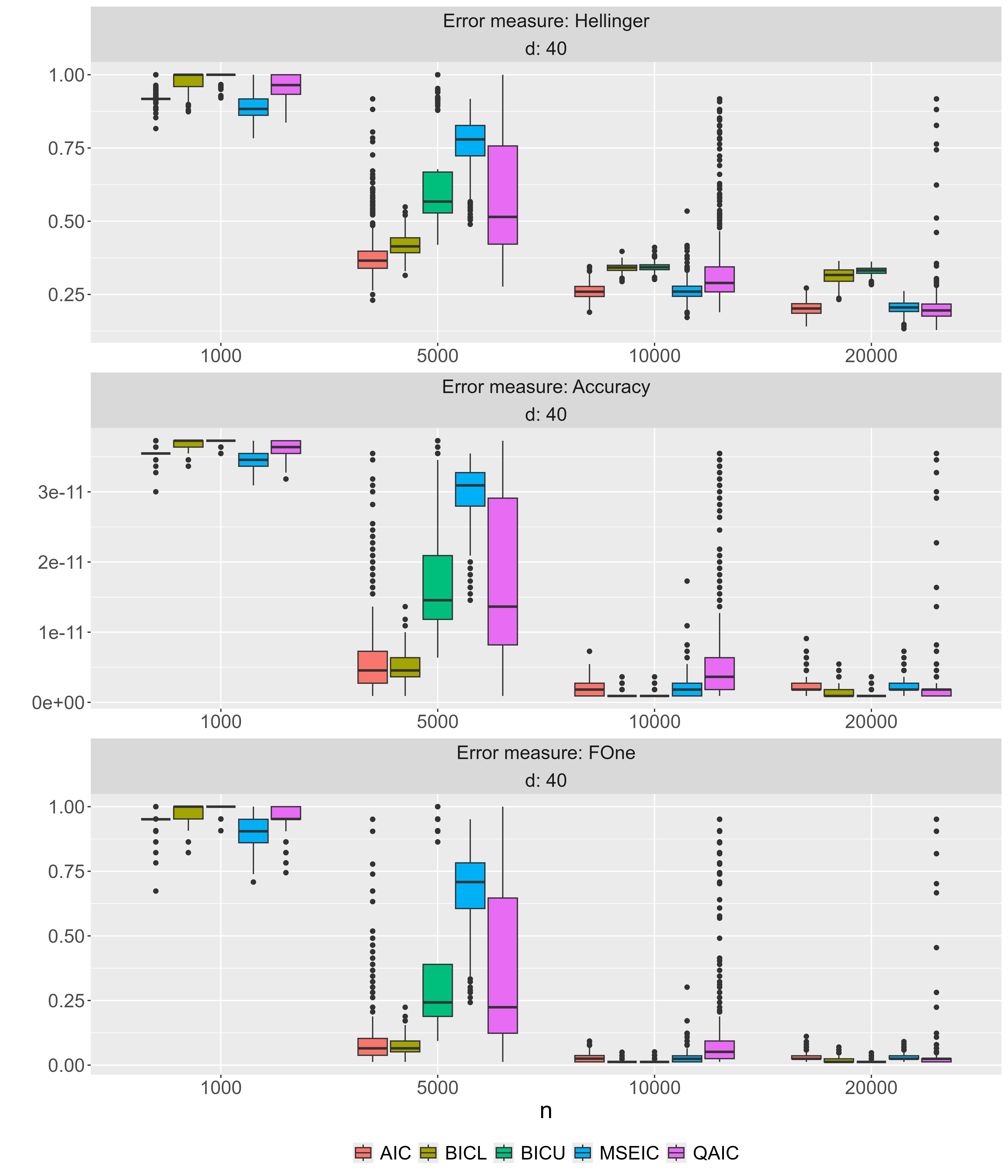}
         \caption{Global model}
         %\label{fig:three sin x}
     \end{subfigure}
         \caption{\footnotesize \textit{Simulations for asymptotically independent data with $s^* = d=40$ directions of extremes:  In the top row we use as error measure the Hellinger distance, in the middle row the Accuracy Error and in the bottom row the $F_1$ Error,  which are plotted against the sample size $n$ on the $x$-axis. 
         }}
    \label{fig:asymp_indep_fixed_k}
\end{figure}
In the local model we see that for small values of $n$, as $n = 5000$ and $n = 10000$, the $\AIC$ and $\MSEIC$ perform better than the other information criteria, while for $n=10000$ the $\QAIC$ performs only slightly worse than the $\AIC$ and the $\MSEIC$. But this changes for  $n=20000$: When evaluating the Accuracy Error and the $F_1$ Error the $\BIC$ and the $\QAIC$ outperform the $\AIC$ and $\MSEIC$. It even seems that  the Accuracy Error and $F_1$ Error of the $\AIC$ and $\MSEIC$ increase, suggesting a tendency toward overfitting, which is in agreement with the theoretical results that the $\AIC$ and $\MSEIC$ are overfitting with a positive probability (\Cref{th:AIC_Cons} and \Cref{th:MSE_Consistency}), whereas the  $\QAIC$ and $\BIC$  are consistent  (\Cref{th:QAIC_Consistency} and \Cref{th:BICU_Consistency}).
If we compare the simulation results for the local model (left part of \Cref{fig:asymp_indep_fixed_k}) with the results for the global model  (right part of \Cref{fig:asymp_indep_fixed_k}),  
we realize that for $n=5000$ and $10000$ the global model of the 
$\AIC$ and $\BIC$ performs better than their corresponding local models, whereas the global model of the $\QAIC$ is, on average, better than its local version, it has many outliers with the tendency to overfit.

\subsection{Asymptotic dependent model} \label{sec:7.3}
Next, we present an additional simulation study for a model with asymptotic dependence which can also be found in \citet{meyer_tail}. Consequently not only directions with $\vb = 1$ are relevant. Let $\bX$ be an $\R^d$ valued random vector and $d_1,d_2,d_3 \in \N \cup \{0 \}$, such that
\begin{align*}
d \geq d_1 + 2 d_2 + 3 d_3.
\end{align*}
The parameters $d_1,d_2,d_3$ specify the number of one, two, and three-dimensional directions. In the following we denote by $\text{Exp}(1)$ the exponential distribution with parameter $1$. The marginal distributions of $\bX$ are defined by
\begin{align*}
X_j & \sim \text{Pareto}(1), \qquad j = 1, \ldots, d_1, \\
(X_j, X_{j+1}) &\sim ( \text{Pareto}(1), X_j + \text{Exp}(1)), \, j = d_1 + 1, d_1 + 3, \ldots, d_1 + 2 \cdot d_2 -1, \\
(X_j,  X_{j+1},  X_{j+2}) & \sim ( \text{Pareto}(1), X_j + \text{Exp}(1), X_j + \text{Exp}(1)), \\ & \qquad j = d_1 + 2 \cdot d_2 + 1, d_1 + 2 \cdot d_2 + 4, \ldots, d_1 + 2 \cdot d_2 + 3  \cdot d_3 - 2,\\
X_j & \sim \text{Exp}(1), \qquad j =  d_1 + 2 \cdot d_2 + 3 \cdot  d_3, \ldots, d.
\end{align*}
The random vector $\bZ$ in \Cref{def:srv} puts mass on the sets
\begin{align*}
&C_{ \{1 \}}, \ldots, C_{ \{d_1 \}},\\
 &C_{ \{d_1 +1, d_1+2 \}}, \ldots, C_{ \{ d_1 + 2 \cdot d_2 -1, d_1 + 2 \cdot d_2 \}}, \\
& C_{ \{d_1 + 2 \cdot d_2 +1,d_1 + 2 \cdot d_2 +2, d_1 + 2 \cdot d_2 +3 \}}, \ldots, C_{ \{d_1 + 2 \cdot d_2 + 3 \cdot  d_3 -2,d_1 + 2 \cdot d_2 + 3 \cdot  d_3-1,d_1 + 2 \cdot d_2 + 3 \cdot  d_3 \}}.
\end{align*}
In total, there are $d_1 + d_2 + d_3$ directions with probability mass, and the goal is again to identify these directions. For the simulation study in  \Cref{fig:asym_dep_fixed_k} we chose $d_1 = 10, d_2 = d_3 = 5$ and $d=50$ resulting in $s^* = 20$ extreme directions. The plots show similar features as for the asymptotic independent case in \Cref{sec:asymp_indep} (cf. \Cref{fig:asymp_indep_fixed_k}).

\begin{figure}
     \centering
     \begin{subfigure}[b]{0.49\textwidth}
         \centering
         \includegraphics[width=\textwidth]{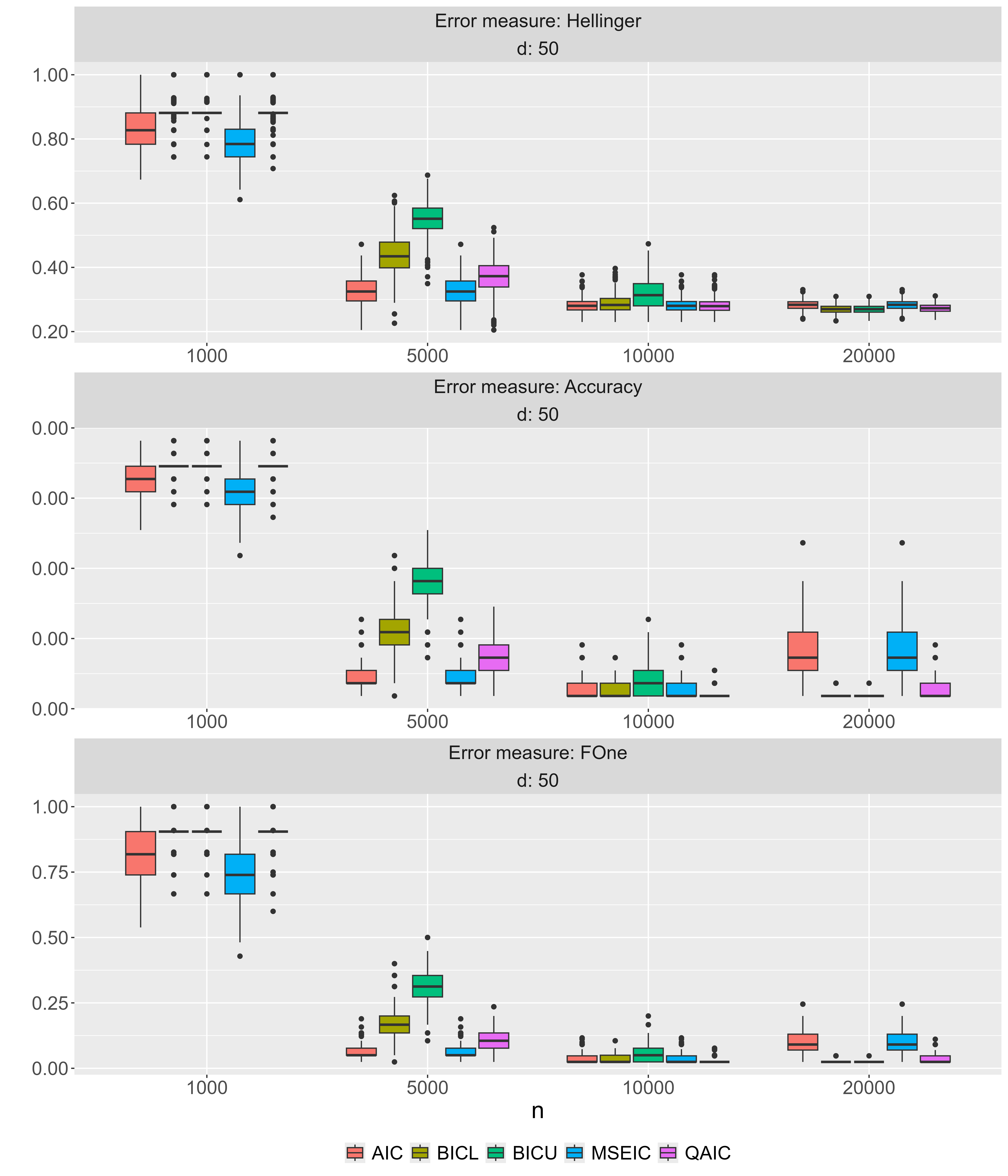}
         \caption{Local model with $k_n / n = 0.05$}
         %\label{fig:y equals x}
     \end{subfigure}
     \hfill
     \begin{subfigure}[b]{0.49\textwidth}
         \centering
         \includegraphics[width=\textwidth]{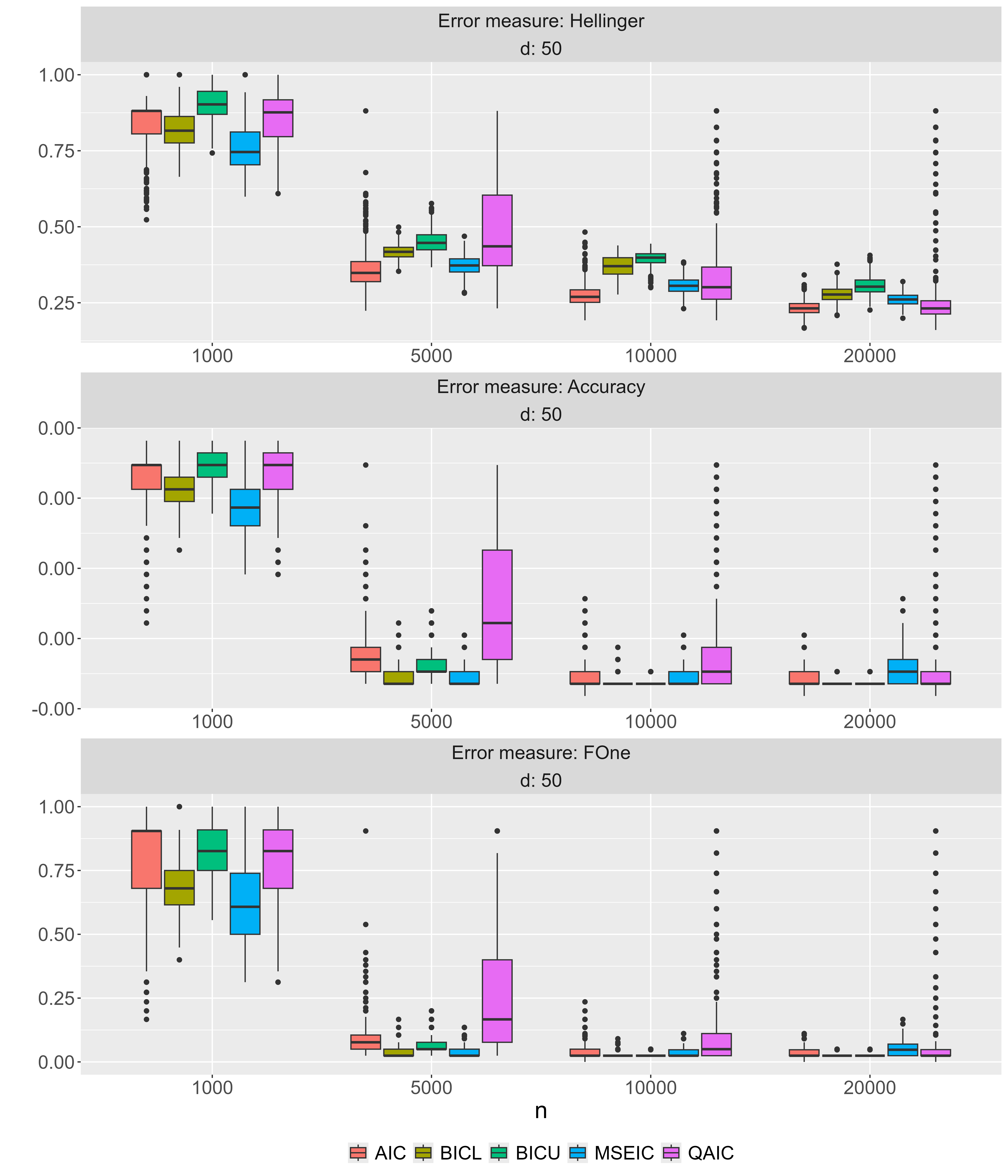}
         \caption{Global model}
         %\label{fig:three sin x}
     \end{subfigure}
    \caption{\footnotesize \textit{Simulations for asymptotic dependent data with $s^* = 20$ directions of extremes and $d=50$:  In the top row we use as an error measure the Hellinger distance, in the middle row the Accuracy Error and in the bottom row the $F_1$ Error, which are plotted on the $y$-axis against the sample size $n$ on the $x$-axis.}}
    \label{fig:asym_dep_fixed_k}
\end{figure}

\section{Application to real-world data} \label{sec:Application}
In this section, we examine the dependence structure of extreme wind speeds using the same example as \citet{meyer_muscle23}. For this purpose, the daily average wind speed at $12$ synoptic meteorological stations in the Republic of Ireland from 1961 until 1978 with \mbox{$n = 6574$} observations are considered. The data was subject to \citet{HR89} and taken from \citet{StatLibData}. To what extent dependencies exist, that are not due to the geographical proximity, will be analyzed in the following. The locations of the stations are shown in \Cref{fig:Wind_map} and consist of:  Belmullet (BEL), Birr (BIR), Claremorris (CLA),
Clones (CLO), Dublin (DUB),  Kilkenny (KIL), Malin Head (MAL), Mullingar (MUL), Roche’s Pt. (RPT), Rosslare (ROS),  Shannon (SHA) and  Valentia (VAL). 
For the preprocessing, we use the same Hill estimator $\hat{\alpha} = 10.7$ as \citet{meyer_muscle23}. We considered values of $k_n $ between $33$ and $1183$.

The values of the estimators for $k_n, k_n/n$, and $s^*$ are presented in \Cref{tab:Wind}. 
\begin{table}[H]
    \centering
\begin{tabular}[t]{lccc}
\toprule
IC & $\hat{k}$ & $\hat{k}/n$ & $\hat{s}$\\
\midrule
$\AIC$ & 460 & 0.07 & 11\\
$\BICU$ & 1118 & 0.17 & 12\\
$\BICL$ & 1118 & 0.17 & 13\\
$\MSEIC$ & 230 & 0.03 & 9\\
$\QAIC$ & 592 & 0.09 & 11\\
\bottomrule
\end{tabular}
\caption{\footnotesize \textit{Estimators for the wind speed data set based on the different information criteria.}}
\label{tab:Wind}
\end{table}

The number of extreme observations $k_n$ varies between $230$ and $1118$, which corresponds to $3 \%$ to $17\%$ of the data. However, the information criteria reported between $9$ and $13$ number of extreme directions, which is not a large range compared to the choice of $k_n$.
On the left-hand side of \Cref{fig:Wind_k_s}, the values of the information criteria are plotted against the threshold $k_n$, while on the right-hand side, the number of estimated directions is mapped as well against $k_n$. The vertical lines indicate the minimum of the information criteria.    It appears that for the number $s$ of extremal directions, there is a more distinct plateau around the optimal value $\hat{k}_n$ for $\BICU, \MSEIC$ and $\QAIC$ compared to $\AIC$ and $\BICL$.

\begin{figure}[ht]
    \centering
    \includegraphics[width=1\textwidth]{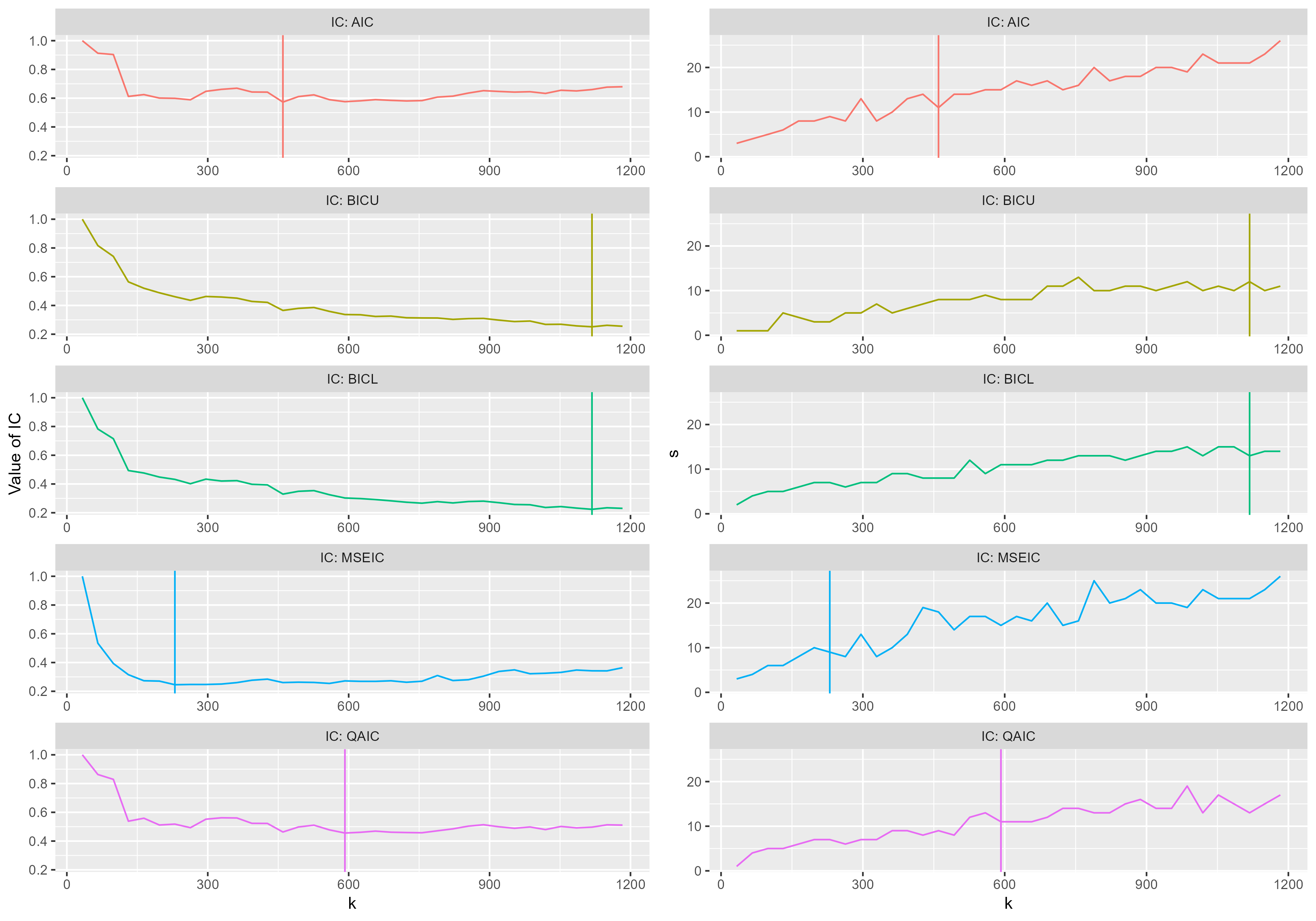}
    \caption{\footnotesize \textit{On the left-hand side in the figure the value of the information criteria (IC) and on the right-hand side, the number $s$ of extremal directions is plotted against  $k_n$. The values of the IC are scaled, such that they start at $1$. The vertical lines indicate the minimum value of the information criteria.}}
    \label{fig:Wind_k_s}
\end{figure}

\begin{figure}[ht]
    \centering
    \includegraphics[width=0.5\textwidth]{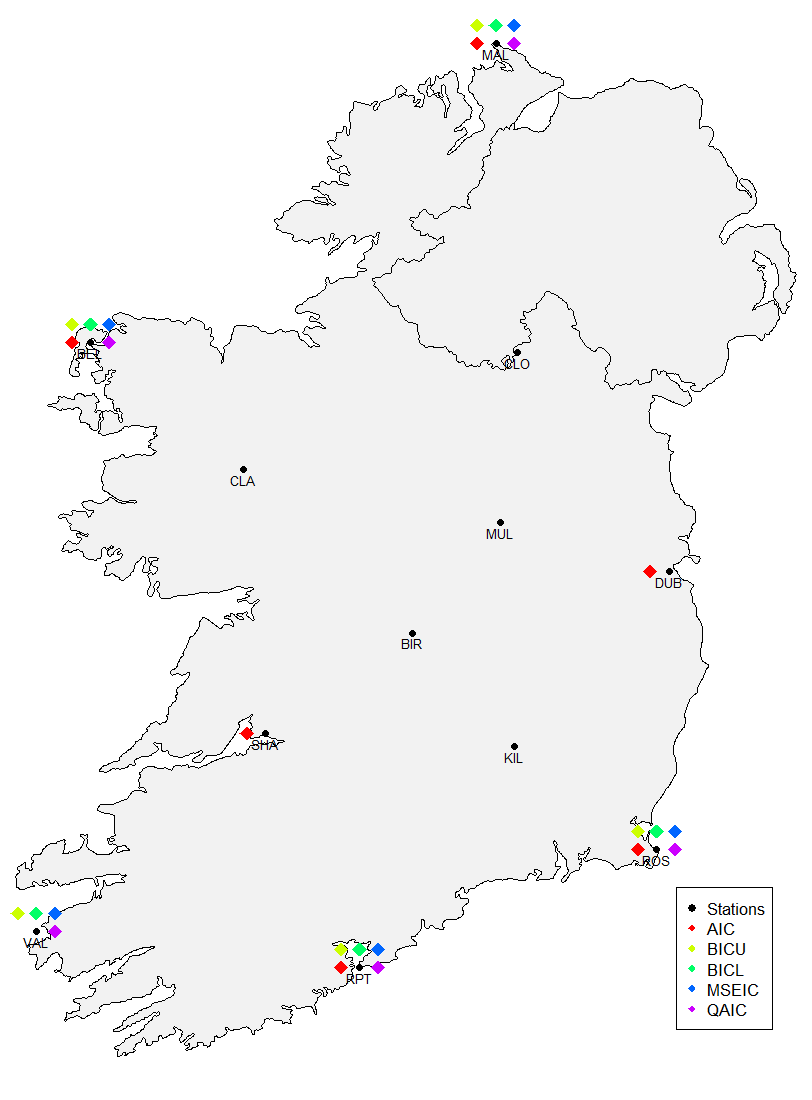}
    \caption{\footnotesize \textit{Maximal subsets recovered by the information criteria of the daily average wind speed.}}
    \label{fig:Wind_map}
\end{figure}

A graphic of the Republic of Ireland is given in \Cref{fig:Wind_map}, where the black dots highlight the different locations of the stations. Colored diamonds close to a station are markers for estimated extreme wind speeds at that station based on an information criterion. All information criteria only identify stations on the coast as extreme, all inland stations have non-extreme wind speeds. 
$\AIC$ missed one station on the coast, which is Valentia located more than $130$ km away from the other stations. $\MSEIC, \QAIC, \BICU$ and $\BICL$ recovered the same maximal clusters and missed the coastal stations Shannon and Dublin. The first station, Shannon, is connected to the ocean but nearly $40$ kilometers away from the open sea. The second station, Dublin, is oriented towards the Irish Sea, rather than the Atlantic Ocean. All information criteria identified Belmullet, Mullingar, Rosslare and Roche’s Pt., and four out of five information criteria also recognized Valentia.

\section{Conclusion} \label{sec:conclusion}
In this paper, we developed three different information criteria for both the number of extreme directions $s^*$ as well as for the choice of the optimal threshold $k_n$. Where the BIC is based on a Bayesian approach for a multinomial model in analogy to the $\AIC$ of \citet{meyer_muscle23}, the $\QAIC$ uses the ideas of an Akaike information criterion, but it is based on a Gaussian likelihood function in comparison to the AIC. In contrast, for $\MSEIC$ no likelihood assumption is necessary; it uses the MSE. The advantage of  $\BICU$, $\BICL$ and $\QAIC$ is that they are weakly consistent information criteria for the number of extreme directions $s^*$, where $\AIC$ and $\MSEIC$ tend to overestimate $s^*$ for large sample sizes,  which we slightly see in the simulation study of the local models for large $n$ but for small $n$ the $\MSEIC$  performs extraordinarily well. 
All information criteria performed quite well, none is particularly superior in all situations.  Finally, the information criteria were successfully applied to a real-world data set, where $\MSEIC$, $\QAIC$, $\BICU$ and $\BICL$ detected the same extreme clusters. In practice, we estimate, of course,  $r$ by $\widehat r_n=\vert\widehat{\mathcal{S}}_n(\bZ)\vert$ and plug this estimate in the information criteria. In this setup, all the consistency results in the paper remain true if we additionally assume that $\sqrt{k_n\rho_n}(\widehat r_n-r)\Pconv 0$ as $n\to\infty$. However, the motivation for the definitions of the information criteria is much clearer when it is assumed that
$r=\vert\widehat{\mathcal{S}}_n(\bZ)\vert$ is deterministic and independent of $n$.

\medskip

%\nosectionappendix 
%\bibliographystyle{imsart-nameyear}
%\bibliography{Informationskriterien_Bib}
%\clearpage

%==================================================================================================
 
%==================================================================================================

\appendix

\section{Proofs}

\subsection{Proofs of Section \ref{sec:QAIC}} \label{sec:proof_QAIC}

\subsubsection{Proof of Theorem \ref{th:AIC_Cons}} \label{proof:AIC_Cons}

\begin{proof}[Proof of Theorem \ref{th:AIC_Cons}]
$\mbox{}$\\
\textbf{Step 1:} Suppose $s > {s^*}$. By the definition of the $\AIC$ and the log-likelihood function in \cref{eq:logLikelihood}  it follows that
\begin{align}
\AIC_{k_n}(s)  &-  \AIC_{k_n}(s^*) \nonumber\\
=&  -  \log L_{M^s_{k_n}}(\widehat{\bp}_n^s\, \vert \, \bT_n(k_n)) + s +  \log L_{M^{s^*}_{k_n}}(\widehat{\bp}_{n}^{s^*}\, \vert \, \bT_n(k_n)) -  {s^*}  \nonumber\\
 =&-  \sum_{j=s^*+1}^{s} T_{n,j}(k_n) \log \left( \frac{T_{n,j}(k_n)}{k_n} \right) -  \log \left( \frac{1}{r-s} \sum_{j=s+1}^r \frac{T_{n,j}(k_n)}{k_n} \right) \sum_{i=s+1}^{r} T_{n,i}(k_n) \nonumber\\
&+ \log \left( \frac{1}{r-{s^*}} \sum_{j={s^*}+1}^r \frac{T_{n,j}(k_n)}{k_n} \right) \sum_{i={s^*}+1}^{r} T_{n,i}(k_n)   + (s - {s^*}),  \label{eq:AIC1_Difference} 
\end{align}
 where we used that $s > s^*$. Inserting the alternative representation
\begin{align*}
    T_{n,j}(k_n) = k_n \rho_n + \sqrt{k_n \rho_n} Y_{n,j} 
\end{align*}
where
\begin{align*}
    Y_{n,j} \coloneqq \sqrt{k_n \rho_n}  \left( \frac{T_{n,j}(k_n)}{\rho_n k_n}  - 1 \right),\quad j = s^*+1, \ldots, r,
\end{align*}
gives that
\begin{align}
\AIC_{k_n}(s) & - \AIC_{k_n}({s^*}) \nonumber\\
=&  -  \sum_{j={s^*}+1}^{s} (k_n \rho_n  + \sqrt{k_n \rho_n} Y_{n,j}) \log \left(  1 + \frac{1}{\sqrt{k_n \rho_n}} Y_{n,j} \right) \nonumber\\
&-  \log \left(1 + \frac{1}{r-s} \sum_{j=s+1}^r  \frac{1}{\sqrt{k_n \rho_n}} Y_{n,j} \right) \sum_{i=s+1}^{r}  (k_n \rho_n  + \sqrt{k_n \rho_n} Y_{n,i}) \nonumber\\
&  + \log \left( 1 + \frac{1}{r-{s^*}} \sum_{j={s^*}+1}^r  \frac{1}{\sqrt{k_n \rho_n}} Y_{n,j} \right) \sum_{i={s^*}+1}^{r} (k_n \rho_n  + \sqrt{k_n \rho_n} Y_{n,i})  \nonumber\\
&+ (s - {s^*}).  \label{eq:AIC3_Difference}
\end{align}
For the asymptotic behavior we apply Assumption~(\ref{asu:directions}\ref{(A4)}) which results in 
\begin{align} \label{eq:AIC_cons_conv}
    (Y_{n,s^*+1}, \ldots, Y_{n,r}) \Dconv (Y_{s^*+1}, \ldots, Y_{r}) \eqqcolon \bY \sim \mathcal{N}_{r-s^*} ( \textbf{0}_{r-s}, \bI_{r-s^*}),  \quad \ninf,
\end{align}
and thus,
\begin{align*}
    Y_{n,i} = O_\P(1) \quad \text{ for } \quad i = s^*+1, \ldots, r.
\end{align*}
This and the Taylor expansion of the logarithm
\begin{align*}
\log(1 + x) = x - \frac{1}{2} x^2 + O(x^3), \quad x \rightarrow 0, 
\end{align*}
we insert in \cref{eq:AIC3_Difference} such that
\begin{align*}
\AIC_{k_n}(s) & - \AIC_{k_n}({s^*}) \nonumber\\
=&  -  \sum_{j={s^*}+1}^{s} (k_n \rho_n  + \sqrt{k_n \rho_n} Y_{n,j})  \left(  \frac{1}{\sqrt{k_n \rho_n}} Y_{n,j} - \frac12 \frac{1}{ k_n \rho_n } Y_{n,j}^2  \right)  \nonumber\\
&-    \bigg( \frac{1}{r-s} \sum_{j=s+1}^r  \frac{1}{\sqrt{k_n \rho_n}} Y_{n,j} - \frac12 \Big( \frac{1}{r-s} \sum_{j=s+1}^r  \frac{1}{\sqrt{k_n \rho_n}} Y_{n,j} \Big)^2\bigg) \nonumber \\
& \hspace{7cm} \cdot \sum_{i=s+1}^{r}  (k_n \rho_n  + \sqrt{k_n \rho_n} Y_{n,i}) \nonumber\\
& + \bigg(  \frac{1}{r-{s^*}} \sum_{j={s^*}+1}^r  \frac{1}{\sqrt{k_n \rho_n}} Y_{n,j} -\frac12 \Big( \frac{1}{r-{s^*}} \sum_{j={s^*}+1}^r  \frac{1}{\sqrt{k_n \rho_n}} Y_{n,j} \Big)^2  \bigg) \nonumber \\
& \hspace{7cm} \cdot   \sum_{i={s^*}+1}^{r} (k_n \rho_n  + \sqrt{k_n \rho_n} Y_{n,i}) \nonumber\\
&+ (s - {s^*})+O_{\P}((k_n\rho_n)^{-1/2}). \nonumber
\end{align*}
Since $k_n\rho_n \rightarrow \infty$ (\Cref{cor:Theorem1_MW}(a)) we receive
\begin{align*}
\AIC_{k_n}(s)  - \AIC_{k_n}({s^*}) 
= & -  \frac{1}{2}\sum_{j={s^*}+1}^{s}  Y_{n,j}^2  -    \frac{1}{2(r-s)}\left(  \sum_{j=s+1}^r  Y_{n,j} \right)^2 \nonumber\\
    &+  \frac{1}{2(r-{s^*})} \left( \sum_{j={s^*}+1}^r   Y_{n,j}  \right)^2 + (s - {s^*})+o_{\P}(1). 
\end{align*}
Due to \cref{eq:AIC_cons_conv}
and the continuous mapping theorem we finally obtain 
as $\ninf$,
\begin{align}
      \AIC_{k_n}(s)  - \AIC_{k_n}({s^*}) 
      \Dconv & -\frac12\sum_{j={s^*}+1}^{s}Y_{j}^2  -   \frac{1}{2(r-s)} \bigg(  \sum_{j=s+1}^r  Y_{j} \bigg)^{\! 2}  \nonumber \\ 
&    +   \frac{1}{2 (r-{s^*})} \bigg( \sum_{j={s^*}+1}^r   Y_{j}  \bigg)^{\! 2}  + (s - {s^*}) . \label{eq:AIC4_Difference2}
\end{align}
Obviously, 
\begin{align*}
     \P& \Biggl( - \frac12 \sum_{j={s^*}+1}^{s}  Y_{j}^2   -   \frac{\left(  \sum_{j=s+1}^r  Y_{j} \right)^2}{2(r-s)}    +   \frac{\left( \sum_{j={s^*}+1}^r   Y_{j}  \right)^2}{2 (r-{s^*})}   + s - {s^*} < 0   \Biggr) >0.
\end{align*}

\textbf{Step 2:} Suppose $s < {s^*}$. We obtain analog to \cref{eq:AIC1_Difference} that
\begin{align}
\AIC_{k_n}(s) & - \AIC_{k_n}({s^*}) \nonumber\\
&= \sum_{j=s+1}^{{s^*}} T_{n,j}(k_n) \log \left( \frac{T_{n,j}(k_n)}{k_n} \right) -  \log \left( \frac{1}{r-s} \sum_{j=s+1}^r \frac{T_{n,j}(k_n)}{k_n} \right) \sum_{i=s+1}^{r} T_{n,i}(k_n) \nonumber\\
&\qquad  + \log \left( \frac{1}{r-{s^*}} \sum_{j={s^*}+1}^r \frac{T_{n,j}(k_n)}{k_n} \right) \sum_{i={s^*}+1}^{r} T_{n,i}(k_n)   + (s - {s^*}). \label{eq:AIC2_Cons_pos}
\end{align}
A direct consequence of $  T_{n,j}(k_n)/k_n \Pconv 0$ for $j = {s^*}+1, \ldots,r$ (\Cref{cor:Theorem1_MW}(c))
and $ \lim_{x \rightarrow 0} x \log(x) = 0$ is that
\begin{align*}
\log \left( \frac{1}{r-{s^*}} \sum_{j={s^*}+1}^r \frac{T_{n,j}(k_n)}{k_n} \right) \sum_{i={s^*}+1}^{r}  \frac{T_{n,i}(k_n)}{k_n} \Pconv 0.
\end{align*}
Furthermore, \Cref{cor:Theorem1_MW}(b) 
 yields $  T_{n,j}(k_n)/k_n \Pconv p_j > 0$ for  $j = 1, \ldots, {s^*}$ and thus, as $\ninf$,
\begin{align}
\sum_{i=s+1}^{{s^*}}  & \frac{T_{n,i}(k_n)}{k_n} \log \left( \frac{T_{n,i}(k_n)}{k_n} \right) -  \log \left( \frac{1}{r-s} \sum_{j=s+1}^r \frac{T_{n,j}(k_n)}{k_n} \right) \sum_{i=s+1}^{r}  \frac{T_{n,i}(k_n)}{k_n} \nonumber\\
&\Dconv \sum_{i=s+1}^{{s^*}}  p_i \left( \log \left( p_i \right) -  \log \left( \frac{1}{r-s} \sum_{j=s+1}^{s^*} p_j \right) \right), \label{conv:AIC2_Cons_pos1}
\end{align}
while we used $p_i=0$ for $s^*\leq i\leq r$.
Next, we apply the log sum inequality (\citet{Log_Sum_ineq}, Theorem 2.7.1) to the limit of \cref{conv:AIC2_Cons_pos1} and receive
\begin{align}
    \sum_{i=s+1}^{{s^*}} & p_i \left( \log \left( p_i \right) -  \log \left( \frac{1}{r-s} \sum_{j=s+1}^{s^*} p_j \right) \right) = \sum_{i=s+1}^{{s^*}}  p_i  \log \left( \frac{p_i}{\frac{1}{r-s} \sum_{j=s+1}^{s^*} p_j } \right) \nonumber\\
    &\geq  \left( \sum_{i=s+1}^{{s^*}}  p_i \right)  \log \left( \frac{\sum_{i=s+1}^{{s^*}}  p_i}{  \frac{s^* - s}{r-s} \sum_{j=s+1}^{s^*} p_j } \right)=  \left( \sum_{i=s+1}^{{s^*}}  p_i \right)  \log \left( \frac{r-s}{s^* -s}  \right) > 0, \label{conv:AIC2_Cons_pos2}
\end{align}
since $r > s^*$.
Dividing  \cref{eq:AIC2_Cons_pos}   by $k_n$ and using   \cref{conv:AIC2_Cons_pos1} and \cref{conv:AIC2_Cons_pos2} gives 
\begin{align*}
\frac{1}{k_n} (\AIC_{k_n}(s) & - \AIC_{k_n}({s^*})) \Dconv \sum_{i=s+1}^{{s^*}}  p_i \left( \log \left( p_i \right) -  \log \left( \frac{1}{r-s} \sum_{j=s+1}^{s^*} p_j \right) \right) > 0,
\end{align*}
and thus, the assertion follows.
\end{proof}

\
\subsubsection{Proof of Proposition \ref{th:QAIC_Likelihood_Approx}} \label{proof:QAIC_Likelihood_Approx}

Before we are able to present the proof of  \Cref{th:QAIC_Likelihood_Approx} we require some auxiliary lemmata whose proofs are moved to \Cref{sec:QAIC_supp} of the Supplementary Material. In the following, we work with the $r$-dimensional multivariate normal distribution  $\mathcal{N}_r( k_n \bB_s(\bpbt^s), k_n \diag(\bB_s(\bpbt^s))),\; \bpbt^s \in \R_+^{s+1},$ which has a negative log-likelihood function
\begin{align*}
-2 \log L_{\mathcal{N}_r }( \bpbt^s \, \vert \, \bcT_n ) 
=& r \log(2 \pi) + r \log(k_n) + \sum_{j=1}^s \log (\pbt^s_{j}) + (r - s) \log( {\rhobt^s}) \\ 
&+ k_n  \bigg( \sum_{j=1}^{s} \frac{1}{\pbt^s_{j}} \Big( \frac{\cT_{n,j} }{k_n} - \pbt^s_{j} \Big)^2   + \sum_{j=s+1}^r \frac{1}{{\rhobt^s}} \Big( \frac{\cT_{n,j} }{k_n} - {\rhobt^s} \Big)^2 \bigg). \nonumber 
\end{align*}

\begin{lemma}\label{th:conv_QAIC_Hessematrix} 
Suppose the assumptions of  \Cref{th:QAIC_Likelihood_Approx} hold and $\bpb^s_n( \bcT_n  )$ is defined analog to $\bpb^s_n( \widetilde{\bcT}_n )$ in \eqref{6.1}. Then as $\ninf$,
\begin{align*}
\bY_n := &   \sqrt{k_n} \diag(p_{n,1}, \ldots, p_{n,s},\frac{\rho_n}{(r-s)}, \rho_n,\ldots,\rho_n)^{-1/2}\begin{pmatrix}
 \left(\bpb^s_n(\widetilde \bcT_n  )-\bpb^s_n( \bcT_n  ) 
     \right)\\
      \left( \frac{\mathcal{T}_{n,s+1} }{k_n} - \rhob^s_{n}( \bcT_n ) \right)\\
     \vdots \\
      \left( \frac{\mathcal{T}_{n,r} }{k_n} - \rhob^s_{n}( \bcT_n ) \right)\\
\end{pmatrix}\\
            \Dconv &\;   
    \mathcal{N}_{{r+1}} \left( \mathbf{0}_{r+1}, \Sigma \right),
\end{align*}
where
\begin{align*}
    \Sigma \coloneqq \begin{pmatrix}
        2 \bI_{s+1} &   \mathbf{0}_{s \times (r-s)}\\
         \mathbf{0}_{(r-s) \times (s+1)}  & \,\, \bI_{r-s}- \frac{\mathbf{1}_{r-s} \mathbf{1}_{r-s}^\top}{r-s} 
    \end{pmatrix}.
\end{align*}    
\end{lemma} 

\begin{lemma} \label{lem:Gradient_Normal_Mult_MLE}
Suppose the assumptions of  \Cref{th:QAIC_Likelihood_Approx} hold and $\bpb^s_n( \bcT_n  )$ is defined analog to $\bpb^s_n( \widetilde{\bcT}_n )$ in \eqref{6.1}. 
\begin{enumerate}[(a)]
    \item 
    Then as $\ninf$,
    \begin{align*}
        \nabla \log L_{\mathcal{N}_r}  ( \bpb^s_n( \bcT_n )   \, \vert \, \bcT_n )( \bpb^s_n( \widetilde{\bcT}_n )    - \bpb^s_n( \bcT_n ))  \Pconv 0.
    \end{align*}
   \item Suppose  $\bar{\bp}_n \coloneqq (\bar{p}_{n,1}, \ldots, \bar{p}_{n,s}, \bar{\rho}_{n})^\top \in \R_+^{s+1}$ satisfies 
$$   \Vert \bar{\bp}_n - \bpb^s_n( \bcT_n )   \Vert  \leq  \Vert   \bpb^s_n( \widetilde{\bcT}_n )  - \bpb^s_n( \bcT_n )   \Vert, \quad n\in\N. $$ Then as $\ninf$,
 \begin{align*}
         ( \bpb^s_n( \widetilde{\bcT}_n )    - \bpb^s_n( \bcT_n ))^\top \Big( \nabla^2 \log L_{\mathcal{N}_r}  ( \bar{\bp}_n \, \vert \, \bcT_n ) &+ k_n \big( \diag(p_{n,1}, \ldots, p_{n,s},\rho_n/(r-s) )^{-1}   \big)  \Big) \\
         & \cdot ( \bpb^s_n( \widetilde{\bcT}_n )    - \bpb^s_n( \bcT_n )) \Pconv 0. 
   \end{align*}
\end{enumerate}
\end{lemma}

\begin{proof}[Proof of \Cref{th:QAIC_Likelihood_Approx}.] 
Using a Taylor expansion of $\log L_{\mathcal{N}_r}  ( \bpb^s_n( \widetilde{\bcT}_n  ) \, \vert \, \bcT_n )$ around $\bpb^s_n( \bcT_n )$ yields the existence of a random vector $ \bar{\bp}_n \coloneqq (\bar{p}_{n,1}, \ldots, \bar{p}_{n,s}, \bar{\rho}_{n})^\top$ with  
\begin{align*}
    \Vert  \bar{\bp}_n - \bpb^s_n( \bcT_n )   \Vert  \leq  \Vert  \bpb^s_n( \widetilde{\bcT}_n ) - \bpb^s_n( \bcT_n )   \Vert
\end{align*}
such that
\begin{align*}
\log  L_{\mathcal{N}_r} & ( \bpb^s_n( \widetilde{\bcT}_n ) \, \vert \, \bcT_n ) \\
= & \log L_{\mathcal{N}_r}  ( \bpb^s_n( \bcT_n ) \, \vert \, \bcT_n )  +   \nabla \log L_{\mathcal{N}_r}  ( \bpb_n^s (\bcT_n )  \, \vert \, \bcT_n ) \big(\bpb_n^s (\widetilde{\bcT}_n )  - \bpb_n^s (\bcT_n  ) \big)\\
&\; + \frac12 \big( \bpb^s_n( \widetilde{\bcT}_n ) - \bpb^s_n( \bcT_n ) \big)^\top  \nabla^2 \log L_{\mathcal{N}_r}  ( \bar{\bp}_n \, \vert \, \bcT_n )  \big( \bpb^s_n( \widetilde{\bcT}_n ) - \bpb^s_n( \bcT_n ) \big).
\end{align*}
Applying \Cref{lem:Gradient_Normal_Mult_MLE} (b) gives
\begin{align*}
\log{} &  L_{\mathcal{N}_r}  ( \bpb^s_n( \widetilde{\bcT}_n ) \, \vert \, \bcT_n )\\ 
= & \log L_{\mathcal{N}_r}  ( \bpb^s_n( \bcT_n ) \, \vert \, \bcT_n ) \\
& - \frac12 \big( \bpb^s_n( \widetilde{\bcT}_n ) - \bpb^s_n( \bcT_n ) \big)^\top   k_n   \diag(p_{n,1}, \ldots, p_{n,s}, \rho_n/(r-s))^{-1}     \big( \bpb^s_n( \widetilde{\bcT}_n ) - \bpb^s_n( \bcT_n ) \big)\\
& +o_{\P}(1).
\end{align*}
Inserting the definition of $\log L_{\mathcal{N}_r}  ( \bpb^s_n( \bcT_n ) \, \vert \, \bcT_n )$ and $\pb^s_{n,j}( \bcT_n ) = \frac{\cT_{n,j}}{k_n}$, $j=1,\ldots,s$, yield
\begin{align*}
 \log{} &L_{\mathcal{N}_r} ( \bpb^s_n( \widetilde{\bcT}_n ) \, \vert \, \bcT_n ) \\
&=  -\frac12 r \log(2 \pi k_n)    - \frac12 \sum_{j=1}^s \log (\pb^s_{n,j}( \bcT_n ))   -\frac12  (r - s) \log( \rhob^s_{n}( \bcT_n )) \\
& \quad  -\frac12  k_n \sum_{j=s+1}^r \frac{1}{{\rhob_n^s(\bcT_n )}} \left( \frac{\cT_{n,j} }{k_n} - \rhob^s_{n}( \bcT_n ) \right)^2   \\
&\quad - \frac12 \big( \bpb^s_n( \widetilde{\bcT}_n ) - \bpb^s_n( \bcT_n ) \big)^{\! \top} \!  k_n  \diag(p_{n,1}, \ldots, p_{n,s},\rho_n/(r-s) )^{-1}   
 \big( \bpb^s_n( \widetilde{\bcT}_n ) - \bpb^s_n( \bcT_n ) \big)\\
 &\quad+o_{\P}(1).
\end{align*}
Next, we move some terms on the right-hand side and use ${\rho_n}/{{\rhob_n^s(\bcT_n )}} \Pconv 1$ (cf. \Cref{cor:Theorem1_MW} (c) and the assumption $s\geq s^*$)
and $\bY_n$ as defined in \Cref{th:conv_QAIC_Hessematrix}, which result in
\begin{align*}
\log & L_{\mathcal{N}_r}  ( \bpb^s_n( \widetilde{\bcT}_n ) \, \vert \, \bcT_n )   +\frac12 r \log(2 \pi k_n)   + \frac12 \sum_{j=1}^s \log (\pb^s_{n,j}( \bcT_n ))   +\frac12  (r - s) \log( \rhob^s_{n}( \bcT_n ))\\
= & -\frac12  k_n \sum_{j=s+1}^r \frac{1}{{\rho}_n} \left( \frac{\cT_{n,j} }{k_n} - \rhob^s_{n}( \bcT_n ) \right)^2  \\
& - \frac12 \big( \bpb^s_n( \widetilde{\bcT}_n )   - \bpb^s_n( \bcT_n ) \big)^\top  \cdot  k_n   \diag(p_{n,1}, \ldots, p_{n,s},\rho_n/(r-s))^{-1} \big( \bpb^s_n( \widetilde{\bcT}_n ) - \bpb^s_n( \bcT_n ) \big) \\
&     +o_{\P}(1)\\
= &-\frac12 \bY_n^\top\bY_n +o_{\P}(1)\\%\\
 \Dconv &  - \frac12 \bY^\top \bY
\end{align*}
by \Cref{th:conv_QAIC_Hessematrix},
where $\bY \sim  \mathcal{N}_{{r+1}} \left( \mathbf{0}_{r+1}, \Sigma \right)$. Since $$\E[-  \frac12 \bY^\top \bY] = - \frac12\E[\text{trace}(\bY^\top \bY)] = -  \frac12 \text{trace}(\Sigma) = - \frac{r+ s+1}{2}$$ the assertion follows.

\end{proof}

\subsubsection{Proof of Theorem \ref{th:QAIC_Consistency}} \label{proof:QAIC_Consistency}

\begin{proof}[Proof of Theorem \ref{th:QAIC_Consistency}] $\mbox{}$\\
   \textbf{Step 1:} Suppose $s < {s^*}$. 
    We have $\widehat{p}^s_{n,j} = \widehat{p}^{s^*}_{n,j}$ for $ j = 1, \ldots, s$  and due to \Cref{cor:Theorem1_MW}(b,c)  we have  as $\ninf$,
        \begin{align*}
       {\widehat{p}^{{s^*}}}_{n,j}  \Pconv p_j  > 0, \quad j = 1, \ldots, s^*,
    \end{align*}
    and similarly ${\widehat{\rho}^{s}}_n \Pconv  \frac{1}{r - s} \sum_{j=s +1}^{s^*} p_{j} > 0$ as well as ${\widehat{\rho}^{{s^*}}}_n \Pconv  0$.
    Thus,
\begin{align*}
    -\sum_{j=s+1}^{s^*} \log ( {\widehat{p}^{s^*}}_{n,j}) + (r - s) \log(  {\widehat{\rho}^{s}}_n) \Pconv -\sum_{j=s+1}^{s^*} \log ( p_{j}) + (r - s) \log \Big( \frac{1}{r - s} \sum_{j=s +1}^{s^*} p_{j} \Big)
\end{align*}
and
 $    \log(  {\widehat{\rho}^{{s^*}}}_n)  \Pconv - \infty.$
Therefore, we have as $\ninf$,   
    \begin{align*}
        \QAIC_{k_n}&(s)  - \QAIC_{k_n}(s^*) \\
        =&   -\sum_{j=s+1}^{s^*} \log ( {\widehat{p}^{s^*}}_{n,j}) + (r - s) \log(  {\widehat{\rho}^{s}}_n)  - (r - {s^*})  \log(  {\widehat{\rho}^{{s^*}}}_n)   +(s-{s^*})  \\
        \Pconv&\;\infty.
    \end{align*}
 \textbf{Step 2:} Suppose $s > {s^*}$. 
In this case, we have by \Cref{cor:Theorem1_MW}(b,c)  that 
        \begin{align*}
        \frac{{\widehat{p}^{s}}_{n,j}}{\rho_n}  \Pconv  1 , \quad j = {s^*}+1, \ldots, s,
    \end{align*}
    and similarly $ {\widehat{\rho}^{s}}_n / \rho_n \Pconv  1$ as well as $ {\widehat{\rho}^{ {s^*}}}_n / \rho_n  \Pconv  1$.
Hence, with the continuous mapping theorem we receive $ \log (\widehat{p}^{s}_{n,j}/{\rho_n})  \Pconv 0$  for $j = {s^*}+1, \ldots, s $, $\log({\widehat{\rho}^{s^*}_n}/{\rho_n}) \Pconv 0$ and $\log ({\widehat{\rho}^{s}}_n/{\rho_n}) \Pconv 0$ as $\ninf$. 
Thus, as $\ninf$,
    \begin{align*}
        \QAIC_{k_n}(s) & - \QAIC_{k_n}(s^*) \\
        &=  \sum_{j={s^*}+1}^s \log \left( \frac{{\widehat{p}^{s}}_{n,j} }{\rho_n}\right) + (r - s) \log \left(  \frac{{\widehat{\rho}^{s}}_n}{\rho_n}\right) - (r - {s^*}) \log \left(  \frac{{\widehat{\rho}^{{s^*}}}_n}{\rho_n}\right) +(s-{s^*})\\
        &\Pconv  s-{s^*}> 0, 
    \end{align*}
    which gives the statement.
\end{proof}

\subsection{Proofs of Section \ref{sec:MSEIC}} \label{sec:proof_MSEIC}
 \subsubsection{Proof of Theorem \ref{conv:MSE_Chisq}}\label{proof:MSE_Chisq}
The proof of \Cref{conv:MSE_Chisq} is similar to the proof of \Cref{th:QAIC_Likelihood_Approx}. In the first step, we start to calculate the Jacobian vector of $\ell^2 \big(\bpbt^s \vert \bT_n(k_n) \big) $   for $\bpbt^s = ( \pbt^s_1, \ldots \pbt^s_s, \rhobt^s) \in \R_+^{s+1}$, which is
\begin{align*}
\nabla \ell^2 \big( \bpbt^s \, \vert \, \bT_n(k_n)  \big)  =&  k_n \Bigg(\frac{(\pbt^s_{1})^2 - \frac{T_{n,1}(k_n) ^2}{k_n^2}}{(\pbt^s_{1})^2}, \ldots, \frac{(\pbt^s_{s})^2 - \frac{ T_{n,s}(k_n) ^2}{k_n^2}}{(\pbt^s_{s})^2},  \sum_{j=s+1}^r \frac{(\rhobt^s)^2 - \frac{ T_{n,j}(k_n) ^2}{k_n^2}}{(\rhobt^s)^2}  \Bigg) 
\end{align*}
and the Hessian matrix is
\begin{align*}
    \nabla^2 \ell^2 \big(\bpbt^s \, \vert \, \bT_n(k_n)  \big)  
    &= 2  \diag \Bigg(   \frac{ T_{n,1}(k_n)^2}{k_n (\pbt_{1}^s)^3} , \ldots,    \frac{ T_{n,s}(k_n)^2}{k_n (\pbt_{s}^s)^3},  \sum_{j=s+1}^r \frac{T_{n,j}(k_n)^2}{k_n (\rhobt^s)^3}   \Bigg).
\end{align*}

Analog to \Cref{th:conv_QAIC_Hessematrix}  and  \Cref{lem:Gradient_Normal_Mult_MLE} we get the following results.
\begin{lemma}\label{th:conv_MSEIC_Hessematrix} 
Suppose \Cref{Assumption:main} holds, $s\geq s^*$ and $\bpb^s_n( \bT_n(k_n)  )$ is defined analogously to $\bpb^s_n( \widetilde{\bT}_n(k_n) )$ in \eqref{6.1}. Then as $\ninf$,
\begin{eqnarray*}
 \mathbf{U}_n &\coloneqq & \sqrt{k_n} \diag \bigg(p_{n,1}, \ldots, p_{n,s},\frac{\rho_n}{(r-s)} \bigg)^{-1/2}
 \left(\bpb^s_n(\widetilde \bT_n(k_n)  )-\bpb^s_n( \bT_n(k_n)  ) 
     \right) \\
    & \Dconv &    \mathcal{N}_{{s+1}} \left( \mathbf{0}_{s+1}, \begin{pmatrix}
        2 ( \bI_s - \sqrt{\bp_{  \{ 1, \ldots, s   \} }} \sqrt{\bp_{  \{1, \ldots, s\}  }}^\top ) & \mathbf{0}_s \\
        \mathbf{0}_s^\top & 2
    \end{pmatrix}\right).
\end{eqnarray*}
\end{lemma}

\begin{lemma} \label{lem:Gradient_Normal_MSEIC}
Suppose \Cref{Assumption:main} holds,  $s\geq s^*$ and $\bpb^s_n( \bT_n(k_n)  )$ is defined analogously to $\bpb^s_n( \widetilde{\bT}_n(k_n) )$ in \eqref{6.1}. 
\begin{enumerate}[(a)]
    \item 
    Then as $\ninf$,
    \begin{align*}
        \nabla \ell^2  \big( \bpb^s_n( \bT_n(k_n) )   \, \vert \, \bT_n(k_n) \big) \big( \bpb^s_n( \widetilde{\bT}_n(k_n) )    - \bpb^s_n( \bT_n(k_n) ) \big)  \Pconv 0.
    \end{align*}
   \item Suppose  $\bar{\bp}_n \coloneqq (\bar{p}_{n,1}, \ldots, \bar{p}_{n,s}, \bar{\rho}_{n})^\top \in \R_+^{s+1}$ satisfies 
$$   \Vert \bar{\bp}_n - \bpb^s_n( \bT_n(k_n) )   \Vert  \leq  \Vert   \bpb^s_n( \widetilde{\bT}_n(k_n) )  - \bpb^s_n( \bT_n(k_n) )   \Vert, \quad n\in\N. $$ Then as $\ninf$,
 \begin{align*}
         \big( \bpb^s_n( \widetilde{\bT}_n(k_n) )    - & \bpb^s_n( \bT_n(k_n)  )  \big)^\top \\
         & \cdot \Bigg( \nabla^2 \ell^2  \big( \bar{\bp}_n \, \vert \, \bT_n(k_n) \big) - 2 k_n   \diag\bigg(p_{n,1}, \ldots, p_{n,s},\frac{\rho_n}{{(r-s)}} \bigg)^{-1}   \Bigg) \\
         & \cdot \big( \bpb^s_n( \widetilde{\bT}_n(k_n) )    - \bpb^s_n( \bT_n(k_n) ) \big) \Pconv 0. 
   \end{align*}
\end{enumerate}
\end{lemma}

 \begin{proof}[Proof of Theorem \ref{conv:MSE_Chisq}]
Using a Taylor expansion of $\ell^2 ( \bpb^s_n( \widetilde{\bT}_n(k_n)  ) \, \vert \, \bT_n(k_n) )$ around $\bpb^s_n( \bT_n(k_n) )$ yields the existence of a random vector $ \bar{\bp}_n \coloneqq (\bar{p}_{n,1}, \ldots, \bar{p}_{n,s}, \bar{\rho}_{n})^\top$ with  
\begin{align*}
    \Vert  \bar{\bp}_n - \bpb^s_n( \bT_n(k_n) )   \Vert  \leq  \Vert  \bpb^s_n( \widetilde{\bT}_n(k_n) ) - \bpb^s_n( \bT_n(k_n) )   \Vert
\end{align*}
such that
\begin{align*}
\ell^2  ( &\bpb^s_n ( \widetilde{\bT}_n(k_n) ) \, \vert \, \bT_n(k_n) ) \\
= & \ell^2  ( \bpb^s_n( \bT_n(k_n) ) \, \vert \, \bT_n(k_n) )  +   \nabla \ell^2  ( \bpb_n^s (\bT_n(k_n) ) \, \vert \, \bT_n(k_n) ) \big(\bpb_n^s (\widetilde{\bT}_n(k_n) )  - \bpb_n^s (\bT_n(k_n)  ) \big)\\
&+ \frac12 \big( \bpb^s_n( \widetilde{\bT}_n(k_n) ) - \bpb^s_n( \bT_n(k_n) ) \big)^\top  \nabla^2 \ell^2  \big( \bar{\bp}_n \, \vert \, \bT_n(k_n) \big)  \big( \bpb^s_n( \widetilde{\bT}_n(k_n) ) - \bpb^s_n( \bT_n(k_n) ) \big).
\end{align*}
Applying \Cref{lem:Gradient_Normal_MSEIC} gives
\begin{align*}
\ell^2  & ( \bpb^s_n( \widetilde{\bT}_n(k_n) ) \, \vert \, \bT_n(k_n) ) \\
&=   \ell^2  ( \bpb^s_n( \bT_n(k_n) ) \, \vert \, \bT_n(k_n) )  +  \big( \bpb^s_n( \widetilde{\bT}_n(k_n) ) - \bpb^s_n( \bT_n(k_n) ) \big)^\top  \\
& \qquad \cdot k_n   \diag(p_{n,1}, \ldots, p_{n,s}, \rho_n/(r-s))^{-1}     \big( \bpb^s_n( \widetilde{\bT}_n(k_n) ) - \bpb^s_n( \bT_n(k_n) ) \big) +o_{\P}(1)\\
&=   k_n \sum_{j=s+1}^r \frac{1}{{\rhob_n^s(\bT_n(k_n) )}} \left( \frac{ T_{n,j}(k_n) }{k_n} - \rhob^s_{n}( \bT_n(k_n) ) \right)^2  + \big( \bpb^s_n( \widetilde{\bT}_n(k_n) ) - \bpb^s_n( \bT_n(k_n) ) \big)^\top \\
&\qquad  \cdot  k_n  \diag(p_{n,1}, \ldots, p_{n,s},\rho_n/(r-s) )^{-1}   
 \big( \bpb^s_n( \widetilde{\bT}_n(k_n) ) - \bpb^s_n( \bT_n(k_n) ) \big) +o_{\P}(1).
\end{align*}
Next, we move some terms on the right-hand side and use \Cref{th:conv_MSEIC_Hessematrix}, which result in
\begin{eqnarray*}
\lefteqn{\ell^2 ( \bpb^s_n( \widetilde{\bT}_n(k_n) ) \, \vert \, \bT_n(k_n) )  -  k_n \sum_{j=s+1}^r \frac{1}{{\rho}_n} \left( \frac{ T_{n,j}(k_n) }{k_n} - \rhob^s_{n}( \bT_n(k_n) ) \right)^2 }\\
&= &   \big( \bpb^s_n( \widetilde{\bT}_n(k_n) ) - \bpb^s_n( \bT_n(k_n) ) \big)^\top   k_n   \diag(p_{n,1}, \ldots, p_{n,s},\rho_n/(r-s))^{-1}    \\
&&  \cdot \big( \bpb^s_n( \widetilde{\bT}_n(k_n) ) - \bpb^s_n( \bT_n(k_n) ) \big)+o_{\P}(1)\\
&=&\mathbf{U}_n^\top \mathbf{U}_n +o_{\P}(1)\\%\\
& \Dconv &   \mathbf{U}^\top \mathbf{U},
\end{eqnarray*}
as $n\to\infty$, where $ \mathbf{U}^\top \mathbf{U} \sim  2 \chi^2_{s} $. Since $\E[   \mathbf{U}^\top \mathbf{U}] =    2s$ the assertion follows.
 \end{proof}

\subsubsection{Proof of Theorem \ref{th:MSE_Consistency}} \label{proof:MSE_Consistency}
\begin{proof}[Proof of Theorem \ref{th:MSE_Consistency}] $\mbox{}$ \\
\textbf{Step 1:} Suppose $s < {s^*}$.
An application of \Cref{cor:Theorem1_MW}(b,c) gives  on the one hand,
\begin{align*}
    \frac{1}{\sum_{l=s+1}^r  \frac{T_{n,l}(k_n)}{k_n(r-s)}} & \sum_{j=s+1}^r  \left(\frac{T_{n,j}(k_n)}{k_n} - \sum_{i=s+1}^r \frac{T_{n,i}(k_n)}{k_n(r-s)} \right) ^2 \\
    &\Pconv  \frac{1}{\sum_{l=s+1}^{s^*} \frac{p_{l}}{r-s}} \sum_{j=s+1}^{s^*} \left( p_{j} - \sum_{i=s+1}^{s^*} \frac{p_{i}}{r-s} \right) ^2,
\end{align*}
where we already applied that $p_j=0$ for $j=s^*,\ldots,r$.
Moreover, 
\begin{align*}
    p_{s+1} - \sum_{i=s+1}^{s^*} \frac{p_{i}}{r-s} \geq p_{s+1} - \frac{{s^*}-{s}}{r-s}  p_{s+1} = \frac{r - {s^*}  }{r-s} p_{s+1} > 0.
\end{align*}
Hence,
\begin{align} \label{eq:MSE_cons_Term1}
    \frac{k_n}{\sum_{l=s+1}^r \frac{T_{n,l}(k_n)}{k_n(r-s)}} \sum_{j=s+1}^r \left(\frac{T_{n,j}(k_n)}{k_n} - \sum_{i=s+1}^r \frac{T_{n,i}(k_n)}{k_n(r-s)} \right) ^2 \Pconv \infty.
\end{align}
On the other hand, define 
\begin{eqnarray*}
\bV_n \coloneqq \sqrt{k_n \rho_n}  \left( \frac{\bT_{n,{\{{s^*}+1, \ldots, r\} }}(k_n)}{\rho_n k_n} - \mathbf{1}_{r-s^*} \right) \quad \text{ and }\quad \bV \sim \mathcal{N}_{r-{s^*}}(\mathbf{0}_{r-{s^*}}, \bI_{r-{s^*}}).
\end{eqnarray*}
By Assumption~(\ref{asu:directions}\ref{(A4)}) we have $\bV_n \Dconv \bV.$  
Furthermore, since $T_{n,l}(k_n)/(k_n \rho_n) \Pconv 1$ for $l={s^*}+1, \ldots, r$ by \Cref{cor:Theorem1_MW}(c), it follows  that
\begin{align}
    &\hspace*{-0.2cm}\frac{\rho_n}{\sum_{l={s^*}+1}^r \frac{T_{n,l}(k_n)}{k_n(r-{s^*})}}  \frac{k_n}{\rho_n}   \sum_{j={s^*}+1}^r \left(\frac{T_{n,j}(k_n)}{k_n} - \sum_{i={s^*}+1}^r \frac{T_{n,i}(k_n)}{k_n(r-{s^*})} \right) ^2 \nonumber \\
    &= \underbrace{\frac{\rho_n}{\sum_{l={s^*}+1}^r \frac{T_{n,l}(k_n)}{k_n(r-{s^*})}}}_{ \Pconv 1} \underbrace{\bV_n^\top (\bI_{r-{s^*}} - \frac{1}{r-{s^*}} \mathbf{1}_{r-s^*} \mathbf{1}_{r-s^*}^\top)^\top (\bI_{r-{s^*}} - \frac{1}{r-{s^*}} \mathbf{1}_{r-s^*} \mathbf{1}_{r-s^*}^\top) \bV_n}_{ \Dconv \chi_{r - s^* -1}^2 } \nonumber \\
    &\Dconv \chi_{r - s^* -1}^2 =  O_\P(1). \label{eq:MSE_cons_Term2}
\end{align}
Combining \cref{eq:MSE_cons_Term1} and  \cref{eq:MSE_cons_Term2} yields
   \begin{eqnarray*}
        \lefteqn{\MSEIC_{k_n}(s)-  \MSEIC_{k_n}(s^*)} \\
        &= &2(s - s^*) + \frac{k_n}{\sum_{l=s+1}^r \frac{T_{n,l}(k_n)}{k_n(r-s)}} \sum_{j=s+1}^r \left(\frac{T_{n,j}(k_n)}{k_n} - \sum_{i=s+1}^r \frac{T_{n,i}(k_n)}{k_n(r-s)} \right) ^2 \\
        & &\quad  - \frac{k_n}{\sum_{l={s^*}+1}^r \frac{T_{n,l}(k_n)}{k_n(r-{s^*})}} \sum_{j={s^*}+1}^r \left(\frac{T_{n,j}(k_n)}{k_n} - \sum_{i={s^*}+1}^r \frac{T_{n,i}(k_n)}{k_n(r-{s^*})} \right) ^2\\
        &\Pconv & \infty.
    \end{eqnarray*}
     \textbf{Step 2:} Suppose $s > {s^*}$.
An application of \eqref{eq:MSE_cons_Term2} and  \Cref{cor:Theorem1_MW}(c) yield
\begin{align}
      \frac{k_n}{\sum_{l={s^*}+1}^r \frac{T_{n,l}(k_n)}{k_n(r-{s^*})}} & \sum_{j={s^*}+1}^r \left(\frac{T_{n,j}(k_n)}{k_n} - \sum_{i={s^*}+1}^r \frac{T_{n,i}(k_n)}{k_n(r-{s^*})} \right) ^2 \nonumber \\
    &  - \frac{k_n}{\rho_n}  \sum_{j={s^*}+1}^r \left(\frac{T_{n,j}(k_n)}{k_n} - \sum_{i={s^*}+1}^r \frac{T_{n,i}(k_n)}{k_n(r-{s^*})} \right) ^2 \nonumber \\
    = & \underbrace{\left( \frac{\rho_n}{\sum_{l={s^*}+1}^r \frac{T_{n,l}(k_n)}{k_n(r-{s^*})}} -1 \right) }_{ \Pconv 0}  \underbrace{\frac{k_n}{\rho_n} \sum_{j={s^*}+1}^r \left(\frac{T_{n,j}(k_n)}{k_n} - \sum_{i={s^*}+1}^r \frac{T_{n,i}(k_n)}{k_n(r-{s^*})} \right) ^2}_{ \Dconv  \chi_{r-s^*-1}^2 \text{ by } \eqref{eq:MSE_cons_Term2}} \nonumber \\
    = {}& o_\P(1). \label{eq:MSE_rho_Vorfaktor}
\end{align}
Since $s > {s^*}$ the analog holds when ${s^*}$ is replaced by $s$.  Using $\bV_n = (V_{n, {s^*}+1} , \ldots,V_{n,r})^\top$ defined as above, we have the representation $\frac{T_{n,j}(k_n)}{k_n \rho_n} = \frac{1}{\sqrt{k_n \rho_n}} V_{n,j} +1.$
Thus, when inserting the definition  of $\MSEIC$ we get with \cref{eq:MSE_rho_Vorfaktor}  that

\begin{align*}
   &\hspace*{-0.7cm} \MSEIC_{k_n}(s) -  \MSEIC_{k_n}(s^*)  \nonumber  \\
   ={} & 2(s - s^*) + \frac{k_n}{\rho_n} \sum_{j=s+1}^r \left(\frac{T_{n,j}(k_n)}{k_n} - \sum_{i=s+1}^r \frac{T_{n,i}(k_n)}{k_n(r-s)} \right) ^2 \nonumber  \\
        & - \frac{k_n}{\rho_n} \sum_{j={s^*}+1}^r \left(\frac{T_{n,j}(k_n)}{k_n} - \sum_{i={s^*}+1}^r \frac{T_{n,i}(k_n)}{k_n(r-{s^*})} \right) ^2 + o_\P(1)\nonumber \\
        ={} & 2(s - s^*) +  \sum_{j=s+1}^r \left\{ \sqrt{ k_n \rho_n} \left( \frac{T_{n,j}(k_n)}{k_n \rho_n}  -1 \right) - \frac{ \sqrt{k_n \rho_n}}{ (r-s)}   \sum_{i=s+1}^r \left(\frac{T_{n,i}(k_n)}{k_n \rho_n} -1  \right) \right\}^2  \nonumber \\
        & - \sum_{j={s^*}+1}^r \left\{ \sqrt{ k_n \rho_n} \left( \frac{T_{n,j}(k_n)}{k_n \rho_n}  -1 \right) - \frac{ \sqrt{k_n \rho_n}}{ (r-{s^*})}   \sum_{i={s^*}+1}^r \left(\frac{T_{n,i}(k_n)}{k_n \rho_n} -1  \right) \right\}^2  + o_\P(1) \nonumber \\
        ={} & 2(s - s^*) +  \sum_{j=s+1}^r \left\{ V_{n,j} - \frac{1}{ (r-s)}   \sum_{i=s+1}^r V_{n,i} \right\}^2  \nonumber \\
        & -  \sum_{j={s^*}+1}^r \left\{ V_{n,j} - \frac{1}{ (r-{s^*})}   \sum_{i={s^*}+1}^r V_{n,i} \right\}^2 + o_\P(1) \nonumber \\
        \Dconv{} & 2(s - s^*) +  \sum_{j=s+1}^r \left\{ V_{j} - \frac{1}{ (r-s)}   \sum_{i=s+1}^r V_{i} \right\}^2 -  \sum_{j={s^*}+1}^r \left\{ V_{j} - \frac{1}{ (r-{s^*})}   \sum_{i={s^*}+1}^r V_{i} \right\}^2. 
\end{align*}

Similar to the proof of \Cref{th:AIC_Cons}, there exists a positive probability that the right-hand side is positive. 
Hence, the assertion follows.
\end{proof}

\subsubsection{Proof of Theorem \ref{th:MSE_global_derivation}} \label{proof:MSE_global_derivation}
Before we are able to present the proof of  \Cref{th:MSE_global_derivation} we require some auxiliary lemmata whose proofs are moved to \Cref{sec:MSEIC_supp} in the Supplementary Material.

\begin{lemma} \label{lem:MSEIC_glob1}
    Suppose assumptions (\ref{asu:Model_global}\ref{asu:BIC_global_expectation_local}) and (\ref{asu:Model_global}\ref{asu:BIC_global_expectation_local2}) hold. Then for $\bp' \in \R_+^r$ the asymptotic behavior 
    \begin{align*}
     \E  &\left[   \left\Vert \sqrt{n- \td}  \diag( \bp' )^{-1/2}  \left( \frac{\bT_{n, \{1, \ldots, r\}}'}{n- \td} -    \bp' \right)  \right\Vert^2_2  \right]  \\
     & \qquad =  n q_n \left( \frac{1 }{k_n} \E[ \ell^2 \big( \bp' \vert \bT_n(k_n) \big)  ] + o \left(  \frac{1}{n q_n} \right) \right) 
\end{align*}
as $n\to\infty$ holds.
\end{lemma}

\begin{lemma} \label{lem:MSEIC_glob2}
For $q' \in (0,1)$ the equality 
        \begin{align*}
     \E &\left[   \left\Vert \sqrt{n} ( q' ( 1-q') )^{-1/2} \left(  \frac{\td}{n}  - (1-q') \right) \right\Vert^2_2  \right] = n q_n  \left( \frac{(1-q_n) }{n q' ( 1-q') } +  \frac{(q'-q_n)^2 }{q_n q' ( 1-q') } \right)  
\end{align*}
holds.
\end{lemma}

\begin{proof}[Proof of Theorem \ref{th:MSE_global_derivation}]
    For $q' \in (0,1)$ and $\bp' \in \R^r_+$ we have as a consequence of \Cref{lem:MSEIC_glob1,lem:MSEIC_glob2}, that
    \begin{align*}
         q'  \E & \left[   \left\Vert \sqrt{n- \td}  \diag( \bp' )^{-1/2}  \left( \frac{\bT_{n, \{1, \ldots, r\}}'}{n- \td} -    \bp' \right)  \right\Vert^2_2  \right] \nonumber \\
   & + (1- q') \E \left[   \left\Vert \sqrt{n} ( q' ( 1-q') )^{-1/2} \left(  \frac{\td}{n}  - (1-q') \right) \right\Vert^2_2  \right] \\
   &= n q_n \left( \frac{q'}{k_n} \E[ \ell^2 \big( \bp' \vert \bT_n(k_n) \big)  ] + o \left(  \frac{q'}{n q_n} \right) \right) + n q_n  \left( \frac{(1-q_n)}{n q'  } + \frac{(q'-q_n)^2 }{q_n q'  } \right).
    \end{align*}

Therefore, it follows that
\begin{align*}
    \E \Bigg[    q'  \E & \left[   \left\Vert \sqrt{n- \td}  \diag( \bp' )^{-1/2}  \left( \frac{\bT_{n, \{1, \ldots, r\}}'}{n- \td} -    \bp' \right)  \right\Vert^2_2  \right] \bigg\vert_{ \bp' =  \widehat{\bp}_n( \frac{\widetilde{\bT}_n(k_n)}{k_n}), q' = \frac{k_n}{n}     }  \Bigg   ] \nonumber \\
   & \quad +  \E \Bigg[ (1- q') \E \left[   \left\Vert \sqrt{n} ( q' ( 1-q') )^{-1/2} \left(  \frac{\td}{n}  - (1-q') \right) \right\Vert^2_2  \right] \bigg\vert_{  q' = \frac{k_n}{n}  }  \Bigg   ]\\
   &= n q_n  \E \Bigg[   \left( \frac{q'}{k_n} \E[ \ell^2 \big( \bp' \vert \bT_n(k_n) \big)  ] + o \left(  \frac{q'}{n q_n} \right)    \right) \Big\vert_{ \bp' =  \widehat{\bp}_n( \frac{\widetilde{\bT}_n(k_n)}{k_n}), q' = \frac{k_n}{n}     } \Bigg] \\
   & \quad + n q_n  \E \Bigg[ \left( \frac{(1-q_n)}{n q'} + \frac{(q'-q_n)^2 }{q_n q'  } \right) \Big\vert_{  q' = \frac{k_n}{n} }  \Bigg] \\
   & = n q_n     \left(   \frac{ 1 }{n} \MSE_{k_n}(s)+    \frac{(1-q_n)}{k_n} +\frac{(\frac{k_n}{n}-q_n)^2 }{q_n  \frac{k_n}{n}  }  + o \left(  n^{-1} \right) \right).
\end{align*}
{Due to the asymptotic behavior as $n\to\infty$,
\begin{eqnarray*}
    \frac{(\frac{k_n}{n}-q_n)^2 }{q_n  \frac{k_n}{n}  } +o \left( n^{-1}\right)
    =\frac{k_n(1-\frac{nq_n}{k_n})^2 }{nq_n   } {+o \left( n^{-1}\right)} =   o((n q_n)^{-1})+o \left( n^{-1}\right)=   o((n q_n)^{-1}),
\end{eqnarray*}
where we used the additional assumption $k_n ( 1 - \frac{nq_n}{k_n})^2 \to 0$ as $n\to\infty$, we can conclude the statement. 
} 
\end{proof}

\subsection{Proofs of Section \ref{sec:BIC}} \label{sec:proof_BIC}

\subsubsection{Proof of Theorem \ref{th:BIC_post_prob}}\label{proof:BIC_post_prob}
In the next two lemmata, we derive auxiliary results used for the derivation of an upper bound of the posterior probability  $\P(  M^s_{k_n} | \bT_n(k_n) )$. First, in \Cref{th:BIC_lokal_Taylor_Approx},  we give a Taylor approximation of the log-likelihood function 
$\log(L_{M^s_{k_n}}(\,\cdot\,|\bT_n(k_n)))$ of  Model $ M^s_{k_n}$, and second, in \Cref{th:BIC_lokal_Eigenwerte_Absch}, we present boundaries for the eigenvalues of the Hessian of the log-likelihood function;  
the proofs of these auxiliary results are included in \Cref{sec:BIC_local_supp} of the Supplementary Material.
Finally, for the proof of the upper bound of the log-posterior distribution in \Cref{th:BIC_post_prob} we combine these two results.

\begin{lemma} \label{th:BIC_lokal_Taylor_Approx}
Let the assumptions of Theorem \ref{th:BIC_post_prob} hold. Define the ball $$U_{\varepsilon_{n, \gamma}}(\widehat{\bp}_n^s) \coloneqq \{ \widetilde{\bp}^s \in\Theta_s: \Vert \widetilde{\bp}^s - \widehat{\bp}_n^s \Vert_2 < \varepsilon_{n, \gamma} \}$$ with radius $\varepsilon_{n, \gamma} \coloneqq {(\rho_n)^\gamma}/{2}$ for $\gamma \ge 4/3$ around $\widehat{\bp}_n^s$.
Then the following statement holds 
    $\begin{aligned}[t]
   \sup_{ \widetilde{\bp}^s \in U_{\varepsilon_{n, \gamma}}(\widehat{\bp}_n^s) }  \bigg \vert  &   \log L_{M^s_{k_n}}(\widetilde{\bp}^s\, \vert\, \bT_n(k_n)) - \log L_{M^s_{k_n}}(\widehat{\bp}_n^s\, \vert \, \bT_n(k_n)) \\
   &  -  \frac12 (\widetilde{\bp}^s - \widehat{\bp}_n^s)^\top    \nabla^2 \log L_{M^s_{k_n}}(\widehat{\bp}^s_n\, \vert \, \bT_n(k_n)) (\widetilde{\bp}^s - \widehat{\bp}_n^s) \bigg \vert   = o_\P(1).
\end{aligned}$
\end{lemma}

\begin{lemma} \label{th:BIC_lokal_Eigenwerte_Absch}
Let the assumptions of Theorem \ref{th:BIC_post_prob} hold. Define $\lambda_{n,2} \coloneqq {k_n }/{T_{n,1}(k_n)}$ and $\lambda_{n,1} \coloneqq {k_n }/{T_{n,s}(k_n)} + {s k_n}/  \sum_{j = s+1}^{r} T_{n,j} $. For $\widetilde{\bp}^s \in \Theta_s$ we have on the one hand,
    \begin{align*}
    \lambda_{n,2} & (\widetilde{\bp}^s - \widehat{\bp}_n^s)^\top  (\widetilde{\bp}^s - \widehat{\bp}_n^s) \leq(\widetilde{\bp}^s - \widehat{\bp}_n^s)^\top  \frac{-1}{k_n}  \nabla^2 \log L_{M^s_{k_n}}(\widehat{\bp}_n^s\, \vert \, \bT_n(k_n)) (\widetilde{\bp}^s - \widehat{\bp}_n^s) \quad \P \text{-a.s.}
\end{align*}
and on the other hand,
   \begin{align*}
    \lambda_{n,1}  (\widetilde{\bp}^s - \widehat{\bp}_n^s)^\top (\widetilde{\bp}^s - \widehat{\bp}_n^s) \geq (\widetilde{\bp}^s - \widehat{\bp}_n^s)^\top  \frac{-1}{k_n}  \nabla^2 \log L_{M^s_{k_n}}(\widehat{\bp}_n^s\, \vert \, \bT_n(k_n)) (\widetilde{\bp}^s - \widehat{\bp}_n^s)  \quad \P \text{-a.s.}
\end{align*}
\end{lemma}

\begin{proof}[Proof of \Cref{th:BIC_post_prob}.]

In the following let $\gamma = 4/3$ and  $\varepsilon_n \coloneqq \varepsilon_{n, 4/3} =  (\rho_n)^{\,4/3}/{2}$.
An application of \Cref{th:BIC_lokal_Taylor_Approx}, \Cref{th:BIC_lokal_Eigenwerte_Absch} and Assumption (\ref{asu:BIC_local}\ref{asu:BIC_local3}) give
\begin{align}
    -2 \log &  \E_{g_s} [ L_{M^s_{k_n}}(\widetilde{\bp}^s\, \vert \, \bT_n(k_n)) ]\nonumber\\
    \le&- 2 \log \int_{  U_{\varepsilon_{n}}(\widehat{\bp}_n^s)}  L_{M^s_{k_n}}(\widetilde{\bp}^s\, \vert \, \bT_n(k_n))  \di \widetilde{\bp}^s -2\log b \nonumber\\
    \le& -2 \log L_{M^s_{k_n}}(\widehat{\bp}_n^s\, \vert \, \bT_n(k_n)) \nonumber \\  & -2 \log \int_{ U_{\varepsilon_{n}}(\widehat{\bp}_n^s) } \! \exp \Big\{  \frac{-k_n}{2} (\widetilde{\bp}^s - \widehat{\bp}_n^s)^\top   \frac{-1}{k_n}  \nabla^2 \log L_{M^s_{k_n}}(\widehat{\bp}_n^s\, \vert \, \bT_n(k_n)) (\widetilde{\bp}^s - \widehat{\bp}_n^s)   \Big\} \! \di \widetilde{\bp}^s \nonumber \\ &  -2\log b+o_\P(1) \notag\\
    \leq & -2 \log L_{M^s_{k_n}}(\widehat{\bp}_n^s\, \vert \, \bT_n(k_n)) -s \log( 2 \pi)  + s \log( k_n \lambda_{n,1})-2\log b \nonumber \\
     & -2 \log \int_{  U_{\varepsilon_{n}}(\widehat{\bp}_n^s)}  \left(\frac{k_n \lambda_{n,1} }{2 \pi}\right)^{s/2} \exp \left\{  \frac{-1}{2}   \frac{(\widetilde{\bp}^s - \widehat{\bp}_n^s)^\top  (\widetilde{\bp}^s - \widehat{\bp}_n^s)}{1/(k_n \lambda_{n,1}) }   \right\}  \di \widetilde{\bp}^s +o_\P(1). \label{BIC:local_log_integral}
\end{align}

The integrand is a $s$-dimensional Gaussian density with expectation vector $\widehat{\bp}_n^s$ and covariance matrix $(k_n \lambda_{n,1})^{-1} \bI_s$. 
Furthermore, due to the definition of $\lambda_{n,1}$, Assumption (\ref{asu:BIC_local}\ref{asu:BIC_local1}) and \Cref{cor:Theorem1_MW}, the asymptotic behavior 
\begin{align}
0  \leq k_n \lambda_{n,1} \varepsilon_{n}^2
& = \underbrace{\frac{k_n (\rho_n)^{5/3}}{4}}_{\rightarrow \infty}  \underbrace{\left( \frac{k_n \rho_n }{T_{n,s}(k_n)} + \frac{s k_n \rho_n }{  \sum_{j = s+1}^{r} T_{n,j} } \right)}_{\overset{\text{\Cref{cor:Theorem1_MW}}}{\rightarrow} \mathds{1}_{\{s\geq s^*\}} + \frac{s}{r-\max(s,s^*)}} \Pconv \infty \label{conv:BIC_Cov_Matrix}
\end{align}
holds in probability. 
Let $\bN \sim \mathcal{N}_s(\mathbf{0}_s, \bI_s)$. Since $\left \Vert  \bN   \right \Vert_2^2 \sim \chi^2_s$ the Markov inequality yields 
\begin{align*}
    \int_{  U_{\varepsilon_n}(\widehat{\bp}_n^s)} & \left(\frac{k_n \lambda_{n,1} }{2 \pi}\right)^{s/2} \exp \left\{  \frac{-1}{2}   \frac{(\widetilde{\bp}^s - \widehat{\bp}_n^s)^\top  (\widetilde{\bp}^s - \widehat{\bp}_n^s)}{1/(k_n \lambda_{n,1}) }   \right\}  \di \widetilde{\bp}^s \\
    &= \P \left(\left. \widehat{\bp}_n^s +  \frac{1}{\sqrt{(k_n \lambda_{n,1})}} \bN \in U_{\varepsilon_n} (\widehat{\bp}_n^s)\right|\bT_n(k_n) \right)\\
    &= 1 -\P \left(\left. \left \Vert  \bN   \right \Vert_2^2 \geq k_n \lambda_{n,1} \varepsilon_n^2 \right|\bT_n(k_n)\right)\\
    &\geq 1 - \frac{s}{k_n \lambda_{n,1} \varepsilon_n^2} \to 1,
\end{align*}
as $n\to\infty$ almost surely,  where we used in the last step \cref{conv:BIC_Cov_Matrix}.
Thus, 
\begin{align}
   -2\log\int_{  U_{\varepsilon_n}(\widehat{\bp}_n^s)}  \left(\frac{k_n \lambda_{n,1} }{2 \pi}\right)^{s/2} \exp \left\{  \frac{-1}{2}   \frac{(\widetilde{\bp}^s - \widehat{\bp}_n^s)^\top  (\widetilde{\bp}^s - \widehat{\bp}_n^s)}{1/(k_n \lambda_{n,1}) }   \right\}  \di \widetilde{\bp}^s =o_{\P}(1). \label{ineq:BIC_local_integral_bound}
\end{align}
Inserting \Cref{ineq:BIC_local_integral_bound} into \cref{BIC:local_log_integral} gives then
\begin{align*}
   -2 \log &  \E_{g_s} [ L_{M^s_{k_n}}(\widetilde{\bp}^s\, \vert \, \bT_n(k_n)) ]\\ 
   & \leq -2 \log L_{M^s_{k_n}}(\widehat{\bp}_n^s\,  \vert \, \bT_n(k_n)) - s \log(2 \pi)   + s \log( k_n \lambda_{n,1})-2\log b  + o_\P(1). 
\end{align*}
Since $T_{n,j}(k_n) \geq 1$ for $j=1, \ldots, s$, we receive the upper bound  
\begin{equation*}
\lambda_{n,1} = \left(\frac{k_n }{T_{n,s}(k_n)} + \frac{s k_n}{  \sum_{j = s+1}^{r} T_{n,j}(k_n) }\right) \leq k_n \left(1 + \frac{s}{  r-s }\right) = k_n \frac{r}{r-s},
\end{equation*}    
and finally,
\begin{align*}
-2 \log &  \E_{g_s} [ L_{M^s_{k_n}}(\widetilde{\bp}^s\, \vert \, \bT_n(k_n)) ]\\
&\leq -2 \log  L_{M^s_{k_n}}(\widehat{\bp}_n^s\,  \vert \, \bT_n(k_n)) - s \log(2 \pi) + 2 s \log \left( k_n \sqrt{\frac{r}{r-s}}  \right)-2\log b   + o_\P(1), 
\end{align*}
which is the statement.
\end{proof}

\subsubsection{Proof of Theorem \ref{th:BICU_Consistency}} \label{proof:BICU_Consistency}
\begin{proof}[Proof of Theorem \ref{th:BICU_Consistency}]
$\mbox{}$\\
(a) Note that 
\begin{align*}
\BICU_{k_n}(s) = 2 \AIC_{k_n}(s) - 2 s + 2 s \log \left( k_n  \right) + s \log \left(\frac{ r}{2 \pi ( r-s) }  \right).
\end{align*}
We consider now the different cases $s>s^*$ and $s<s^*$ separately.\\
 \textbf{Step 1:} Suppose $s > s^*$. We receive with  \cref{eq:AIC4_Difference2} that
\begin{align*}
\BICU_{k_n}(s&)  - \BICU_{k_n}(s^*) \nonumber\\
=& 2 \AIC_{k_n}(s) - 2 s + 2 s \log \left( k_n  \right) + s \log \left(\frac{ r}{2 \pi ( r-s) }  \right) \nonumber \\
&- 2 \AIC_{k_n}(s^*)+ 2 {s^*} - 2 {s^*} \log \left( k_n  \right) - {s^*} \log \left(\frac{ r}{2 \pi ( r-{s^*}) }  \right)  \nonumber\\
=&   2(s - {s^*}) \log(k_n) 
+O_\P\left( 1 \right). 
\end{align*}
Dividing the last equation by $\log(k_n)$  results in
\begin{align*}
&\frac{\BICU_{k_n}(s)  - \BICU_{k_n}(s^*)}{\log(k_n)} \Pconv 2(s - {s^*}) > 0, 
\end{align*}
where we used  $\log (k_n) \rightarrow \infty$.\\
\textbf{Step 2:} Suppose $s < s^*$. Here we have as in the proof of \Cref{th:AIC_Cons} and due to $\log(k_n)/k_n\to 0$ that
\begin{align*}
    &\frac{\BICU_{k_n}(s)  -  \BICU_{k_n}(s^*)}{k_n} \\
    &\quad= 2  \frac{\AIC_{k_n}(s) -  \AIC_{k_n}(s^*)}{k_n} +  \frac{- 2 s + 2 s \log \left( k_n  \right) + s \log \left(\frac{ r}{2 \pi ( r-s) }  \right)}{k_n} \nonumber \\
&\qquad  + \frac{2 {s^*} - 2 {s^*} \log \left( k_n  \right) - {s^*} \log \left(\frac{ r}{2 \pi ( r-{s^*}) }  \right)}{k_n}  \nonumber\\
&\quad\Dconv 2 \sum_{i=s+1}^{{s^*}}  p_i \Bigg( \log \left( p_i \right) -  \log \Bigg( \frac{1}{r-s} \sum_{j=s+1}^{s^*} p_j \Bigg) \Bigg) > 0, 
\end{align*}
and thus, the assertion follows.

(b) \, Again, note that
\begin{align*}
\BICL_{k_n}(s) = 2 \AIC_{k_n}(s) - 2 s + s \log \left( k_n  \right) + s \log \left(\frac{ k_n}{2 \pi T_{n,1}(k_n) }  \right).
\end{align*}
By a calculation analog to part (a), the $\BICL$ is also consistent since $s \log \left(\frac{ k_n}{2 \pi T_{n,1}(k_n) }  \right) \Pconv s \log \left(\frac{ 1}{2 \pi p_{1} }  \right) > 0$ as $n\to\infty$. 
\end{proof}

\subsubsection{Proof of Theorem \ref{th:BIC_global_bound}} \label{proof:BIC_global_bound}
First, we derive some auxiliary results before we prove Theorem \ref{th:BIC_global_bound}. 
Therefore, note that   due \eqref{Model global} (cf. Equation (1.23) in the Supplementary Material of \citet{meyer_muscle23}) and $\sum_{j=1}^{2^d-1} T_{n,j}' = n- \td$, the likelihood function of Model $M^{\prime s}_n$  can be written as
\begin{align}
& L_{M^{\prime s}_n} (\widetilde{\bp}^{\prime s} \, \vert \, \bT_n') =  L_{M^s_{n-T_{n,2^d}'}} ( \widetilde{\bp}^s \, \vert \, \bT_{n, \{1, \ldots, r \}}') \cdot  L_{\Bin_n} (1 - \widetilde{q} \, \vert \, T_{n,2^d}'), \label{eq:Likelihood_Multinomial}
\end{align}
for $\widetilde{\bp}^{\prime s} = (\widetilde{\bp}^{s}, \widetilde{q}) \in \Theta_s' = \Theta_s \times (0,1)$, where  
 \begin{align} \label{Bin}
    L_{\Bin_n} (1 - \widetilde{q} \, \vert \, T_{n,2^d}'):=
    \left( \binom{n}{T_{n,2^d}'}  \right)(1-\widetilde q)^{T_{n,2^d}'}
    \widetilde{q}^{n-T_{n,2^d}'} 
\end{align}    
    is the likelihood function of the binomial model.  
   Next, we define the following expectations with respect to the Lebesgue measure $\lambda$. Let
\begin{align}
\begin{array}{rl}
    \E_{\lambda}  [ L_{M^s_{n-T_{n,2^d}'}} ( \widetilde{\bp}^s \, \vert \, \bT_{n, \{1, \ldots, r \}}') ] &\coloneqq  \int_{ \Theta_s}  L_{M^s_{n-T_{n,2^d}'}} ( \widetilde{\bp}^s \, \vert \, \bT_{n, \{1, \ldots, r \}}')   \di \widetilde{\bp}^s, \\
    \E_{\lambda}  [ L_{\Bin_n} (1 - \widetilde{q} \, \vert \, T_{n,2^d}')] &\coloneqq  \int_{ (0,1)}  L_{\Bin_n} (1 - \widetilde{q} \, \vert \, T_{n,2^d}') \di \widetilde{q}.  
    \end{array} \label{eq:BIC_local_post}
\end{align}
Then taking the expectation and logarithm in \cref{eq:Likelihood_Multinomial} results under Assumption (\ref{asu:BIC_global}\ref{asu:BIC_global_prior_density})  in
\begin{align}
-2 & \log \E_{g_s'}  [  L_{M^{\prime s}_n} (\widetilde{\bp}^{\prime s} \, \vert \, \bT_n') ] \nonumber \\
&\le - 2 \log b' -2 \log \Big\{ \int_{\Theta_s \times (0,1)} \!  L_{M^s_{n-T_{n,2^d}'}} \! ( \widetilde{\bp}^s \, \vert \, \bT_{n, \{1, \ldots, r \}}') \cdot  L_{\Bin_n} (1 - \widetilde{q} \, \vert \, T_{n,2^d}')  \di (\widetilde{\bp}^{s}, \widetilde{q}) \Big\} \nonumber \\
&= - 2 \log b'  -2 \log \Big\{ \int_{\Theta_s }  L_{M^s_{n-T_{n,2^d}'}} ( \widetilde{\bp}^s \, \vert \, \bT_{n, \{1, \ldots, r \}}')  \di \widetilde{\bp}^{s} \cdot \int_{ (0,1)} L_{\Bin_n} (1 - \widetilde{q} \, \vert \, T_{n,2^d}')  \di \widetilde{q} \Big\} \nonumber \\
&= - 2 \log b'  -2  \log \E_{\lambda}  [ L_{M^s_{n-T_{n,2^d}'}} ( \widetilde{\bp}^s \, \vert \, \bT_{n, \{1, \ldots, r \}}') ] -2 \log   \E_{\lambda}  [L_{\Bin_n} (1 - \widetilde{q} \, \vert \, T_{n,2^d}')].  \label{eq:expectationglobal}
\end{align}

In the following two auxiliary lemmata, we determine upper bounds for the expectation of both summands.

\begin{proposition} \label{th:BIC_global_Mult_Approx}
Under Assumptions~(\ref{asu:Model_global}\ref{asu:BIC_global_expectation_local}), (\ref{asu:Model_global}\ref{asu:BIC_glob_qn_kn}) and (\ref{asu:BIC_global}\ref{E1})   the asymptotic upper bound as $n\to\infty$,
      \begin{align*}
    -&2  \E\big[  \log\E_{\lambda}  [ L_{M^s_{n-T_{n,2^d}'}}( \widetilde{\bp}^s \, \vert \, \bT_{n, \{1, \ldots, r \}}') ] \big] \notag\\
    &\leq  -2 \E \bigl[  \log \big( (n - T_{n,2^d}') ! \big)  - (n -\td) \left( \log ( n- \td) - 1 \right) \bigl]  \\
    &\quad-2 \frac{n q_n}{k_n} \E[ \log L_{M^s_{k_n}} ( \widehat{\bp}_n^s(\bT_n(k_n)) \, \vert \, \bT_n(k_n) ) ]     + 2  s   \log \Big(  k_n   \sqrt{ \frac{r}{2 \pi (r-s)}} \Big)    + C \log(n q_n), 
\end{align*}
for a constant $C > 0$ independent of $s$ and $n$, holds.
\end{proposition}

\begin{proposition}\label{ineq:LogBinAbsch}
Suppose Assumptions (\ref{asu:BIC_global}\ref{asu:BIC_conv_nqn53})  and   (\ref{asu:Model_global}\ref{asu:BIC_glob_qn_kn}) hold. The expectation of the  binomial likelihood satisfies as $n\to\infty$ the inequality
\begin{align*} 
   -2 \E[  \log  \E_{\lambda}  [ L_{\Bin_n} (1 - \widetilde q \, \vert \, T_{n,2^d}')]  ] 
   \leq& -2  \log(n!) +2 \E[\log( ( n - \td)!)]  +2 \E[ \log( \td!)]  \\
   &-2 n q_n \log( k_n/n)   +2 \log( n)  + C n q_n, 
\end{align*}
for a constant $C > 0$ independent of $s$ and $n$.
\end{proposition}

\begin{proof}[Proof of \Cref{th:BIC_global_bound}.] For the ease of notation we define $x\log x$ as zero if $x=0$.
Inserting the bounds derived in \Cref{th:BIC_global_Mult_Approx} with constant $C_1$ and \Cref{ineq:LogBinAbsch} with constant $C_2$ into \cref{eq:expectationglobal}  gives 
for sufficiently large $n$ that
\begin{align}
-2 \E [& \log  \E_{g_s'}  [ L_{M^{\prime s}_n} (\widetilde{\bp}^{\prime s} \, \vert \, \bT_n') ] ] + 2 \log b' \nonumber\\
&\le -2 \E [  \log  \E_{\lambda}  [ L_{M^s_{n-T_{n,2^d}'}} ( \widetilde{\bp}^s \, \vert \, \bT_{n, \{1, \ldots, r \}}') ] ] -2 \E [  \log  \E_{\lambda}  [ L_{\Bin_n} (1 - \widetilde{q} \, \vert \, T_{n,2^d}')] ] 
\nonumber \\
 &\leq -2 \E \bigl[  \log \big( (n - T_{n,2^d}') ! \big)  - (n -\td) \left( \log ( n- \td) -1 \right) \bigr]  \nonumber\\
    &\quad-2 \frac{n q_n}{k_n} \E[ \log L_{M^s_{k_n}} ( \widehat{\bp}_n^s  \, \vert \, \bT_n(k_n) ) ]     + 2  s    \log \Big( k_n   \sqrt{ \frac{r}{2 \pi (r-s)}} \Big) 
    \nonumber\\
&\quad-2  \log(n!) +2 \E[\log( ( n - \td)!)]  +2 \E[ \log( \td!)]   \nonumber\\
   &\quad-2 n q_n \log( k_n/n)    + 2 \log( n)  +  (C_1+C_2) n q_n \nonumber\\
&= \Bigl\{ -2  \log(n!)+2 \E \bigl[   (n -\td) \left( \log ( n- \td) -1 \right) \bigl] +2 \E[ \log( \td!)]   \Bigr\} \nonumber\\
    &\quad + \Big\{-2 \frac{n q_n}{k_n} \E[ \log L_{M^s_{k_n}} ( \widehat{\bp}_n^s  \, \vert \, \bT_n(k_n) ) ]     + 2  s    \log \Big( k_n   \sqrt{ \frac{r}{2 \pi (r-s)}} \Big) \nonumber\\
   &\quad\quad -2 n q_n \log( k_n/n)    +2 \log( n) \Big\}   + (C_1+C_2) n q_n \nonumber \\
   &\eqqcolon I_{n,1} + I_{n,2}+ (C_1+C_2) n q_n. \label{ineq:BIC_global_bin_log_exp_l}
\end{align}
Next, we simplify $I_{n,1}$. Therefore, we use the following calculation.
Let $B$ be a positive random variable with finite positive variance. For $u > 0$ and $ x > 0$ we the inequality $\log( x/u) \le x/u -1$ holds, which is equivalent to $x \log(x) \le x^2/u + x \log(u) -x$. Then we have
\begin{align*}
    \E[B \log(B)] \le \frac{\E[ B^2]}{u} + \E[B] \log(u) - \E[B],  
\end{align*}
and in particular for $u=\E[ B^2]/\E[ B]$ we receive
\begin{align*}
    \E[B \log(B)] \le  \E[B] \log(\E[ B^2]/\E[ B]).  
\end{align*}
Since $\E[ \td \mathbbm{1}\{\td > 0  \} ] = \E[ \td] =  n( 1-q_n)$ and $\E[ \td^2 \mathbbm{1}\{\td > 0  \} ] = \E[ \td^2] = n q_n (1-q_n)$ the previous inequality gives 
\begin{align}
    \E[\td \log(\td) & \mathbbm{1}\{\td > 0  \}] \nonumber\\
    &\le n (1-q_n) \log \left( \frac{n^2 (1-q_n)^2  + n q_n (1-q_n)}{n (1-q_n) } \right) \nonumber\\
    &= n (1-q_n) \log \left( n (1-q_n) +  q_n \right) \nonumber\\
    &= n (1-q_n) \log \left( n (1- q_n)  \right) + n (1-q_n) \log \left( \frac{n (1- q_n)  + q_n}{n (1- q_n)}  \right) \nonumber\\
    &\leq n (1-q_n) \log \left( n (1- q_n)  \right) + C_3     \label{A.17}
\end{align}
for a constant $C_3 > 0$ independent of $s$ and $n$. 
Furthermore, we use the inequality 
\begin{align}
 n \log n - n < \log(n!) < n \log n - n + \log n +1\, \label{hilf:4}
\end{align}
to derive  a bound for $\E[ \log( \td!)]$.
Hence, 
using  the upper bound \Cref{hilf:4}, \Cref{A.17} and applying Jensen inequality we receive that
\begin{align*}
     \E[ & \log  ( \td!)] \\
     &=  \E[ \log( \td!) \mathbbm{1}\{\td > 0  \}] \\
     &  \leq  \E[ \td  \log ( \td) \mathbbm{1}\{\td > 0  \}] - \E[ \td  \mathbbm{1}\{\td > 0  \}]+\E[\log(\td \mathbbm{1}\{\td > 0  \})] \\
     &\leq n (1-q_n) \log \left( n (1- q_n)  \right)-n(1-q_n) + \log \left( n (1- q_n)  \right)+ C_4 
\end{align*}
for a constant $C_4 > 0$ independent of $s$ and $n$.
Additionally to the last inequality, we obtain by \Cref{hilf:4} and \Cref{A.17} (for $n -\td$ instead of $\td$ and $q_n$ instead of $1-q_n$, respectively) that
\begin{align}
    I_{n,1}  &= -2  \log(n!)+2 \E \bigl[   (n -\td) \left( \log ( n- \td) -1 \right) \bigl] +2 \E[ \log( \td!)]    \notag \\
    &< -2 n \log(n) + 2n +2 n q_n \log( n q_n) -2 n q_n + 2 n ( 1 - q_n) \log(n ( 1 -q_n)) \notag \\
    &\qquad- 2 n (1-q_n)   +  2 \log( n (1 - q_n) )+C_5 \notag \\
     &= [-2 n \log(n)+ 2 n q_n \log( n ) + 2 n ( 1 - q_n) \log(n ( 1 -q_n))] \notag \\
    &\qquad  +  [2 n q_n \log(  q_n) + 2 \log( n (1 - q_n) )+C_5] \notag \\
    &\leq    2 n q_n \log(  q_n) + 2 \log( n  )+C_5 
   \label{eq:BIC_Global_Vereinf3} 
\end{align} 
for some constant $C_5 > 0$ independent of $s$ and $n$ holds, where we used that the bracket in the second last equation is negative.

Combining 
\cref{ineq:BIC_global_bin_log_exp_l,eq:BIC_Global_Vereinf3} ends up with
\begin{align*}
-2   &\E [   \log\E_{g_s'}  [ L_{M^{\prime s}_n} (\widetilde{\bp}^{\prime s} \, \vert \, \bT_n') ] ] \\
\leq &  I_{n,1} + I_{n,2} - 2 \log b' + C_2 n q_n \\
\leq & -2 \frac{n q_n}{k_n} \E[ \log L_{M^s_{k_n}} ( \widehat{\bp}_n^s  \, \vert \, \bT_n(k_n) ) ]     + 2  s \log \left( k_n \sqrt{ \frac{r}{2 \pi (r-s)}} \right)   \\
   &\qquad + n q_n  \Big(     2  \log \left( \frac{n q_n}{k_n} \right) +  \frac{ 2 \log( n)  }{n q_n} \Big)  + n q_n \;\max_{ i = 1, \ldots, 5} C_i \,  \\
=& 2 n q_n \left[ - \frac{\E[ \log L_{M^s_{k_n}} ( \widehat{\bp}_n^s \, \vert \, \bT_n(k_n) ) ]}{k_n}      +   \frac{s}{n q_n} \log \Bigl( k_n \sqrt{ \frac{r}{2 \pi (r-s)}} \Bigr)  +  \frac{  \log( n)  }{n q_n}  \right] + C n q_n ,
\end{align*}
for a constant $C > 0$ independent of $s$ and $n$.
\end{proof}

\end{bibunit}

\begin{bibunit}

\renewcommand\appendixname{Supplementary Material}

\newpage

\thispagestyle{empty}
\setcounter{page}{1}
\setcounter{section}{1}

\begin{center}

{\normalfont {\normalsize{S}\footnotesize{UPPLEMENTARY} \normalsize{M}\footnotesize{ATERIAL FOR}}} \\ [6mm]

{\normalfont \bfseries{\large
INFORMATION CRITERIA FOR THE NUMBER OF \vspace*{0.2cm} \\ DIRECTIONS OF EXTREMES IN HIGH-DIMENSIONAL DATA}}\\[6mm]

{\normalfont  {\normalsize{B}\footnotesize{Y} \normalsize{L}\footnotesize{UCAS} \normalsize{B}\footnotesize{UTSCH  AND}
 \normalsize{V}\footnotesize{ICKY} \normalsize{F}\footnotesize{ASEN}-\normalsize{H}\footnotesize{ARTMANN}}}
\end{center}

\bigskip

\section{Auxiliary results for the quasi-Akaike information criterion} \label{sec:QAIC_supp}
In this section, we present supplementary results for \Cref{sec:QAIC}.

\subsection{Proof of Lemma \ref{th:conv_QAIC_Hessematrix}}

\begin{namedthm*}{\Cref{th:conv_QAIC_Hessematrix}} \label{th:conv_QAIC_Hessematrix_supp}
Suppose the assumptions of  \Cref{th:QAIC_Likelihood_Approx} hold and $\bpb^s_n( \bcT_n  )$ is defined analog to $\bpb^s_n( \widetilde{\bcT}_n )$ in \eqref{6.1}. Then as $\ninf$,
\begin{align*}
\bY_n := &   \sqrt{k_n} \diag(p_{n,1}, \ldots, p_{n,s},\frac{\rho_n}{(r-s)}, \rho_n,\ldots,\rho_n)^{-1/2}\begin{pmatrix}
 \left(\bpb^s_n(\widetilde \bcT_n  )-\bpb^s_n( \bcT_n  )    ) \right)\\
      \left( \frac{\mathcal{T}_{n,s+1} }{k_n} - \rhob^s_{n}( \bcT_n ) \right)\\
     \vdots \\
      \left( \frac{\mathcal{T}_{n,r} }{k_n} - \rhob^s_{n}( \bcT_n ) \right)\\
\end{pmatrix}\\
            \Dconv &\;   
    \mathcal{N}_{{r+1}} \left( \mathbf{0}_{r+1}, \Sigma \right),
\end{align*}
where
\begin{align*}
    \Sigma \coloneqq \begin{pmatrix}
        2 \bI_{s+1} &   \mathbf{0}_{s \times (r-s)}\\
         \mathbf{0}_{(r-s) \times (s+1)}  & \,\, \bI_{r-s}- \frac{\mathbf{1}_{r-s} \mathbf{1}_{r-s}^\top}{r-s} 
    \end{pmatrix}.
\end{align*}

\end{namedthm*}
\begin{proof}
    
From Assumption~(\ref{Assumption:main}\ref{(A5)})  and the continuous mapping theorem we receive that
\begin{align*}
    \begin{pmatrix}
        \bI_s & \mathbf{0}_{s \times (r-s)} \\
        \mathbf{0}_{s}^\top & \frac{\mathbf{1}_{r-s}^\top}{\sqrt{r-s}}\\
         \mathbf{0}_{(r-s) \times s} & \bI_{r-s}- \frac{\mathbf{1}_{r-s} \mathbf{1}^\top_{r-s}}{r-s} 
    \end{pmatrix}
            & \sqrt{k_n} \diag(  \bp_{n}^*)^{-1/2}\left( \frac{\bcT_{n} }{k_n} - \bp_{n}^* \right)\\
            &\Dconv   
    \mathcal{N}_{{r+1}} \left( \mathbf{0}_{r+1}, \begin{pmatrix}
        \bI_{s+1} &  \mathbf{0}_{(s+1) \times (r-s)}\\
         \mathbf{0}_{(r-s) \times (s+1)}  & \,\, \bI_{r-s}- \frac{\mathbf{1}_{r-s} \mathbf{1}^\top_{r-s}}{r-s} 
    \end{pmatrix} \right).
\end{align*}
Finally, it follows from the independence of $\bcT_n $ and $\widetilde{\bcT}_n $ as well as $p_{n,j} / \rho_n \rightarrow 1, \ j > s^*$, by \Cref{Assumption:main},  that as $\ninf$,
\begin{align*}
   \bY_n   = &    \begin{pmatrix}
        \bI_s & \mathbf{0}_{s \times (r-s)} \\
        \mathbf{0}_{s}^\top & \frac{\mathbf{1}_{r-s}^\top}{\sqrt{r-s}}\\
         \mathbf{0}_{(r-s) \times s} & \bI_{r-s}- \frac{\mathbf{1}_{r-s} \mathbf{1}_{r-s}^\top}{r-s} 
    \end{pmatrix}
            \sqrt{k_n}  \diag(  \bp_{n}^*)^{-1/2}\left( \frac{\bcT_{n} }{k_n} - \bp_{n}^* \right)\\
              &\qquad - 
    \begin{pmatrix}
        \bI_s & \mathbf{0}_{s \times (r-s)} \\
        \mathbf{0}_{s}^\top & \frac{\mathbf{1}_{r-s}^\top}{\sqrt{r-s}}\\
         \mathbf{0}_{(r-s) \times s} & \mathbf{0}_{(r-s) \times (r-s)} 
    \end{pmatrix}
            \sqrt{k_n} \diag(  \bp_{n}^*)^{-1/2} 
     \left(\frac{\widetilde{\bcT}_{n} }{k_n} - \bp_{n}^* \right) + o_\P(1) \\
             \Dconv &\;   
    \mathcal{N}_{{r+1}} ( \mathbf{0}_{r+1}, \Sigma).
\end{align*}
\end{proof}

\subsection{Proof of Lemma \ref{lem:Gradient_Normal_Mult_MLE}}

\begin{namedthm*}{\Cref{lem:Gradient_Normal_Mult_MLE}}  \label{lem:Gradient_Normal_Mult_MLE_supp}
Suppose the assumptions of  \Cref{th:QAIC_Likelihood_Approx} hold and $\bpb^s_n( \bcT_n  )$ is defined analog to $\bpb^s_n( \widetilde{\bcT}_n )$ in \eqref{6.1}. 
\begin{enumerate}[(a)]
    \item 
    Then as $\ninf$,
    \begin{align*}
        \nabla \log L_{\mathcal{N}_r}  ( \bpb^s_n( \bcT_n )   \, \vert \, \bcT_n )( \bpb^s_n( \widetilde{\bcT}_n )    - \bpb^s_n( \bcT_n ))  \Pconv 0.
    \end{align*}
   \item Suppose  $\bar{\bp}_n \coloneqq (\bar{p}_{n,1}, \ldots, \bar{p}_{n,s}, \bar{\rho}_{n})^\top$ satisfies 
$$   \Vert \bar{\bp}_n - \bpb^s_n( \bcT_n )   \Vert  \leq  \Vert   \bpb^s_n( \widetilde{\bcT}_n )  - \bpb^s_n( \bcT_n )   \Vert, \quad n\in\N. $$ Then as $\ninf$,
 \begin{align*}
         ( \bpb^s_n( \widetilde{\bcT}_n )    &- \bpb^s_n( \bcT_n ))^\top \Big( \nabla^2 \log L_{\mathcal{N}_r}  ( \bar{\bp}_n \, \vert \, \bcT_n ) \\
         &+ k_n \big( \diag(p_{n,1}, \ldots, p_{n,s},\rho_n/(r-s) )^{-1}   \big)  \Big) \cdot ( \bpb^s_n( \widetilde{\bcT}_n )    - \bpb^s_n( \bcT_n )) \Pconv 0. 
   \end{align*}
\end{enumerate}
\end{namedthm*}
\begin{proof}
\begin{enumerate}[(a)]
    \item 

The derivatives of the log-likelihood function are
\begin{align*}
\frac{\partial}{\partial \pbt^s_{j}} \log L_{\mathcal{N}_r }( \bpbt_n \, \vert \, \bcT_n ) =& -  \frac{1}{2\pbt^s_{j}} - \frac{k_n}{2} \frac{(\pbt^s_{j})^2 - \frac{\cT_{n,j} ^2}{k_n^2}}{(\pbt^s_{j})^2} , \qquad j = 1, \ldots, s
\end{align*}
and
\begin{align*}
\frac{\partial}{\partial \rhobt^s} \log L_{\mathcal{N}_r }( \bpbt^s \, \vert \, \bcT_n ) =& - \frac{ ( r - s) }{2\rhobt^s} - \frac{k_n}{2 } \sum_{j=s+1}^r \frac{(\rhobt^s)^2 - \frac{\cT_{n,j} ^2}{k_n^2}}{(\rhobt^s)^2}, \quad \bpbt^s \in \R_+^{s+1}.
\end{align*}
Hence, 
 \begin{align*}
   \nabla& \log  L_{\mathcal{N}_r}  ( \bpb_n^s ( \bcT_n ) \, \vert \, \bcT_n ) (\bpb_n^s  (\widetilde{\bcT}_n )  - \bpb_n^s( \bcT_n  )) \\
    &= \sum_{j=1}^s \left( -\frac12 \frac{1}{\pb_{n,j}^s(\bcT_{n})} - \frac{k_n}{2} \frac{(\pb_{n,j}^s(\bcT_{n}) )^2 - \frac{\bcT_{n,j} ^2}{k_n^2}}{(\pb_{n,j}^s)^2} \right) \big(\pb_{n,j}^s (\widetilde{\bcT}_n )  - \pb_{n,j}^s ( \bcT_n  ) \big)\\
    &\qquad -   \left( \frac{ ( r - s) }{2\rhob^s_n( \bcT_n)} + \frac{k_n}{2 } \sum_{i=s+1}^r \frac{(\rhob^s_n( \bcT_n))^2 - \frac{\cT_{n,i} ^2}{k_n^2}}{(\rhob^s_n( \bcT_n)) ^2} \right)\big(\rhob_{n}^s(\widetilde{\bcT}_n)  - \rhob_{n}^s(\bcT_n)\big) \\
    &=  -\frac12 \sum_{j=1}^s  \frac{(\pb_{n,j}^s (\widetilde{\bcT}_n )  - \pb_{n,j}^s (\bcT_n  ))     }{\pb_{n,j}^s(\bcT_n)} 
  -\frac{ ( r - s) }{2} \frac{(\rhob_{n}^s(\widetilde{\bcT}_n)    - \rhob_{n}^s(\bcT_n))}{\rhob_{n}^s(\bcT_n)}   \\
    & \qquad - \frac{k_n}{2} \left(\rhob_{n}^s(\widetilde{\bcT}_n)  - \rhob_{n}^s(\bcT_n)\right) \sum_{i=s+1}^r \frac{(\rhob_{n}^s(\bcT_n))^2 - \frac{\cT_{n,i}^2}{k_n^2}}{(\rhob_{n}^s(\bcT_n))^2} \\
    &=:I_{n,1}+I_{n,2}+I_{n,3}. 
 \end{align*} 
 First, note that $I_{n,1}=o_{\P}(1)=I_{n,2}$
 due to \Cref{th:conv_QAIC_Hessematrix} and $\sqrt{k_n\rho_n}\to\infty$.
 Therefore, it remains to investigate $I_{n,3}$.
We define the function $g:\R^{r-s}\to\R$ as $g(\bx) \coloneqq (r-s)^2 \frac{\bx^\top \bx}{( \mathbf{1}_{r-s}^\top \bx)^2}$ with Jacobian vector 
\begin{align*}
    \nabla g(\bx) = 2 (r-s)^2 \Big( \frac{\bx^\top}{( \mathbf{1}_{r-s}^\top \bx)^2} - \frac{\bx^\top \bx \mathbf{1}_{r-s}^\top}{( \mathbf{1}_{r-s}^\top \bx)^3} \Big) \quad \text{ for }\bx \in \R^{r-s}. 
\end{align*}    
Then, $g(\mathbf{1}_{r-s}) =  r-s $ and $\nabla g ( \mathbf{1}_{r-s}) = \mathbf{0}_{r-s}^\top.$
From  Assumptions~(\ref{Assumption:main}\ref{(A5)}) we already get the asymptotic behavior 
\begin{align*}
    \sqrt{k_n \rho_n} \left( \frac{\bcT_{n, \{s+1, \ldots, r\}}}{\rho_n k_n} - \mathbf{1}_{r-s}   \right) \Dconv \mathcal{N}_{r-s}( \mathbf{0}_{r-s}, \bI_{r-s}).
\end{align*}
Then an application of the delta-method  yields
\begin{align*}
     \sqrt{k_n \rho_n} \left( g \Big(\frac{\bcT_{n, \{s+1, \ldots, r\}}(k_n)}{\rho_n k_n} \right) - g( \mathbf{1}_{r-s})   \Big) \Pconv 0
\end{align*}
or equivalently
\begin{align*}
     \sqrt{k_n \rho_n} \left( \frac{\sum_{j=s+1}^r \frac{\cT_{n,j}^2}{\rho_n^2 k_n^2}}{\Big( \sum_{j=s+1}^r  \frac{\cT_{n,j}}{\rho_n k_n (r-s)}   \Big)^2 } -  (r-s)  \right) &=  \sqrt{k_n \rho_n} \Big( \sum_{j=s+1}^r\frac{ \frac{\cT_{n,j}^2}{k_n^2}}{ (\rhob_n^s(\cT_n))^2 } -  (r-s)  \Big)\\
     &=o_{\P}(1). 
\end{align*}
On the other hand, \Cref{th:conv_QAIC_Hessematrix} implies that
\begin{eqnarray*}
    \sqrt{ \frac{k_n}{\rho_n}} (\rhob_{n}^s (\widetilde{\bcT}_n)   - \rhob_{n}^s (\bcT_n ) )=O_{\P}(1). 
\end{eqnarray*}
Finally, this results in
\begin{align*}
    I_{n,3}&=- \frac{k_n}{2}(\rhob_{n}^s(\widetilde{\bcT}_n)   - \rhob_{n}^s(\bcT_n))  \sum_{j=s+1}^r \frac{(\rhob_{n}^s(\bcT_n))^2 - \frac{\cT_{n,j}^2}{k_n^2}}{(\rhob_{n}^s(\bcT_n))^2} \\
    &= \frac{1}{2} \sqrt{ \frac{k_n}{\rho_n}} (\rhob_{n}^s (\widetilde{\bcT}_n)   - \rhob_{n}^s (\bcT_n ) )    \sqrt{k_n \rho_n}  \left(   \sum_{j=s+1}^r \frac{  \frac{\cT_{n,j}^2}{k_n^2}}{(\rhob_{n}^s(\bcT_n))^2} -(r-s)\right)=o_{\P}(1).
\end{align*}

\item The Hessian matrix of the log-likelihood function is
\begin{align*}
    \nabla^2 \log L_{\mathcal{N}_r}  ( \bpbt^s \, \vert \, \bcT_n ) 
    &=  \diag \Bigg( \frac{1}{2 (\pbt_1^s)^2} - k_n \frac{\cT_{n,1}^2}{k_n^2}\frac{1}{(\pbt_{1}^s)^3} , \ldots, \frac{1}{2 (\pbt_s^s)^2} - k_n \frac{\cT_{n,s}^2}{k_n^2}\frac{1}{(\pbt_{s}^s)^3}, \\
    & \qquad \qquad  \frac{(r-s)}{2  (\rhobt^s)^2} -  k_n  \sum_{j=s+1}^r \frac{\cT_{n,j}^2}{k_n^2} \frac{1}{(\rhobt^s)^3}   \Bigg)  , \quad \bpbt^s \in \R_+^{s+1}.
\end{align*}
Let $\bar{\bp}_n \coloneqq (\bar{p}_{n,1}, \ldots, \bar{p}_{n,s}, \bar{\rho}_{n})^\top$ with  
 $   \Vert \bar{\bp}_n - \bpb^s_n( \bcT_n )   \Vert  \leq  \Vert  \bpb^s_n( \widetilde{\bcT}_n ) - \bpb^s_n( \bcT_n )   \Vert $. Then,
\begin{align*}
     \nabla^2 \log L_{\mathcal{N}_r}  &( \bar{\bp}_n \, \vert \, \bcT_n ) +  k_n \diag(p_{n,1}, \ldots, p_{n,s},\rho_n/(r-s) )^{-1}  \\
     &= \diag \Big(\frac{1}{2 \bar{p}_{n,1}^2} - k_n \frac{\mathcal{T}_{n,1} ^2}{k_n^2} \frac{1}{\bar{p}_{n,1}^3} + \frac{k_n}{p_{n,1}}, \ldots, \frac{1}{2 \bar{p}_{n,s}^2} - k_n \frac{\mathcal{T}_{n,s} ^2}{k_n^2} \frac{1}{\bar{p}_{n,s}^3} + \frac{k_n}{p_{n,s}},  \\
    & \qquad  \qquad \frac{(r-s)}{2} \frac{1}{ \bar{\rho}_n^2} - k_n \sum_{j=s+1}^r \frac{\mathcal{T}_{n,j} ^2}{k_n^2} \frac{1}{\bar{\rho}_n^3} + \frac{k_n(r-s) }{\rho_n}  \Big)\\
    &\eqqcolon \diag(B_{n}(1),\ldots,B_{n}(s), B_{n}(s+1)). 
\end{align*}
Since $\bar{p}_{n,j}/p_{n,j} \Pconv 1,\; j = 1, \ldots, s$,  we receive for the entries  $B_{n}(j),\; j = 1, \ldots, s$ that
 \begin{align}
     \frac{p_{n,j}}{k_n}B_{n}(j)=\frac{p_{n,j}}{k_n} \left(  \frac{1}{2 \bar{p}_{n,j}^2} - k_n \frac{\mathcal{T}_{n,j} ^2}{k_n^2} \frac{1}{\bar{p}_{n,j}^3} + \frac{k_n}{p_{n,j}} \right) = \frac{p_{n,j}}{2 k_n \bar{p}_{n,j}^2} -   \frac{\mathcal{T}_{n,j}^2 p_{n,j} }{k_n^2 \bar{p}_{n,j}^3} + 1 \Pconv 0 .\label{conv:QAIC_Vorfaktor3}
 \end{align}
Similarly we receive with $\rho_n k_n \rightarrow \infty$ and $\bar{\rho}_{n}/\rho_{n} \Pconv 1$ for the entry $B_{n}(s+1)$   that
\begin{align}
     \frac{{\rho_n}}{k_n} B_{n}(s+1)=  \frac{{\rho_n}}{k_n}  &\left(\frac{(r-s)}{2 } \frac{1}{ \bar{\rho}_n^2} - k_n \sum_{j=s+1}^r \frac{\mathcal{T}_{n,j} ^2}{k_n^2} \frac{1}{\bar{\rho}_n^3} +  \frac{k_n (r-s)}{\rho_n} \right)   \Pconv 0. \label{eq:QAIC_Mult_Sum4}
\end{align}
Additionally, due to \Cref{th:conv_QAIC_Hessematrix} we have as $n\to\infty$,
\begin{align}
    \sqrt{\frac{k_n}{p_{n,j} }}\bigl(\pb_{n,j}^s (\widetilde{\bcT}_n )  - \pb_{n,j}^s( \bcT_n  ) \bigr) 
    &= O_\P(1) \quad \text{ and } \quad
      \sqrt{\frac{k_n}{\rho_{n} }} \bigl( \rhob_{n}^s (\widetilde{\bcT}_n ) -\rhob_{n}^s( \bcT_n  )   \bigr)
     = O_\P(1). \label{conv:QAIC_Vorfaktor4}
\end{align}
Therefore, Slutzky's lemma, \eqref{conv:QAIC_Vorfaktor3}, \eqref{eq:QAIC_Mult_Sum4} and \eqref{conv:QAIC_Vorfaktor4} yield
\begin{align*}
    \big(\bpb^s_n&( \widetilde{\bcT}_n ) - \bpb^s_n( {\bcT}_n )\big)^\top \diag(B_{n}(1),\ldots,B_{n}(s), B_{n}(s+1)) \big(\bpb^s_n( \widetilde{\bcT}_n ) - \bpb^s_n( {\bcT}_n )\big) \\
    &=  \sum_{j=1}^s \left( \sqrt{\frac{k_n}{p_{n,j} }} \big(\pb_{n,j}^s (\widetilde{\bcT}_n )  -  \pb_{n,j}^s( \bcT_n  ) \big) \right)^2  \left( \frac{p_{n,j}}{k_n}B_{n}(j)\right) 
    \\
    & \quad + \left(  \sqrt{\frac{k_n}{\rho_{n} }} \big( \rhob_{n}^s (\widetilde{\bcT}_n )- \rhob_{n}^s( \bcT_n  )   \big)  \right)^2  \left(\frac{{\rho_n}}{k_n}  B_{n}(s+1)\right) \\ 
    &\Pconv 0,
\end{align*}
as $n\to\infty$, the statement.
\end{enumerate}
\end{proof}

\section{Auxiliary results for the mean squared error information criterion} \label{sec:MSEIC_supp}
In this section, we present supplementary results for \Cref{sec:MSEIC}.

\subsection{Proof of Lemma \ref{lem:MSEIC_glob1}}
\begin{namedthm*}{\Cref{lem:MSEIC_glob1}}
  Suppose assumptions (\ref{asu:Model_global}\ref{asu:BIC_global_expectation_local}) and (\ref{asu:Model_global}\ref{asu:BIC_global_expectation_local2}) hold. Then for $\bp' \in \R_+^r$ the asymptotic behavior 
    \begin{align*}
     \E  &\left[   \left\Vert \sqrt{n- \td}  \diag( \bp' )^{-1/2}  \left( \frac{\bT_{n, \{1, \ldots, r\}}'}{n- \td} -    \bp' \right)  \right\Vert^2_2  \right]  \\
     & \qquad =  n q_n \left( \frac{1 }{k_n} \E[ \ell^2 \big( \bp' \vert \bT_n(k_n) \big)  ] + o \left(  \frac{1}{n q_n} \right) \right) 
\end{align*}
as $n\to\infty$ holds.
\end{namedthm*}
\begin{proof}
Under the assumptions (\ref{asu:Model_global}\ref{asu:BIC_global_expectation_local}) and (\ref{asu:Model_global}\ref{asu:BIC_global_expectation_local2}) we get 
\begin{align*}
    \E & \left[ \left( \frac{T_{n,j}'}{ (n - \td) }  - p_j'  \right)^2 \Big\vert \td \right] \\
    &= \E \left[  \frac{(T_{n,j}')^2}{ (n - \td)^2 }  - 2  p_j' \frac{T_{n,j}'}{(n - \td) } + (p_j')^2 \Big\vert \td \right]\\
    &= \E \left[  \frac{(T_{n,j}(k_n))^2}{k_n^2 }  - 2  p_j'  \frac{T_{n,j}(k_n)}{k_n } + (p_j')^2   \right] + o_\P \left( \frac{1}{n - \td} \right)\\
    &=  \E \left[ \left(  \frac{T_{n,j}(k_n)}{k_n } -  p_j'  \right)^2\right] + o_\P \left( \frac{1}{n - \td} \right).
\end{align*}
Hence,
\begin{align*}
     \E & \left[   \left\Vert \sqrt{n- \td}  \diag( \bp' )^{-1/2}  \left( \frac{\bT_{n, \{1, \ldots, r\}}'}{n- \td} -    \bp' \right)  \right\Vert^2_2  \right] \nonumber \\
    &=  \E \left[ (n- \td) \sum_{j=1}^{r} \frac{1}{p_j'}  \E  \left[ \left( \frac{T_{n,j}'}{ (n - \td) }  - p_j'  \right)^2 \Big\vert \td \right]      \right] \nonumber \\
    &=  n q_n \left(   \sum_{j=1}^{r} \frac{1}{p_j'} \E \left[ \left(  \frac{T_{n,j}(k_n)}{k_n } -  p_j'  \right)^2\right] + o \left( \frac{1}{n q_n} \right) \right) \nonumber \\
    &=  n q_n \left( \frac{1}{k_n} \E[ \ell^2 \big( \bp' \vert \bT_n(k_n) \big)  ] + o \left(  \frac{1}{n q_n} \right)  \right). 
\end{align*}
\end{proof}

\subsection{Proof of Lemma \ref{lem:MSEIC_glob2}}

\begin{namedthm*}{\Cref{lem:MSEIC_glob2}}
    For $q' \in (0,1)$ the equality
        \begin{align*}
     \E &\left[   \left\Vert \sqrt{n} ( q' ( 1-q') )^{-1/2} \left(  \frac{\td}{n}  - (1-q') \right) \right\Vert^2_2  \right] = n q_n  \left( \frac{(1-q_n) }{n q' ( 1-q') } +  \frac{(q'-q_n)^2 }{q_n q' ( 1-q') } \right)  
\end{align*}
holds.
\end{namedthm*}
\begin{proof}
A straightforward calculation gives with
\begin{align*}
      \E &\left[   \left\Vert \sqrt{n} ( q' ( 1-q') )^{-1/2} \left(  \frac{\td}{n}  - (1-q') \right) \right\Vert^2_2  \right] \nonumber \\
   &=   \frac{n  }{ ( q' ( 1-q') )} \left( \frac{n q_n (1 - q_n)}{n^2} + \frac{n^2 (1-q_n)^2 }{n^2} - 2 ( 1 - q') \frac{n ( 1 - q_n)}{n} + ( 1 - q')^2  \right) \nonumber \\
   &= \frac{q_n(1-q_n) }{q' ( 1-q') } + n  \frac{(1-q_n)^2 - 2 ( 1 - q')(1-q_n) + ( 1 - q')^2}{q' ( 1-q') }\nonumber  \\
   &= \frac{q_n(1-q_n) }{q' ( 1-q') } + n  \frac{(q'-q_n)^2 }{q' ( 1-q') } \nonumber  \\
      &= q_n  \left( \frac{(1-q_n)}{q' ( 1-q') } + n  \frac{(q'-q_n)^2 }{q_n q' ( 1-q') } \right) \nonumber \\
      &= n q_n  \left( \frac{(1-q_n) }{n q' ( 1-q') } +  \frac{(q'-q_n)^2 }{q_n q' ( 1-q') } \right)  
\end{align*}
the statement.
\end{proof}

\section{Auxiliary results for the Bayesian information criterion} \label{sec:BIC_supp}
In this section, we present supplementary results for \Cref{sec:BIC}.

\subsection{Proof of Lemma \ref{th:BIC_lokal_Taylor_Approx} and Lemma \ref{th:BIC_lokal_Eigenwerte_Absch}}  \label{sec:BIC_local_supp}
First, we provide the proofs of the auxiliary results of \Cref{proof:BIC_post_prob} in this subsection.

\begin{namedthm*}{\Cref{th:BIC_lokal_Taylor_Approx}} \label{th:BIC_lokal_Taylor_Approx_supp} 
Let the assumptions of Theorem \ref{th:BIC_post_prob} hold. Define the ball $$U_{\varepsilon_{n, \gamma}}(\widehat{\bp}_n^s) \coloneqq \{ \widetilde{\bp}^s \in\Theta_s: \Vert \widetilde{\bp}^s - \widehat{\bp}_n^s \Vert_2 < \varepsilon_{n, \gamma} \}$$ with radius $\varepsilon_{n, \gamma} \coloneqq {(\rho_n)^\gamma}/{2}$ for $\gamma \ge 4/3$ around $\widehat{\bp}_n^s$.
Then the following statement holds
    $\begin{aligned}[t]
   \sup_{ \widetilde{\bp}^s \in U_{\varepsilon_{n, \gamma}}(\widehat{\bp}_n^s) }  \bigg \vert  &   \log L_{M^s_{k_n}}(\widetilde{\bp}^s\, \vert\, \bT_n(k_n)) - \log L_{M^s_{k_n}}(\widehat{\bp}_n^s\, \vert \, \bT_n(k_n)) \\
   &  -  \frac12 (\widetilde{\bp}^s - \widehat{\bp}_n^s)^\top    \nabla^2 \log L_{M^s_{k_n}}(\widehat{\bp}^s_n\, \vert \, \bT_n(k_n)) (\widetilde{\bp}^s - \widehat{\bp}_n^s) \bigg \vert   = o_\P(1).
\end{aligned}$
\end{namedthm*}
\begin{proof}
 First, we apply
a multivariate Taylor expansion to the log-likelihood function $\log L_{M^s_{k_n}} (\cdot \mid \bT_n(k_n))$  around the MLE $\widehat{\bp}_n^s$ at $\widetilde \bp^s$ analog to Lemma 2 of \citet{meyer_muscle23} (based on a generalization of Cauchy's Mean Value Theorem (see \citet{Cauchy_Mean_Value})) which gives the existence of a constant  $\theta_n \in (0,1)$ such that 
\begin{align*}
     \log L_{M^s_{k_n}}&(\widetilde{\bp}^s\, \vert \, \bT_n(k_n)) \nonumber \\
    ={} & \log L_{M^s_{k_n}}(\widehat{\bp}_n^s\, \vert \, \bT_n(k_n)) + (\widetilde{\bp}^s - \widehat{\bp}_n^s)^\top \nabla \log L_{M^s_{k_n}}(\widehat{\bp}_n^s\, \vert \, \bT_n(k_n)) \nonumber \\
    &  + \frac12 (\widetilde{\bp}^s - \widehat{\bp}_n^s)^\top  \nabla^2 \log L_{M^s_{k_n}}(\theta_n \widehat{\bp}_n^s + (1-\theta_n) \widetilde{\bp}^s \, \vert \, \bT_n(k_n))(\widetilde{\bp}^s - \widehat{\bp}_n^s) \nonumber\\
    ={} &  \log L_{M^s_{k_n}}(\widehat{\bp}_n^s\, \vert \, \bT_n(k_n)) \nonumber \\
    & +  \frac{(\widetilde{\bp}^s - \widehat{\bp}_n^s)^\top}{2}    \nabla^2 \log L_{M^s_{k_n}}(\theta_n \widehat{\bp}_n^s + (1-\theta_n) \widetilde{\bp}^s \, \vert \, \bT_n(k_n)) (\widetilde{\bp}^s - \widehat{\bp}_n^s).
\end{align*}

Thus, we receive for the left hand side in (a) that
\begin{align}
    &\bigg \vert   \log L_{M^s_{k_n}}(\widetilde{\bp}^s\, \vert \, \bT_n(k_n)) - \log L_{M^s_{k_n}}(\widehat{\bp}_n^s\, \vert \, \bT_n(k_n))  \nonumber \\
    & \qquad-  \frac12 (\widetilde{\bp}^s - \widehat{\bp}_n^s)^\top    \nabla^2 \log L_{M^s_{k_n}}(\widehat{\bp}_n^s\, \vert \, \bT_n(k_n)) (\widetilde{\bp}^s - \widehat{\bp}_n^s) \bigg \vert  \nonumber \\
    &= \frac12 \bigg \vert    (\widetilde{\bp}^s - \widehat{\bp}_n^s)^\top \bigg(   \nabla^2 \log L_{M^s_{k_n}}(\theta_n \widehat{\bp}_n^s + (1-\theta_n) \widetilde{\bp}^s \, \vert \, \bT_n(k_n)) \nonumber \\
    & \qquad - \nabla^2 \log L_{M^s_{k_n}}(\widetilde{\bp}^s\, \vert \, \bT_n(k_n)) \bigg) (\widetilde{\bp}^s - \widehat{\bp}_n^s)  \bigg \vert .  \label{eq:BIC_local_Betrag}
\end{align} 
Therefore, to prove the statement, we show that the right side is $o_\P(1)$.
Inserting the derivatives of the log-likelihood function
\begin{align}
\nabla \log L_{M^s_{k_n}}(\widetilde{\bp}^s\, \vert \, \bT_n(k_n)) 
&= \left( \begin{array}{crl}
 \frac{ T_{n,1}}{\widetilde{p}_{1}^s} - \frac{\sum_{j = s+1}^{2^d-1} T_{n,j}}{1 -\sum_{j = 1}^s \widetilde{p}^s_{j} }  \\
\vdots \\
 \frac{T_{n,s}}{\widetilde{p}_{s}^s} - \frac{\sum_{j = s+1}^{2^d-1} T_{n,j}}{1 -\sum_{j = 1}^s \widetilde{p}^s_{j} }  \\
\end{array} \right), \label{BIC:Likelihood_Derivative} \\
\nabla^2  \log L_{M^s_{k_n}}(\widetilde{\bp}^s\, \vert \, \bT_n(k_n)) &=  - \diag  \biggl( \frac{T_{n,1}(k_n) }{ (\widetilde{p}^s_{1})^2}, \ldots, \frac{T_{n,s}(k_n) }{ (\widetilde{p}^s_{s})^2} \biggr) - \frac{\sum_{j = s+1}^{r}  T_{n,j}}{\big( 1 -\sum_{j = 1}^s \widetilde{p}_{j}^s \big)^2 } \cdot \mathbf{1}_s \cdot \mathbf{1}_s^\top,  \nonumber
\end{align}
and applying the triangle inequality yields
\begin{align*}
    & \bigg\vert (\widetilde{\bp}^s - \widehat{\bp}_n^s)^\top \biggl(   \nabla^2  \log L_{M^s_{k_n}}(\theta_n \widehat{\bp}_n^s + (1-\theta_n) \widetilde{\bp}^s \, \vert \, \bT_n(k_n)) \nonumber \\
    & \hspace*{5cm} - \nabla^2 \log L_{M^s_{k_n}}(\widehat{\bp}_n^s \, \vert \, \bT_n(k_n)) \bigg) (\widetilde{\bp}^s - \widehat{\bp}_n^s)  \bigg\vert \\
    &\leq \bigg\vert  (\widetilde{\bp}^s - \widehat{\bp}_n^s)^\top \bigg\{  \diag \biggl( \frac{T_{n,1}(k_n) }{(\theta_n \widehat{p}_{n,1}^s + (1-\theta_n) \widetilde{p}_1^s)^2}, \ldots, \frac{T_{n,s}(k_n) }{(\theta_n \widehat{p}_{n,s}^s + (1-\theta_n) \widetilde{p}_s^s)^2} \biggr) \\
    & \hspace*{5cm} - \diag \biggl( \frac{T_{n,1}(k_n) }{(\widehat{p}_{n,1}^s)^2}, \ldots, \frac{T_{n,s}(k_n) }{ ( \widehat{p}_{n,s}^s )^2} \biggr) \bigg\} (\widetilde{\bp}^s - \widehat{\bp}_n^s) \bigg\vert  \\
    &\qquad +  \Bigg\vert  (\widetilde{\bp}^s - \widehat{\bp}_n^s)^\top \Bigg\{ \frac{\sum_{j = s+1}^{r}  T_{n,j}}{\left( 1 -\sum_{j = 1}^s  (\theta_n \widehat{p}_{n,j}^s + (1-\theta_n) \widetilde{p}_j^s) \right)^2 }  \nonumber \\
    & \hspace*{5cm} - \frac{\sum_{j = s+1}^{r}  T_{n,j}}{\left( 1 -\sum_{j = 1}^s  \widehat{p}_{n,j}^s  \right)^2 } \Bigg\} \cdot \mathbf{1}_s \cdot \mathbf{1}_s^\top (\widetilde{\bp}^s - \widehat{\bp}_n^s)  \Bigg\vert \\
    &\eqqcolon I_1(\widetilde{\bp}^s) + I_2(\widetilde{\bp}^s).
\end{align*}
In the following we only show that $I_2(\widetilde{\bp}^s)$ is 
uniformly $o_\P(1)$; the calculation for  $I_1(\widetilde{\bp}^s)$ is similar but with a faster rate, since $p_{j} > 0, \; j = 1, \ldots, s^*$ and $\rho_n \rightarrow 0$.  Therefore,
an application of the  mean value theorem to the function $x \mapsto 1/x^2$ yields
\begin{align}
    I_2(\widetilde{\bp}^s) 
    &=  k_n \Vert \widetilde{\bp}^s - \widehat{\bp}_n^s \Vert_1^2 \Bigg\vert  \frac{\sum_{j = s+1}^{r}  \frac{T_{n,j}(k_n)}{k_n}  }{\left( 1 -\sum_{j = 1}^s  (\theta_n \widehat{p}_{n,j}^s + (1-\theta_n) \widetilde{p}_j^s) \right)^2 }  - \frac{\sum_{j = s+1}^{r}  \frac{T_{n,j}(k_n)}{k_n}}{\left( 1 -\sum_{j = 1}^s  \widehat{p}_{n,j}^s  \right)^2 }  \Bigg\vert  \nonumber \\
    &= k_n \Vert \widetilde{\bp}^s - \widehat{\bp}_n^s \Vert_1^2 (r-s) \widehat\rho_n^s  \bigg\vert  \frac{1  }{\left( \theta_n \widehat\rho_n^s + (1-\theta_n) \widetilde{\rho}^s \right)^2 }  - \frac{1}{\left( \widehat\rho_n^s  \right)^2 }  \bigg\vert  \nonumber \\
    &\le k_n \Vert \widetilde{\bp}^s - \widehat{\bp}_n^s \Vert_1^2 (r-s) \widehat\rho_n^s  \frac{2 (1 - \theta_n) \vert \widehat\rho_n^s - \widetilde{\rho}^s  \vert }{ \min(\vert \widehat\rho_n^s \vert, \vert \widetilde{\rho}^s   \vert)^3} \label{eq:BIC_local_I2} 
\end{align}
Since $\widetilde{\bp}^s \in U_{\varepsilon_{n, \gamma}}(\widehat{\bp}_n^s)$ and $\widehat\rho_n^s=O_{\P}(\rho_n)$ we obtain
\begin{align}
    \sup_{\widetilde{\bp}^s \in U_{\varepsilon_{n, \gamma}}(\widehat{\bp}_n^s) } I_2(\widetilde{\bp}^s)
&= O_\P( k_n \varepsilon_{n,\gamma}^3 \rho_n^{-2}). \nonumber 
\end{align}
Finally, $\varepsilon_{n, \gamma} = {(\rho_n)^{\gamma}}/{2}$ 
and $ k_n (\rho_n)^{3 \gamma - 2} \leq k_n (\rho_n)^{2} \rightarrow 0$ (due to  Assumption (\ref{asu:BIC_local}\ref{asu:BIC_local1})) which results in 
 the uniform convergence of  $\sup_{\widetilde{\bp}^s \in U_{\varepsilon_{n, \gamma}}(\widehat{\bp}_n^s) } I_2 (\widetilde{\bp}^s) \Pconv 0$ and the statement follows.
\end{proof}

Next, we derive boundaries for the eigenvalues of the second-order derivative of the log-likelihood function.

\begin{namedthm*}{\Cref{th:BIC_lokal_Eigenwerte_Absch}} \label{th:BIC_lokal_Eigenwerte_Absch_supp} 
Let the assumptions of Theorem \ref{th:BIC_post_prob} hold.
Define $\lambda_{n,2} \coloneqq \frac{k_n }{T_{n,1}(k_n)}$ and $\lambda_{n,1} \coloneqq \frac{k_n }{T_{n,s}(k_n)} + \frac{s k_n}{  \sum_{j = s+1}^{r} T_{n,j} }$. For $\widetilde{\bp}^s \in \Theta_s$ we have on the one hand,
    \begin{align*}
    \lambda_{n,2} & (\widetilde{\bp}^s - \widehat{\bp}_n^s)^\top  (\widetilde{\bp}^s - \widehat{\bp}_n^s) \leq(\widetilde{\bp}^s - \widehat{\bp}_n^s)^\top  \frac{-1}{k_n}  \nabla^2 \log L_{M^s_{k_n}}(\widehat{\bp}_n^s\, \vert \, \bT_n(k_n)) (\widetilde{\bp}^s - \widehat{\bp}_n^s) \quad \P \text{-a.s.}
\end{align*}
and on the other hand,
   \begin{align*}
    \lambda_{n,1}  (\widetilde{\bp}^s - \widehat{\bp}_n^s)^\top (\widetilde{\bp}^s - \widehat{\bp}_n^s) \geq (\widetilde{\bp}^s - \widehat{\bp}_n^s)^\top  \frac{-1}{k_n}  \nabla^2 \log L_{M^s_{k_n}}(\widehat{\bp}_n^s\, \vert \, \bT_n(k_n)) (\widetilde{\bp}^s - \widehat{\bp}_n^s)  \quad \P \text{-a.s.}
\end{align*}
\end{namedthm*}
\begin{proof}

Let $\widetilde{\bp}^s \in \Theta_s$. 
Inserting the MLE $\widehat{\bp}_n^s$ in the second order derivative in \cref{BIC:Likelihood_Derivative} yields
\begin{align*}
     \frac{-1}{k_n}  \nabla^2 \log L_{M^s_{k_n}}(\widehat{\bp}_n^s\, \vert \, \bT_n(k_n)) &= \diag \Bigl( \frac{k_n }{T_{n,1}(k_n)}, \ldots, \frac{k_n }{T_{n,s}(k_n)} \Bigr) + \frac{k_n}{  \sum_{j = s+1}^{r} T_{n,j} } \cdot \mathbf{1}_s \! \cdot \! \mathbf{1}_s^\top \\
     &\eqqcolon M_n + N_n.
\end{align*}
The eigenvalues of $M_n$ and $N_n$ are
\begin{equation*}
    \mu_i =\frac{k_n }{T_{n,s-i+1}(k_n)},\quad i=1, \ldots, s,
\end{equation*}
and
 \begin{equation*}
     \nu_1 = \frac{s k_n }{T_{n,s}(k_n)} \quad \text{and} \quad \nu_i  = 0,\quad i = 2,\ldots,s,
 \end{equation*}
 respectively.
By $\lambda_1, \ldots, \lambda_s$ with $\lambda_1 \ge \cdots \ge \lambda_s$ we denote the ordered eigenvalues of $ M_n + N_n.$
Then Weyl's inequality (cf.\ \citet{weyl_ineq}, p. 239, Theorem 4.3.1) and Assumption (A2) yield  
\begin{align} \label{ineq:det_V_ln}
    \lambda_{n,2} =\frac{k_n }{T_{n,1}(k_n)} = \mu_s \leq \lambda_s \leq \lambda_1 \leq \mu_1 + \nu_1 = \frac{k_n }{T_{n,s}(k_n)} + \frac{s k_n}{  \sum_{j = s+1}^{r} T_{n,j} } = \lambda_{n,1}.
\end{align}
An application of  A.2.5 in \citet{multivariate_statistics} and inequality (\ref{ineq:det_V_ln}) give then with
\begin{align*}
    \lambda_{n,2}  (\widetilde{\bp}^s - \widehat{\bp}_n^s)^\top  (\widetilde{\bp}^s - \widehat{\bp}_n^s) 
    &\leq  (\widetilde{\bp}^s - \widehat{\bp}_n^s)^\top  \frac{-1}{k_n}  \nabla^2 \log L_{M^s_{k_n}}(\widehat{\bp}_n^s\, \vert \, \bT_n(k_n)) (\widetilde{\bp}^s - \widehat{\bp}_n^s) \\
     &\leq  \lambda_{n,1}  (\widetilde{\bp}^s - \widehat{\bp}_n^s)^\top (\widetilde{\bp}^s - \widehat{\bp}_n^s) 
\end{align*} 
the statement.
\end{proof}

\subsection{Proof of Proposition \ref{th:BIC_global_Mult_Approx}}

\begin{namedthm*}{\Cref{th:BIC_global_Mult_Approx}} \label{th:BIC_global_Mult_Approx_supp}
Under Assumptions~(\ref{asu:Model_global}\ref{asu:BIC_global_expectation_local}), (\ref{asu:Model_global}\ref{asu:BIC_glob_qn_kn}) and (\ref{asu:BIC_global}\ref{E1})  the asymptotic upper bound as $n\to\infty$,
      \begin{align*}
    -2  \E&\big[  \log\E_{\lambda}  [ L_{M^s_{n-T_{n,2^d}'}}( \widetilde{\bp}^s \, \vert \, \bT_{n, \{1, \ldots, r \}}') ] \big] \notag\\
    \leq & -2 \E \bigl[  \log \big( (n - T_{n,2^d}') ! \big)  - (n -\td) \left( \log ( n- \td) - 1 \right) \bigl]  \\
    &\, -2 \frac{n q_n}{k_n} \E[ \log L_{M^s_{k_n}} ( \widehat{\bp}_n^s(\bT_n(k_n)) \, \vert \, \bT_n(k_n) ) ]     + 2  s   \log \Big(  k_n   \sqrt{ \frac{r}{2 \pi (r-s)}} \Big)    + C \log(n q_n), 
\end{align*}
for a constant $C > 0$ independent of $s$ and $n$, holds. 
\end{namedthm*}
\begin{proof}
Assumption (\ref{asu:BIC_global}\ref{E1}) says that 
\begin{align*}
   \E   \Bigl[  -2 \log& \E_{\lambda}  [ L_{M^s_{n-T_{n,2^d}'}} ( \widetilde{\bp}^s \, \vert \, \bT_{n, \{1, \ldots, r \}}') ]  \Bigr] \notag\\
    \leq \, &  \, \E \Big[  \E \Big[ -2 \log L_{M^s_{n-T_{n,2^d}'}} ( \widehat{\bp}_n^s(\bT_{n, \{1, \ldots, r \}}') \, \vert \, \bT_{n, \{1, \ldots, r \}}') \Big| \td  \Big]  \Big]\nonumber\\
    & \, + 2 s \E \Big[  \log \left( (n - \td ) \sqrt{ \frac{r}{r-s}} \right)  \Big] - s \log(2 \pi)   + o(1). 
\end{align*}
First, we find an upper bound for the first term.
Therefore, note that
for $j = 1, \ldots, s$ the equality
\begin{align}
    \E \Big[ T_{n,j}'& \log \big( \widehat{p}_{n,j}^s(\bT_{n, \{1, \ldots, r \}}')  \big) \Big| \td  \Big] \nonumber \\
    &=\E \Big[ T_{n,j}' \log \Big( \frac{T_{n,j}' }{n - \td}  \Big) \Big| \td  \Big]\nonumber \\
    &=\E  [ T_{n,j}' \log ( T_{n,j}' )  | \td  ] - \E  [ T_{n,j}' \log ( n - \td  )  | \td   ] \nonumber 
\intertext{holds. An application of \cref{A.17} in the first step (which holds as well in analog form for $ T_{n,j}'$) and Assumption (\ref{asu:Model_global}\ref{asu:BIC_global_expectation_local}) in the second step give then}
    &=\E [ T_{n,j}' | \td  ]  \log ( \E [T_{n,j}' | \td  ])  - \log ( n - \td  ) \E [ T_{n,j}' | \td  ] + O_\P(1)\nonumber \\
    &= \frac{n - \td}{k_n} \E \Big[ T_{n,j}(k_n) \log \Big( \frac{\E[ T_{n,j}(k_n) ] }{k_n}   \Big) \Big] + \frac{n - \td}{k_n} \E[ T_{n,j}(k_n)] \log( n - \td)\nonumber\\
    & \quad - \frac{n - \td}{k_n}  \E \Big[ T_{n,j}(k_n) \Big] \log ( n - \td   ) + O_\P(1) \nonumber\\
    &= \frac{n - \td}{k_n} \E \Big[ T_{n,j}(k_n) \log \Big( \frac{ \E[ T_{n,j}(k_n) ]  }{k_n}   \Big) \Big]  + O_\P(1) ,\nonumber
 \intertext{where we used in the calculations as well that $(n - \td)/k_n=O_\P(1)$ due Assumption (\ref{asu:Model_global}\ref{asu:BIC_glob_qn_kn}). Finally, we apply again \cref{A.17} to receive }
    &= \frac{n - \td}{k_n} \E \Big[ T_{n,j}(k_n) \log \Big( \frac{T_{n,j}(k_n) }{k_n}  \Big)   \Big] + O_\P(1)\nonumber\\
    &= \frac{n - \td}{k_n} \E \Big[ T_{n,j}(k_n) \log \big( \widehat p_{n,j}^s  \big)   \Big] + O_\P(1). \label{eq:lem51}
\end{align}
Similarly, we obtain as well
\begin{align}
    \sum_{j=s+1}^r \E \Big[ &T_{n,j}' \log \big( \widehat{\rho}_{n}^s(\bT_{n, \{1, \ldots, r \}}')  \big) \Big| \td  \Big] \nonumber\\
    &= \frac{n - \td}{k_n} \sum_{j=s+1}^r \E \Big[ T_{n,j}(k_n) \log \big( \widehat \rho_{n}^s  \big)   \Big] + O_\P(1). \label{eq:lem52}
\end{align}
A consequence of the log-likelihood function (cf. \eqref{eq:logLikelihood}), \eqref{eq:lem51} and \eqref{eq:lem52} is then
\begin{align}
    \E  & \Big[ -2\log L_{M^s_{n-T_{n,2^d}'}} ( \widehat{\bp}_n^s(\bT_{n, \{1, \ldots, r \}}') \, \vert \, \bT_{n, \{1, \ldots, r \}}') \Big| \td  \Big] \nonumber\\
    &= -2 \log \big( (n - T_{n,2^d}') ! \big) + 2 \sum_{j=1}^{r} \E[ \log( T_{n,j}'!) | \td ] +2 \frac{n - \td}{k_n} \log(k_n!)  \nonumber\\
    & \quad -2 \frac{n - \td}{k_n} \log(k_n!)  -2 \frac{n - \td}{k_n}\sum_{j=1}^s  \E \Big[ T_{n,j}(k_n) \log \big( \widehat p_{n,j}^s  \big)   \Big] \nonumber\\
    & \quad -2   \frac{n - \td}{k_n} \sum_{j=s+1}^r \E \Big[ T_{n,j}(k_n) \log \big( \widehat \rho_{n}^s  \big)   \Big]  + O_\P(1). \label{eq:lem53}
\end{align}
By the last equality on page 28 in  \citet{meyer_muscle23} and $\sum_{j=1}^{r} T_{n,j}(k_n) = k_n$ we receive that
\begin{align}  \label{eq:lem54}
    \sum_{j=1}^{r} &\E[ \log( T_{n,j}'!) | \td ]  \\
    &\le \frac{n - \td}{k_n} \sum_{j=1}^r \E [ \log( T_{n,j}(k_n)!)] + (n - \td) \log \Big( \frac{n - \td}{k_n} \Big)  + C_1  \log(n - \td) \nonumber 
\end{align}
and
\begin{align}
    2(n - \td) & \log \Big( \frac{n - \td}{k_n} \Big) +2 \frac{n - \td}{k_n} \log(k_n!) \nonumber \\
     &\le 2(n - \td)  ( \log(n - \td) - 1 ) +  C_2 \log(n - \td), \label{eq:lem55}
\end{align}
for some constants $C_1,C_2 > 0$ independent of $s$ and $n$.

Plugging then \eqref{eq:lem54} into \eqref{eq:lem53} yields
\begin{align}
    \E  & \Big[ -2\log L_{M^s_{n-T_{n,2^d}'}} ( \widehat{\bp}_n^s(\bT_{n, \{1, \ldots, r \}}') \, \vert \, \bT_{n, \{1, \ldots, r \}}') \Big| \td  \Big] \nonumber\\
    &\le -2 \log \big( (n - T_{n,2^d}') ! \big) + 2(n - \td) \log \Big( \frac{n - \td}{k_n} \Big) +2 \frac{n - \td}{k_n} \log(k_n!)\nonumber\\
    & \quad -2 \frac{n - \td}{k_n} \Bigg\{ \log(k_n!) -\sum_{j=1}^r \E [ \log( T_{n,j}(k_n)!)]  + \sum_{j=1}^s  \E \Big[ T_{n,j}(k_n) \log \big( \widehat p_{n,j}^s  \big)   \Big]\nonumber\\
    & \quad  \quad \quad \quad  \quad \quad \quad  \quad + \sum_{j=s+1}^r \E \Big[ T_{n,j}(k_n) \log \big( \widehat \rho_{n}^s  \big)   \Big] \Bigg\} + C_1 \log(n - \td) \nonumber \\
    &= -2 \log \big( (n - T_{n,2^d}') ! \big) + 2(n - \td) \log \Big( \frac{n - \td}{k_n} \Big) +2 \frac{n - \td}{k_n} \log(k_n!)\nonumber\\
    & \quad -2 \frac{n - \td}{k_n} \E[ \log L_{M^s_{k_n}} ( \widehat{\bp}_n^s(\bT_n(k_n)) \, \vert \, \bT_n(k_n) ) ] + C_1 \log(n - \td), \nonumber 
\intertext{and using inequality \eqref{eq:lem55} gives then }
    &\leq -2 \log \big( (n - T_{n,2^d}') ! \big) + 2(n - \td)  ( \log(n - \td) - 1 )\nonumber\\
    & \quad -2 \frac{n - \td}{k_n} \E[ \log L_{M^s_{k_n}} ( \widehat{\bp}_n^s(\bT_n(k_n)) \, \vert \, \bT_n(k_n) ) ] + {C_3\log(n - \td)}. \nonumber 
\end{align}
Finally, Assumption (\ref{asu:BIC_global}\ref{E1}), the last upper bound  and Jensen's inequality  result in 
\begin{align*}
   \E   \Bigl[  -2 &\log  \E_{\lambda}  [ L_{M^s_{n-T_{n,2^d}'}} ( \widetilde{\bp}^s \, \vert \, \bT_{n, \{1, \ldots, r \}}') ]  \Bigr] \notag\\
    \leq \, &  \, -2 \E \biggl[  \log \big( (n - T_{n,2^d}') ! \big)  - (n -\td) \left( \log ( n- \td) -1 \right) \biggr] \\
    &- 2\E   \Bigl[   \frac{n-\td}{k_n} \E[ \log L_{M^s_{k_n}} ( \widehat{\bp}_n^s(\bT_n(k_n)) \, \vert \, \bT_n(k_n) ) ]      \biggr] \nonumber \\
    &+ 2  s  \E \Big[  \log \Big( \left(n - \td \right)   \sqrt{ \frac{r}{2 \pi (r-s)}} \Big)  \Big]   + C_3 \log(n q_n)  \\
    \le \, &  \, -2 \E \biggl[  \log \big( (n - T_{n,2^d}') ! \big)  - (n -\td) \left( \log ( n- \td) - 1 \right) \biggr]  \\
    &-2 \frac{n q_n}{k_n} \E[ \log L_{M^s_{k_n}} ( \widehat{\bp}_n^s(\bT_n(k_n)) \, \vert \, \bT_n(k_n) ) ]     + 2  s   \log \Big(  k_n  \sqrt{ \frac{r}{2 \pi (r-s)}} \Big)  \\
    &  + C \log(n q_n), 
\end{align*}
where $C > 0$ is a constant independent of $s$ and $n$.
\end{proof}

\subsection{Proof of Proposition \ref{ineq:LogBinAbsch}}

The target of this section is to prove \Cref{ineq:LogBinAbsch}.

\begin{lemma} \label{lem:Lemma1Bin}
Under Assumption (\ref{asu:BIC_global}\ref{asu:BIC_conv_nqn53})  we have for sufficiently large $n$ that 
\begin{align*}
  -2 \log \E_{\lambda}  & [ L_{\Bin_n} \! (1 - \widetilde q \, \vert \, T_{n,2^d}')] \\
  &\leq  -2 \log L_{\Bin_n} \! (1 - \widehat{q}_n \, \vert \, T_{n,2^d}')  -2 \log(2 \pi) + \log(n / \widehat{q}_n)+ o_\P(1),
\end{align*}
where   $\widehat{q}_n \coloneqq (n - \td) / n$  is an estimator for $q_n$. The expectation of the $o_\P(1)$ term is of order $o(1).$
\end{lemma}
The proof of the lemma is analog to the proof of \Cref{th:BIC_post_prob} by taking the uniform distribution on $(0,1)$ as the prior density, and is therefore omitted.

\begin{namedthm*}{\Cref{ineq:LogBinAbsch}}\label{ineq:LogBinAbsch_supp}
Suppose Assumptions (\ref{asu:BIC_global}\ref{asu:BIC_conv_nqn53})  and   (\ref{asu:Model_global}\ref{asu:BIC_glob_qn_kn}) hold. The expectation of the  binomial likelihood satisfies as $n\to\infty$ the inequality
\begin{align*} 
   -2 \E[  \log \E_{\lambda}  [ L_{\Bin_n} (1 - \widetilde q \, \vert \, T_{n,2^d}')]  ] 
   \leq& -2  \log(n!) +2 \E[\log( ( n - \td)!)]  +2 \E[ \log( \td!)]  \\
   &-2 n q_n \log( k_n/n)   +2 \log( n)  + C n q_n, 
\end{align*}
for a constant $C > 0$ independent of $s$ and $n$.
\end{namedthm*}
\begin{proof}
Without loss of generality, we assume in the following that the constant $C > 0$, which is independent of $s$ and $n$, is chosen sufficiently large such that the following inequalities hold.

Under  Assumption (\ref{asu:Model_global}\ref{asu:BIC_glob_qn_kn}), we are allowed to use the second equation on page 31 in the proof of  Lemma 6 in \citet{meyer_muscle23} 
\begin{align*}
    \E[ & \log L_{\Bin_n} (1 - q_n \, \vert \, T_{n,2^d}')]  \\
    &=  \E[ \log L_{\Bin_n} (1 - \frac{k_n}{n} \, \vert \, T_{n,2^d}')]   +  \left( \frac{k_n}{n} - q_n\right) \left( \frac{n q_n}{k_n/n} - \frac{n (1 - q_n)}{1  - k_n/n} \right).
\end{align*}
A combination with the asymptotic expansion in the last equation on page 31 in the proof of  Lemma 6 in \citet{meyer_muscle23}
\begin{align*}
    \frac{n q_n}{k_n/n} - \frac{n (1 - q_n)}{1  - k_n/n} = \left( q_n - \frac{k_n}{n}  \right) \frac{n}{k_n/n} + O(k_n),
\end{align*}
 gives then 
\begin{align*}
    \E[ & \log L_{\Bin_n} (1 - q_n \, \vert \, T_{n,2^d}')] =  \E[ \log L_{\Bin_n} (1 - \frac{k_n}{n} \, \vert \, T_{n,2^d}')]   - \left(\frac{k_n}{n} - q_n \right)^2   \frac{n}{k_n/n} + O(k_n).
\end{align*}
By  Assumption (\ref{asu:Model_global}\ref{asu:BIC_glob_qn_kn}) follows the existence of a positive constant $C_1>0$ such that
\begin{equation*}
   \E[ \log L_{\Bin_n} (1 - q_n \, \vert \, T_{n,2^d}')] \ge \E[ \log L_{\Bin_n} (1 - \frac{k_n}{n} \, \vert \, T_{n,2^d}')]  -  \left(\frac{k_n}{n} - q_n \right)^2   \frac{n}{k_n/n} - C_1 n q_n.
\end{equation*}
Since $ nq_n\to\infty$  and for $\widehat{q}_n \coloneqq ( n- \td) / n$ we have  $$\E[ \log L_{\Bin_n} (1 - q_n \, \vert \, T_{n,2^d}')] - \E[ \log L_{\Bin_n} (1 - \widehat{q}_n \, \vert \, T_{n,2^d}')] \rightarrow 0,$$ as $\ninf$,  it follows the existence of a constant $C_2>0$ such that
\begin{equation*}
   \E[ \log L_{\Bin_n} (1 - \widehat{q}_n \, \vert \, T_{n,2^d}')] \ge \E[ \log L_{\Bin_n} (1 - \frac{k_n}{n} \, \vert \, T_{n,2^d}')]  -  \left(\frac{k_n}{n} - q_n \right)^2   \frac{n}{k_n/n} - C_2 n q_n.
\end{equation*}
A combination of \Cref{lem:Lemma1Bin} and the equation above gives the existence of a constant $C_3>0$ such that
\begin{align} 
   -2 &\E[  \log  \E_{\lambda}  [ L_{\Bin_n} (1 - \widetilde q \, \vert \, T_{n,2^d}')]  ] \notag\\
   &\leq -2 \E[ \log L_{\Bin_n} (1 - \widehat{q}_n \, \vert \, T_{n,2^d}')]  -2 \log(2 \pi) + \E \left[ \log \Big( \frac{n}{\widehat{q}_n} \Big) \right] +o(1) \notag \\
   &\le -2 \E[ \log L_{\Bin_ n} (1 - \frac{k_n}{n} \, \vert \, T_{n,2^d}')] +2 \Big(\frac{k_n}{n} - q_n \Big)^{\! 2} \!  \frac{n}{k_n/n} + \E \Big[ \log \Big( \frac{n}{\widehat{q}_n} \Big) \Big]  + C_3 n q_n 
  .\label{eq:LogBinCombined}
\end{align}
Inserting 
\begin{align*}
   &\E[    \log L_{\Bin_ n}(1 - \frac{k_n}{n} \, \vert \, T_{n,2^d}') ] \\
   &\;= \log(n!) - \E[\log( ( n - \td)!) - \log( \td!)] + n (1-q_n) \log \Big( 1- \frac{k_n}{n} \Big) + n q_n \log \Big( \frac{k_n}{n} \Big) 
\end{align*}
into \cref{eq:LogBinCombined} yields
\begin{align} 
   -2 \E[ &\log \E_{\lambda}  [ L_{\Bin_n} (1 - \widetilde q \, \vert \, T_{n,2^d}')]  ] \notag\\
   \leq& -2  \log(n!) +2 \E[\log( ( n - \td)!)]  +2 \E[ \log( \td!)] -2 n (1-q_n) \log \Big( 1- \frac{k_n}{n} \Big)\nonumber \\
   &-2 n q_n \log \Big( \frac{k_n}{n} \Big) +2 \left(\frac{k_n}{n} - q_n \right)^2   \frac{n}{\frac{k_n}{n}}   + \E \left[ \log \Big( \frac{n}{\widehat{q}_n} \Big) \right]+ C_3 n q_n . \label{E.2}
\end{align}
We have by Assumption (\ref{asu:Model_global}\ref{asu:BIC_glob_qn_kn}) that $\log(1 - k_n / n) \le C_4 q_n $ for some $C_4>0$ and thus,
\begin{align}
     -2 n &(1-q_n) \log( 1- k_n/n)  +2 \left(\frac{k_n}{n} - q_n \right)^2   \frac{n}{k_n/n}     \le C_5 n q_n \label{eq:BIC_Global_Vereinf1}
\end{align}
for some $C_5>0$.

Finally, we use for $B \sim \Bin(n, p_n)$ with $n p_n \to \infty$ a Taylor expansion  and the Chernoff inequality  resulting in the existence of a positive constant $C>0$ such that
\begin{align*}
\log(\E[B]) - C \le \E[\log(B) \mathbbm{1}\{B >0   \}] .  \end{align*}
But due to Assumption (\ref{asu:BIC_global}\ref{asu:BIC_conv_nqn53}) we know that $n \widehat{q}_n = n - \td \sim  \Bin(n, q_n)$ with $n q_n\to \infty$ such that 
 \begin{align}
    \E \left[ \log \Big( \frac{n}{\widehat{q}_n} \Big) \right] \le \log \Big( \frac{n}{q_n} \Big) + C_6 \leq   2 \log( n)    + C_6\label{eq:BIC_Global_Vereinf2}
 \end{align}
 for some constant $C_6>0$.
Hence, the statement follows from \Cref{E.2}-\Cref{eq:BIC_Global_Vereinf2}.
\end{proof}

\section{Additional simulation study}

We explore an additional simulation study for the max-mixture model of \citet{tawn}, which exhibits asymptotic dependence. 
For $\beta \in \Pd$ and $d=5$ suppose $\bF_\beta=(F_{\beta,j})_{j\in\beta} $ is a $|\beta|$-dimensional random vector with Fréchet(1) distributed margins and the following dependence structure.
First, $\bF_{\{1,2\}}, \bF_{\{4,5\}}$ have a bivariate Gaussian copula with correlation parameter $\rho = 0.25$. On the other hand, $\bF_{\{1,2,3\}}, \bF_{\{3,4,5\}}$ and $\bF_{\{1,2,3,4,5\}}$ have a three-dimensional and five-dimensional extreme value logistic copula, respectively, with dependence parameter $\vartheta$. 
Then the regular varying vector $\bX \in \R^5$ of index $-1$ is defined as
\begin{align*}
    \bX \coloneqq (X_1, \ldots, X_5)^\top \coloneqq 
    \begin{pmatrix} 
    \max \Big\{  \frac{5}{7} F_{\{1,2\} , 1 },\, \frac{1}{7} F_{ \{1,2, 3\} , 1}, \,  \frac{1}{7} F_{ \{1,2, 3,4,5\} , 1} \Big\} \\
    \max \Big\{  \frac{5}{7} F_{\{1,2\} , 2 },\, \frac{1}{7} F_{ \{1,2, 3\} , 2}, \,  \frac{1}{7} F_{ \{1,2, 3,4,5\} , 2} \Big\} \\
    \max \Big\{  \frac{3}{7} F_{\{1,2,3\} , 3 },\, \frac{3}{7} F_{ \{3,4,5\} , 3}, \,  \frac{1}{7} F_{ \{1,2, 3,4,5\} , 3} \Big\} \\
    \max \Big\{  \frac{5}{7} F_{\{4,5\} , 4 },\, \frac{1}{7} F_{ \{3,4,5\} , 4}, \,  \frac{1}{7} F_{ \{1,2, 3,4,5\} , 4} \Big\} \\
    \max \Big\{  \frac{5}{7} F_{\{4,5\} , 5 },\, \frac{1}{7} F_{ \{3,4,5\}  , 5}, \,  \frac{1}{7} F_{ \{1,2, 3,4,5\} , 5} \Big\} \\ 
    \end{pmatrix}.
\end{align*}
Since the Gaussian copula exhibits pairwise asymptotic independence, the random vector $\bTheta$  puts mass on the cones $C_{\{1\}}, \,C_{\{2\}}, \,C_{\{4\}}, \, C_{\{5\}}, \, C_{\{1,2,3\}},$ $\, C_{\{3,4,5\}}, \, C_{\{1,2,3,4,5\}}$ and by the choice of the scaling factors, each cone has the same probability. However, the distribution of $\bZ$ is not discrete and we need to estimate the support of $\bZ$ via a Monte-Carlo simulation, where we use the implementation of \citet{meyer_muscle23}.

The simulation results of this $5$-dimensional model with $s^*=7$ are presented in \Cref{fig:max_mix_fixed_k}. In this simulation study  the dependence parameter $\vartheta$ takes values $0.1$, $0.5$ and $0.9$ and the sample sizes is  $n = 1000$, $5000$, $10000$ and $20000$. As before, we conduct 500 repetitions. We report only the Hellinger distance, as the Accuracy error and $F_1$ error are not informative in this context. This is because, in the Monte Carlo simulation used to estimate the probabilities of the cones (which are not known explicitly), all $2^5-1 = 31$ possible cones were detected and thus classified as a relevant direction.
\begin{figure}
     \centering
     \begin{subfigure}[b]{0.49\textwidth}
         \centering
         \includegraphics[width=\textwidth]{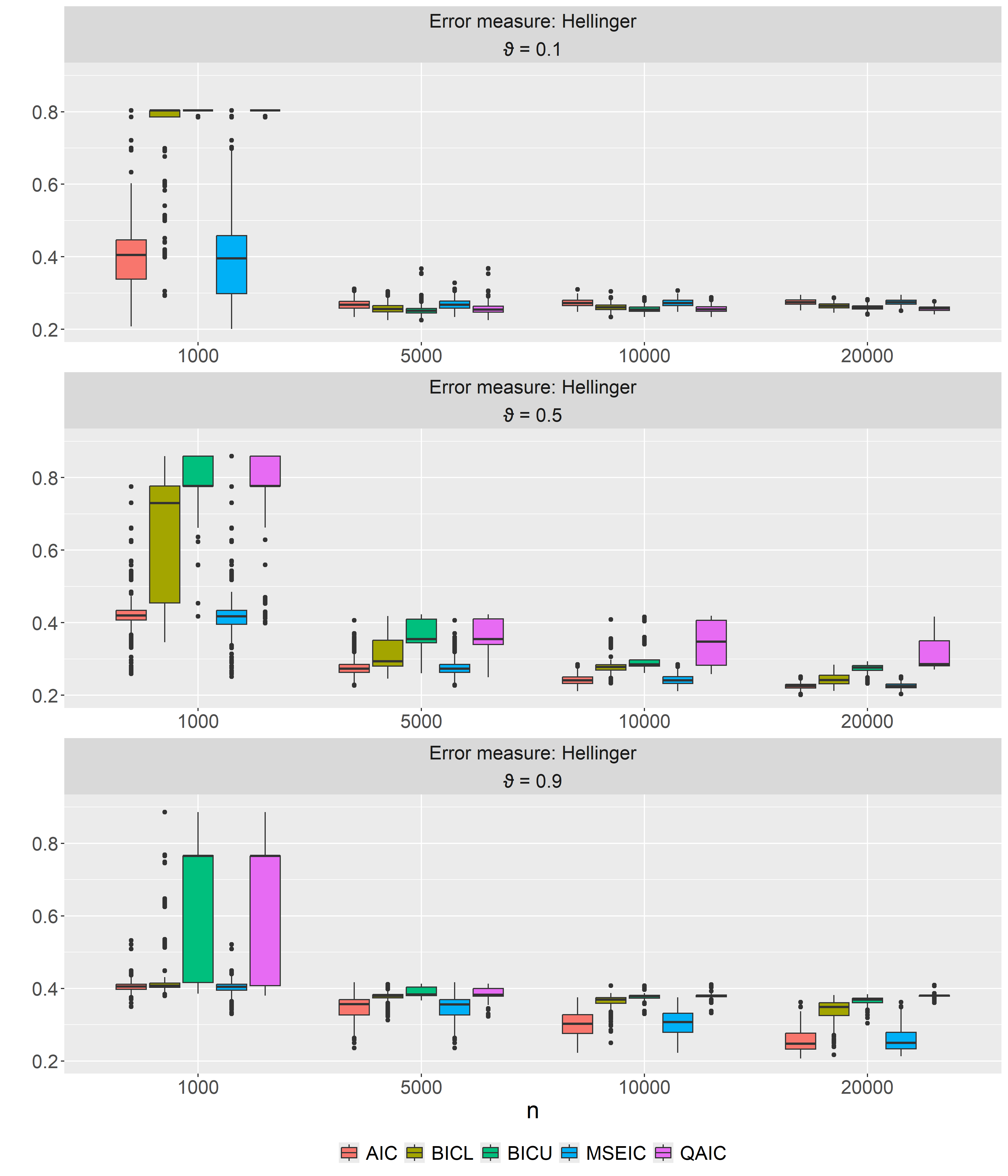}
         \caption{Local model with $k_n / n = 0.05$}
         %\label{fig:y equals x}
     \end{subfigure}
     \hfill
     \begin{subfigure}[b]{0.49\textwidth}
         \centering
         \includegraphics[width=\textwidth]{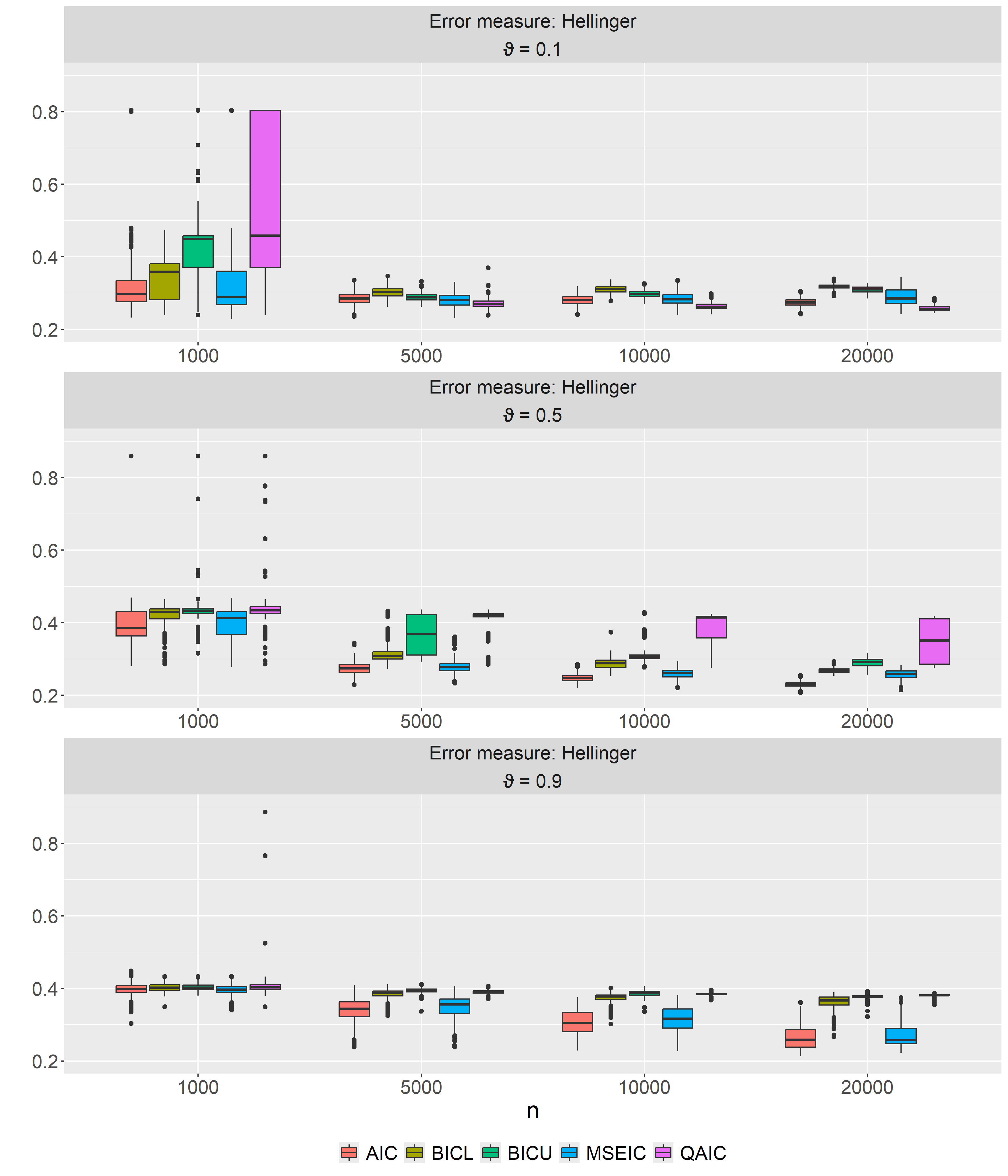}
         \caption{Global model}
         %\label{fig:three sin x}
     \end{subfigure}
         \caption{\footnotesize \textit{Simulations for the max-mixture model with $s^* = 5$ directions of extremes and $d=5$:  From top to the bottom, the dependence parameter increases from $\vartheta = 0.1$, $\vartheta = 0.5$ to $\vartheta = 0.9$. The Hellinger distance is plotted against the sample size $n$ on the $x$-axis. 
         }}
    \label{fig:max_mix_fixed_k}
\end{figure}
The figure shows similar patterns across all information criteria. In particular, as the sample size $n$ increases, the performance improves. The dependence parameter $\vartheta$ does not appear to have a strong impact on the information criteria. However, for $n = 1000$, the Hellinger distance tends to be smaller when $\vartheta$ is higher, suggesting a potential influence at smaller sample sizes.

\end{bibunit}
\end{document}